\documentclass[reqno]{article}

\usepackage[utf8]{inputenc}   
\usepackage[T1]{fontenc}      

\usepackage{amsmath,amsthm} 
\usepackage{amssymb,mathrsfs} 
\usepackage{amssymb}
\usepackage{amsfonts}
\usepackage{upgreek}
\usepackage{nicefrac}
\usepackage{geometry}
\usepackage{authblk}
\usepackage{graphicx}
\usepackage[hypertexnames=false]{hyperref}
\hypersetup{
  colorlinks = true,
  citecolor = blue,
}
 \usepackage{color}
 \usepackage{algorithm2e}
\usepackage{stmaryrd}

\usepackage{enumitem}
\usepackage{url}
\usepackage{graphicx} 
\usepackage{lmodern} 
\usepackage{xcolor}
\usepackage{bbm}

\usepackage{ifthen}
\usepackage{xargs}

\usepackage[disable]{todonotes}

\usepackage{aliascnt}
\usepackage{cleveref}
\usepackage{autonum}
\makeatletter
\newtheorem{theorem}{Theorem}
\crefname{theorem}{theorem}{Theorems}
\Crefname{Theorem}{Theorem}{Theorems}

\newtheorem*{lemma_nonumber*}{Lemma}

\newaliascnt{lemma}{theorem}
\newtheorem{lemma}[lemma]{Lemma}
\aliascntresetthe{lemma}
\crefname{lemma}{lemma}{lemmas}
\Crefname{Lemma}{Lemma}{Lemmas}

\newaliascnt{corollary}{theorem}
\newtheorem{corollary}[corollary]{Corollary}
\aliascntresetthe{corollary}
\crefname{corollary}{corollary}{corollaries}
\Crefname{Corollary}{Corollary}{Corollaries}

\newaliascnt{proposition}{theorem}
\newtheorem{proposition}[proposition]{Proposition}
\aliascntresetthe{proposition}
\crefname{proposition}{proposition}{propositions}
\Crefname{Proposition}{Proposition}{Propositions}

\newaliascnt{definition}{theorem}

\aliascntresetthe{definition}
\crefname{definition}{definition}{definitions}
\Crefname{Definition}{Definition}{Definitions}

\newaliascnt{remark}{theorem}
\newtheorem{remark}[remark]{Remark}
\aliascntresetthe{remark}
\crefname{remark}{remark}{remarks}
\Crefname{Remark}{Remark}{Remarks}

\crefname{example}{example}{examples}
\Crefname{Example}{Example}{Examples}

\crefname{figure}{figure}{figures}
\Crefname{Figure}{Figure}{Figures}

\newtheorem{assumption}{\textbf{H}\hspace{-3pt}}
\Crefname{assumption}{\textbf{H}\hspace{-3pt}}{\textbf{H}\hspace{-3pt}}
\crefname{assumption}{\textbf{H}}{\textbf{H}}

\newtheorem{assumptionA}{\textbf{A}\hspace{-3pt}}
\Crefname{assumptionA}{\textbf{A}\hspace{-3pt}}{\textbf{A}\hspace{-3pt}}
\crefname{assumptionA}{\textbf{A}}{\textbf{A}}

\Crefname{assumptionG}{\textbf{G}\hspace{-3pt}}{\textbf{G}\hspace{-3pt}}
\crefname{assumptionG}{\textbf{G}}{\textbf{G}}

\renewcommand{\bar}[1]{\overline{#1}}

\renewcommand{\theenumi}{(\alph{enumi})} 
\renewcommand{\theenumii}{\roman{enumii}}

\DeclareFontFamily{U}{matha}{\hyphenchar\font45}
\DeclareFontShape{U}{matha}{m}{n}{
      <5> <6> <7> <8> <9> <10> gen * matha
      <10.95> matha10 <12> <14.4> <17.28> <20.74> <24.88> matha12
      }{}
\DeclareSymbolFont{matha}{U}{matha}{m}{n}
\DeclareFontSubstitution{U}{matha}{m}{n}

\DeclareFontFamily{U}{mathx}{\hyphenchar\font45}
\DeclareFontShape{U}{mathx}{m}{n}{
      <5> <6> <7> <8> <9> <10>
      <10.95> <12> <14.4> <17.28> <20.74> <24.88>
      mathx10
      }{}
\DeclareSymbolFont{mathx}{U}{mathx}{m}{n}
\DeclareFontSubstitution{U}{mathx}{m}{n}

\DeclareMathDelimiter{\vvvert}{0}{matha}{"7E}{mathx}{"17}

\newcommand{\ulambda}{\underline{\lambda}}
\newcommand{\olambda}{\overline{\lambda}}

\newcommand{\kappavara}{\kappa_1}
\newcommand{\kappavarb}{\kappa_2}

\newcommand{\lambdaref}{\lambda_{\rm ref}}

\newcommand{\nax}{\nabla_x}
\newcommand{\divx}{\operatorname{div}_x}

\newcommand{\nav}{\nabla_v}
\newcommand{\Piv}{\Pi_v}

\newcommand{\measv}{\nu}

\usepackage{cancel}

\newcommand{\Lmu}{{\rmL^2(\mu)}}
\newcommand{\tLmu}{{\rmL_0^2(\mu)}}

\newcommand{\Lpi}{{\rmL^2(\pi)}}
\newcommand{\tLnu}{{\rmL_0^2(\nu)}}

\newcommand{\calBLmu}{{\calBLmu}}
\newcommand{\calBLnu}{{\calBLnu}}
\newcommand{\calBLkappa}{{\calBLkappa}}

\newcommand{\lang}{\left\langle}
\newcommand{\rang}{\right\rangle}

\newcommand{\half}{{\nicefrac{1}{2}}}

\newcommand{\ie}{\textit{i.e.}}

\newcommand{\calA}{\mathcal{A}}
\newcommand{\calB}{\mathcal{B}}

\newcommand{\calD}{\mathcal{D}}

\newcommand{\calL}{\mathcal{L}}

\newcommand{\calR}{\mathcal{R}}
\newcommand{\calS}{\mathcal{S}}
\newcommand{\calT}{\mathcal{T}}

\newcommand{\bbN}{\mathbb{N}}
\newcommand{\bbR}{\mathbb{R}}
\newcommand{\bbS}{\mathbb{S}}
\newcommand{\bbT}{\mathbb{T}}
\newcommand{\bbZ}{\mathbb{Z}}

\newcommand{\bbE}{\mathbb{E}}

\newcommand{\rme}{\mathrm{e}}

\newcommand{\rmH}{\mathrm{H}}
\newcommand{\rmL}{\mathrm{L}}

\newcommand{\bfG}{\mathbf{G}}

\newcommand{\bfM}{\mathbf{M}}

\newcommand{\bfone}{\mathbf{1}}

\newcommand{\Fscr}{\mathscr{F}}

\def\msa{\mathsf{A}}
\def\msd{\mathsf{X}}
\def\msdd{\mathsf{D}}

\def\msc{\mathsf{C}}
\def\mse{\mathsf{E}}
\def\msf{\mathsf{F}}

\def\msh{\mathsf{H}}
\def\msm{\mathsf{M}}

\def\msv{\mathsf{V}}

\def\msx{\mathsf{X}}

\def\mca{\mathcal{A}}
\def\bmca{\bar{\mathcal{A}}}
\def\mcd{\mathcal{D}}
\def\mcl{\mathcal{L}}
\def\mcs{\mathcal{S}}
\def\mct{\mathcal{T}}
\def\tmct{\tilde{\mathcal{T}}}
\def\tmcs{\tilde{\mathcal{S}}}

\def\mcab{\bar{\mathcal{A}}}
\def\mcbb{\mathcal{B}}  
\newcommand{\mcb}[1]{\mathcal{B}(#1)}
\def\mcc{\mathcal{C}}

\def\mcx{\mathcal{X}}
\def\mce{\mathcal{E}}
\def\mcf{\mathcal{F}}

\def\mcv{\mathcal{V}}
\def\mcr{\mathcal{R}}
\def\tmcr{\tilde{\mathcal{R}}}

\def\rset{\mathbb{R}}

\def\nset{\mathbb{N}}

\def\Rset{\mathbb{R}}

\def\mrm{\mathrm{m}}
\def\rmd{\mathrm{d}}

\def\mrl{\mathrm{L}}

\def\rmH{\mathrm{H}}
\def\rml{\mathrm{L}}

\def\rme{\mathrm{e}}
\def\rmn{\mathrm{n}}

\def\mrC{\mathrm{C}}
\def\mrCb{\mathrm{C}_{\operatorname{b}}}
\def\rmCb{\mathrm{C}_{\operatorname{b}}}

\def\rmc{\mathrm{c}}
\def\rmC{\mathrm{C}}

\newcommand{\abs}[1]{\left\vert #1 \right\vert}
\newcommand{\absLigne}[1]{\vert #1 \vert}
\newcommand{\tvnorm}[1]{\| #1 \|_{\mathrm{TV}}}

\newcommandx{\Vnorm}[2][1=V]{\| #2 \|_{#1}}
\newcommandx{\normpi}[2][2=\mrl^2(\pi)]{\left\Vert  #1 \right\Vert_{#2}}
\newcommandx{\normH}[2][2=2]{\left\Vert  #1 \right\Vert}
\newcommandx{\normHLigne}[2][2=2]{\Vert  #1 \Vert}
\newcommandx{\normHLine}[2][2=2]{\Vert  #1 \Vert}
\newcommandx{\normmu}[2][2=2]{\left\Vert  #1 \right\Vert_{#2}}
\newcommandx{\normopmu}[2][2=2]{\left\vvvert  #1 \right\vvvert_{#2}}
\newcommandx{\normoppi}[2][2=\rml^2(\pi)]{\left\vvvert  #1 \right\vvvert_{#2}}
\newcommandx{\normopH}[2][2=2]{\left\vvvert  #1 \right\vvvert}
\newcommandx{\normop}[2][2=2]{\left\vvvert  #1 \right\vvvert}
\newcommand{\ps}[2]{#1^{\top}#2}
\newcommand{\pspi}[2]{\left\langle#1,#2 \right\rangle_2}
\newcommand{\psH}[2]{\left\langle#1,#2 \right\rangle}
\newcommand{\psmu}[2]{\left\langle#1,#2 \right\rangle_2}
\newcommand{\psmuLigne}[2]{\langle#1,#2 \rangle_2}
\newcommandx{\normpiLine}[2][2=2]{\Vert  #1 \Vert_{#2}}
\newcommandx{\normpiLigne}[2][2=2]{\Vert  #1 \Vert_{#2}}
\newcommandx{\normmuLine}[2][2=2]{\Vert  #1 \Vert_{#2}}
\newcommandx{\normmuLigne}[2][2=2]{\Vert  #1 \Vert_{#2}}
\newcommandx{\normopmuLine}[2][2=2]{\vvvert  #1 \vvvert_{#2}}
\newcommandx{\normopmuL}[2][2=2]{\normopmuLine{#1}[#2]}
\newcommandx{\normopHLine}[2][2=2]{\vvvert  #1 \vvvert}
\newcommandx{\normopLine}[2][2=2]{\vvvert  #1 \vvvert}
\newcommand{\pspiLine}[2]{\langle#1,#2 \rangle_2}

\newcommandx{\VnormEq}[2][1=V]{\left\| #2 \right\|_{#1}}
\newcommandx{\norm}[2][1=]{\ifthenelse{\equal{#1}{}}{\left\Vert #2 \right\Vert}{\left\Vert #2 \right\Vert^{#1}}}
\newcommandx{\normLigne}[2][1=]{\ifthenelse{\equal{#1}{}}{\Vert #2 \Vert}{\Vert #2\Vert^{#1}}}

\newcommand{\parenthese}[1]{\left(#1 \right)}

\newcommand{\parentheseDeux}[1]{\left[ #1 \right]}
\newcommand{\defEns}[1]{\left\lbrace #1 \right\rbrace }

\newcommandx\probaMarkovTilde[2][2=]
{\ifthenelse{\equal{#2}{}}{{\widetilde{\mathbb{P}}_{#1}}}{\widetilde{\mathbb{P}}_{#1}\left[ #2\right]}}

\newcommand{\plusinfty}{+\infty}
\def\ie{\textit{i.e.}}

\def\eqsp{\;}
\renewcommand{\iint}[2]{\{ #1,\ldots,#2\}}
\newcommand{\coint}[1]{\left[#1\right)}

\newcommand{\ooint}[1]{\left(#1\right)}
\newcommand{\ccint}[1]{\left[#1\right]}

\newcommand{\oointLigne}[1]{(#1)}
\newcommand{\ccintLigne}[1]{[#1]}

\newcommand{\oointLine}[1]{(#1)}

\newcommandx\sequence[3][2=,3=]
{\ifthenelse{\equal{#3}{}}{\ensuremath{\{ #1_{#2}\}}}{\ensuremath{\{ #1_{#2}, \eqsp #2 \in #3 \}}}}
\newcommandx\sequenceD[3][2=,3=]
{\ifthenelse{\equal{#3}{}}{\ensuremath{\{ #1_{#2}\}}}{\ensuremath{( #1)_{ #2 \in #3} }}}

\newcommandx{\sequencen}[2][2=n\in\N]{\ensuremath{\{ #1_n, \eqsp #2 \}}}
\newcommandx\sequenceDouble[4][3=,4=]
{\ifthenelse{\equal{#3}{}}{\ensuremath{\{ (#1_{#3},#2_{#3}) \}}}{\ensuremath{\{  (#1_{#3},#2_{#3}), \eqsp #3 \in #4 \}}}}
\newcommandx{\sequencenDouble}[3][3=n\in\N]{\ensuremath{\{ (#1_{n},#2_{n}), \eqsp #3 \}}}

\def\iid{i.i.d.}

\def\eg{e.g.}
\def\Id{\operatorname{Id}}
\def\Idd{\operatorname{I}_d}
\newcommand{\ensemble}[2]{\left\{#1\,:\eqsp #2\right\}}
\newcommand{\ensembleLigne}[2]{\{#1\,:\eqsp #2\}}

\def\Tr{\mathrm{Tr}}
\def\generator{\calL}
\newcommand{\core}{\rmc}

\def\bfe{\mathbf{e}}

\def\domain{\mathrm{D}}
\def\card{\operatorname{card}}
\def\bigone{\operatorname{1}}
\def\bigO{\mathcal{O}}
\def\Cp{C_{\operatorname{P}}}
\def\core{\mathsf{C}}
\def\entropyA{\mathscr{H}}
\def\range{\operatorname{Ran}}
\def\kernel{\operatorname{Ker}}
\def\spec{\operatorname{Spec}}

\def\ff{f}
\def\gg{g}

\newcommand\overlineb[1]{{#1}}
\def\poly{\mathrm{poly}}
\def\veps{\varepsilon}
\def\trace{\operatorname{Tr}}
\def\sign{\operatorname{sgn}}
\def\diag{\operatorname{diag}}

\def\rateConv{\alpha}

\def\sphere{\mathbb{S}}
\newcommand{\1}{\mathbbm{1}}

\newcommand{\vphi}{\varphi}
\newcommand{\even}{\mathrm{e}}

\def\restriction#1#2{\mathchoice
              {\setbox1\hbox{${\displaystyle #1}_{\scriptstyle #2}$}
              \restrictionaux{#1}{#2}}
              {\setbox1\hbox{${\textstyle #1}_{\scriptstyle #2}$}
              \restrictionaux{#1}{#2}}
              {\setbox1\hbox{${\scriptstyle #1}_{\scriptscriptstyle #2}$}
              \restrictionaux{#1}{#2}}
              {\setbox1\hbox{${\scriptscriptstyle #1}_{\scriptscriptstyle #2}$}
              \restrictionaux{#1}{#2}}}
\def\restrictionaux#1#2{{#1\,\smash{\vrule height .8\ht1 depth .85\dp1}}_{\,#2}}

\newcommand\blfootnote[1]{%
  \begingroup
  \renewcommand\thefootnote{}\footnote{#1}%
  \addtocounter{footnote}{-1}%
  \endgroup
}

\title{Hypocoercivity of Piecewise Deterministic Markov Process-Monte Carlo}

\author[1]{Christophe Andrieu}
\author[2]{Alain Durmus}
\author[3]{Nikolas N\"usken}
\author[4]{Julien Roussel}

\affil[1]{School of Mathematics, University of Bristol, UK.}
\affil[2]{CMLA - \'Ecole normale supérieure Paris-Saclay, CNRS, Université Paris-Saclay, 94235 Cachan, France.}
\affil[3]{Imperial College London, UK.}
\affil[4]{\'Ecole des ponts ParisTech and INRIA, Paris, France.}

\begin{document}

\maketitle

\blfootnote{$^{1}$c.andrieu@bristol.ac.uk; $^{2}$alain.durmus@cmla.ens-cachan.fr; $^{3}$nik.nuesken@gmx.de;  $^{4}$julien.roussel@enpc.fr}
\begin{abstract}
    {In this work, we establish  $\rml^2$-exponential convergence for a broad class of Piecewise Deterministic Markov Processes recently proposed in the context of Markov Process Monte Carlo methods and covering in particular the Randomized Hamiltonian Monte Carlo \cite{Dolbeault15,BouRabee17}, the Zig-Zag process \cite{Bierkens16}  and the Bouncy Particle Sampler  \cite{peters2012rejection,BouchardCote15}. The kernel of the symmetric part of the generator of such processes is non-trivial, and we follow the ideas recently introduced in \cite{Dolbeault09,Dolbeault15} to develop a rigorous framework for hypocoercivity in a fairly general and unifying set-up, while deriving tractable estimates of the constants involved in terms of the parameters of the dynamics. As a by-product we characterize the scaling properties of these algorithms with respect to the dimension of classes of problems, therefore providing some theoretical evidence to support their practical relevance.}
\end{abstract}


\section{Introduction}


Consider a probability distribution $\pi$ defined on the Borel $\sigma$-field $\mcx$ of some domain $\msx = \bbR^d$ or $\msx = \bbT^d$ where $\bbT = \bbR / \bbZ$. Assume that $\pi$ has a density with respect to the Lebesgue measure also denoted $\pi$ and of the form
$	\pi = \rme^{-U} / \int_{\msx} \rme^{-U(y)} \rmd y $
where $U\colon\msx\to\bbR$ is a continuously differentiable function and is  referred to as the potential associated with $\pi$. Sampling from such distributions is of interest in computational statistical mechanics and in Bayesian statistics and allows one, for example, to compute efficiently expectations of functions $\ff : \msx \to \rset$ with respect to $\pi$ by invoking empirical process limit theorems, e.g. the law of large numbers. In practical set-ups, sampling exactly from $\pi$ directly is either impossible or computationally prohibitive. A standard and versatile approach to sampling from such distributions  consists of using Markov Chain Monte Carlo (MCMC) techniques \cite{GelmanCaStDuVeRu2014,liu2008monte,robert2013monte}, where the ability of simulating realizations of ergodic Markov chains leaving $\pi$ invariant is exploited. Markov Process Monte Carlo (MPMC) methods are the continuous time counterparts of MCMC but their exact implementation is most often impossible on computers and requires additional approximation, such as time discretization of the process in the case of the Langevin diffusion. A notable exception, which has recently attracted significant attention, is the class of MPMC relying on Piecewise Deterministic Markov Processes (PDMP) \cite{Davis1984,davis:1993}, which in addition to being simpler to simulate than earlier MPMC, are nonreversible, offering the promise of better performance. We now briefly introduce a class of processes covering existing algorithms. The generic mathematical notation we use in the introduction is fairly standard and fully defined at the end of the section.

Known PDMP Monte Carlo methods rely on the use of the auxiliary variable trick, that is the introduction of an instrumental variable and probability distribution $\mu$ defined on an extended domain, of which $\pi$ is a marginal distribution, which may facilitate simulation. In the present set-up, one introduces the velocity variable $v \in \msv \subset \bbR^d$ associated with a probability distribution $\nu$ defined on the $\sigma$-field $\mcv$ of $\msv$, where the subset $\msv$ is assumed to be closed. Standard choices for $\nu$ include the centered normal distribution with covariance matrix $m_2 \Idd$, where $\Idd$ is the $d$-dimensional identity matrix, the uniform distribution on the unit sphere $\bbS^{d-1}$, or the uniform distribution on $\msv=\{-1, 1 \}^d$. Let $\mse = \msx \times \msv$ and define the probability measure $\mu = \pi \otimes \measv$. The aim is now to sample from the probability distribution $\mu$.

We denote by $\rmCb^2(\mse)$ the set of bounded functions of $\rmC^2(\mse)$. The PDMP Monte Carlo algorithms we are aware of fall in a class of processes associated with generators of the form, for $\ff \in \rmCb^2(\mse)$ and $(x,v) \in \mse$,
\begin{align}
  \label{eq:generator pdmp}
& \qquad \calL_1 \ff(x,v) \\
& \qquad = v^\top \nax \ff(x,v) + \sum_{k=1}^K \lambda_k(x,v) \left( \mcbb_k - \Id \right)\ff(x,v) +m_2^{\half} \lambdaref(x)\mcr_v \ff(x,v)\eqsp,
\end{align}
where $K\in \nset$,  $\lambda_k : \mse \to \bbR_+$ for $k \in \{1,\ldots, K\}$, $\lambdaref : \msx \to \bbR_+$, $(\calR_v,\domain(\calR_v))$ and $(\mcbb_k,\domain(\mcbb_k))$ for $k \in \{1,\ldots, K\}$ are operators we specify below, and for $i \in  \{1,\ldots, d\}$ we assume

\begin{equation}
 \label{eq:def_m_2}
 	m_2 = \int_{\msv} v_i^2 \ \rmd \measv (v) \eqsp,
 \end{equation}
 which is assumed to be finite.  For any $k \in \{1,\ldots, K\}$, $\lambda_k$ will be referred to as a jump rate and $\lambdaref$ as the refreshment rate.

 In the case where $\msv = \rset^d$ and $\nu$ is the zero-mean Gaussian distribution on $\rset^d$ with covariance matrix $m_2 \Idd$, we also consider generators of the form, for any $\ff \in \rmCb^2(\mse)$ and $(x,v) \in \mse$,
 \begin{equation}
 \label{eq:generator_r_d_2}
 \mcl_2 \ff(x,v) = \mcl_1 \ff(x,v) - m_2 F_0(x)^{\top} \nabla_v f(x,v) \eqsp,
 \end{equation}
 where $F_0 : \msx \to \rset^d$. 

 For any $k \in \{1,\ldots , K\}$, the jump operators $\mcbb_k$ we consider are associated with continuous vector fields $F_k : \msx \to \rset^d$ of the form, for any $\ff : \mse \to \rset$ and $(x,v) \in \mse$,
 \begin{equation}
    \label{eq:def-bounces}
 \begin{aligned}
   \mcbb_k \ff(x, v) &= \ff \left( x,v - 2 \big(v^\top \rmn_k(x)\big) \, \rmn_k(x) \right) \eqsp, \\
   \rmn_k(x) &= \begin{cases}
	F_k(x)/\abs{F_k(x)} & \text{ if $F_k(x) \neq 0$} \eqsp,\\
	0 & \text{ otherwise \eqsp.} 
      \end{cases}
    \end{aligned}
  \end{equation}
These operators correspond to reflections of the velocity through the hyperplanes orthogonal to $F_k(X)$ at the event position $X$, \ie \ a flip of the component of the velocity in the direction given by $F_k$ inducing an elastic ``bounce'' of the position trajectory with the hyperplane.  As we shall see, the $K+1$ vector fields $F_k$ are tied to the potential $U$ by the relation $\nax U = \sum_{k=0}^K F_k$, required to ensure that $\mu$ is left invariant by the associated semi-group.  Informally, assuming for the moment that $\lambdaref=0$ and $F_0 = \nax U_0$ for some $U_0 \colon \msx \to \rset$, the corresponding process follows the solution of Hamilton's equations $(\dot{x}_t,\dot{v}_t) = \big(v_t, -\nax U_0(x_t)\big)$ for a random time of distribution governed by an inhomogeneous Poisson process with rate $(x,v) \mapsto \sum_{k=1}^K \lambda_k(x,v)$. When an event occurs and the current state of the process is $(X,V)$, one chooses between the $K$ possible updates of the state available, with probability proportional to $\lambda_1(X,V),\ldots,\lambda_K(X,V)$, with the particularity here that the position $X$ is left unchanged.

The vector fields $\{F_k : \msx \to \rset^d \, ; \, k \in \{1,\ldots,K\} \}$ and jump rates $\{\lambda_k : \mse \to \rset_+ \, ; \, k \in \{1,\ldots,K\} \}$ are linked by the  relations $\lambda_k(x,v) -\lambda_k(x,-v)=v^\top F_k(x)$ for $k \in \{1,\ldots, K\}$ and $(x,v) \in \mse$, together with other conditions,  required to ensure that $\mu$ is an invariant distribution of the associated semi-group. A standard choice, sometimes referred to as canonical, consists of choosing jump rates $\lambda_k(x,v)  =[\ps{v}{F_k(x)}]_+$  for $k \in \{1,\ldots, K\}$  and $(x,v) \in \mse$.

Denote by $\rmL^2(\mu)$ the set of measurable functions $\gg : \mse \to \rset$ such that $\int_{\mse} \gg^2 \, \rmd \mu < \plusinfty$. We let $\normpiLine{\cdot}$ be the norm induced by the scalar product
\begin{equation} \label{eq:def-L2-inner-product}
\text{for all }	 \ff, \gg \in \Lmu \eqsp, \quad \pspi{ \ff}{ \gg} = \int_\mse \ff \, \gg \, \rmd \mu \eqsp,
\end{equation}
making $\rmL^2(\mu)$ a Hilbert space. 

 The operator  $\calR_v$ will be referred to  as the refreshment operator, a standard
 example of which is $\calR_v = \Piv - \Id$ where $\Piv$ is the following
 orthogonal projector in $\Lmu$: for any $f \in \Lmu$, 
\begin{equation}
\label{eq:def_piv}
 \Piv \ff(x,v) = \int_\msv \ff(x, w) \, \rmd \measv(w) \eqsp,
\end{equation}
in which case the velocity is drawn afresh from the marginal invariant distribution, while the position is left unchanged. In this scenario the informal description of the process given above carries on with $\lambdaref \neq 0$ added to the rate $(x,v) \mapsto \sum_{k=1}^K \lambda_k(x,v)$, $\Piv$ an additional possible update to the velocity chosen with probability proportional to $\lambdaref$. Another possible choice is the generator of an Ornstein-Uhlenbeck operator leaving $\nu$ invariant.

 In all the paper we assume the following condition to hold for either $\mcl_1$ or $\mcl_2$, a condition satisfied by the examples covered in this manuscript.

 \begin{assumptionA}
 \label{as:generator}
 \begin{enumerate}
 \item  The operator $\mcl$ is closed in $\mrl^2(\mu)$, generates a strongly continuous contraction semi-group $(P_t)_{t \geq 0}$ on $\mrl^2(\mu)$, \ie~$P_0 = \Id$, for any $t,s \in \rset_+$, $P_{s+t} = P_sP_t$, for any $f \in \Lmu$, $\normpiLine{P_tf} \leq \normpiLine{f}$ and   $\lim_{t \to 0} \normpiLine{P_t f -f} = 0$.
    \item $\mu$ is a a stationary measure for $(P_t)_{t \geq 0}$, \ie~for any $t \in \rset_+$, $\mu P_t = \mu$.
    \item There exists a core $\msc$ for $\mcl$ such that $\msc$ is dense in $\mrl^2(\mu)$ and $\msc \subset \domain(\mcl) \cap \domain(\mcl^{\star})$, where $(\generator^{\star},\domain(\generator^{\star}))$ is the
adjoint of $\generator$  on $\mrl^2(\mu)$.
 \end{enumerate}
\end{assumptionA}

Note that if $\mcl$ generates a strongly continuous contraction semi-group then $\domain(\generator)$ is dense by
\cite[Theorem 2.12]{ethier:kurtz:1986} and the
adjoint of $\generator$  on $\mrl^2(\mu)$ is therefore well-defined and closed by \cite[Theorem 5.1.5]{pedersen1995analysis}, and $\domain(\generator^{\star})$ is dense. { Establishing that an operator $\mcl$ generates a continuous contraction semigroup is well known to be difficult in general, although we note recent progress in this direction in \cite{grothaus2020hypocoercivity}. However as discussed in Section \ref{sec:DMSforPDMP}, concerned with the application of our abstract results to PDMPs, operators such as defined in \eqref{eq:generator pdmp} and \eqref{eq:generator_r_d_2} can be shown to arise from well defined processes. Indeed \cite{Davis1984} establishes the existence of PDMP processes and identifies the extended generator solving the associated Martingale problem. Building on this earlier work \cite{durmus-invariant-2018arXiv180705421D} have recently developed a general framework to characterize the strong generator of a broad class of PDMPs for which $\mrC_{b}^{k}(\mse)$, the set of real valued functions with up to order $k \in \nset$ bounded differentials defined on a Riemanian manifold $\mse$, can be shown to be a core. }

We now describe how various choices of $K$ and $F_k$ lead to known algorithms. For simplicity of exposition, we assume for the moment that $\msv=\bbR^d$, $\nu$ is the zero-mean Gaussian distribution with covariance matrix $m_2 \Idd$ and $\calR_v=\Piv-\Id$, but as we shall see later our results cover more general scenarios. 
\begin{itemize}
\item The particular choice $K=0$ and $F_0 = \nax U$ corresponds to the procedure described in \cite{Duane87} as a motivation for the popular hybrid Monte Carlo method. This process is also known as the Linear Boltzman/kinetic equation in the statistical physics literature \cite{bhatnagar:gross:krook:1954} or randomized Hamiltonian Monte Carlo \cite{BouRabee17}.
  In this scenario the process follows the isocontours of $\mu$ for random times distributed according to an inhomogeneous Poisson law of parameter $\lambdaref>0$, triggering events where the velocity is sampled afresh from $\nu$.
\item The scenario where $K=d$, $F_0 = 0$ and for $k \in \iint{1}{d}$, $x \in \msx, \ F_k(x) = \partial_k U(x) \bfe_k$ where $(\bfe_k)_{k \in \iint{1}{d}}$ is the canonical basis, corresponds to the Zig-Zag (ZZ) process \cite{Bierkens16}, where the $x$ component of the process follows straight lines in the direction $v$ which remains constant between events. In this scenario, the choice of $\mcbb_k$ to update the velocity, consists of negating its $k$-th component; see also \cite{Faggionato2009} for related ideas motivated by other applications.
\item The standard Bouncy Particle Sampler (BPS) of \cite{peters2012rejection}, extended by \cite{BouchardCote15}, correspond to the choice $K=1$, $F_0 = 0$ and $F_1 = \nax U$. 

\item More elaborate versions of the ZZ and BPS processes, motivated by computational considerations, take advantage of the possibility to decompose the energy as $U=\sum_{k=0}^K U_k$ and corresponds to the choice $F_k = \nax U_k$ \cite{Michel14,BouchardCote15}, where in the former the sign flip operation is replaced with a component swap. 
\item It should be clear that one can consider more general deterministic dynamics with $F_0\neq0$, effectively covering the Hamiltonian Bouncy Particle Sampler, suggested in \cite{Vanetti17}.
\item We remark that the well-known Langevin algorithm corresponds to $K=0$, $F_0 = \nax U$ and the situation where $\calR_v$ is the Ornstein-Uhlenbeck process. 
\end{itemize}
More general bounces involving randomization (see~\cite{Vanetti17,Wu17,Michel17}) can also be considered in our framework, at the cost of additional complexity and reduced tightness of our bounds.

The main aim of the present paper is the study of the long time behaviour  for the class of processes described above using hypercoercivity methods popularized by  \cite{Villani09}. More precisely, consider $(P_t)_{t \ge 0}$ the semigroup associated to the PDMP with generator $\calL \in \{\mcl_1,\mcl_2\}$ defined above, we aim to find simple and verifiable conditions on $U,F_k,\calR_v$ and $\lambdaref$ ensuring the existence of $A\ge 1$ and $\alpha > 0$, and their explicit computation in terms of characteristics of the data of the problem, such that for any $\ff \in \rml_0^2(\mu)=\ensemble{ \gg \in \Lmu}{\ \int_{\mse} \gg \, \rmd \mu = 0 }$ and $t \geq 0$, 
\begin{equation}
\label{eq:decay estimate}
\normmu{ P_t \ff } \le A \rme^{-\alpha t} \normmu{ \ff } \eqsp .
\end{equation}

Establishing such a result is of interest to practitioners for multiple reasons. Explicit bounds may provide insights into expected performance properties of the algorithm in various situations or regimes. In particular the above leads to an upper bound on the integrated autocorrelation, which is a performance measure of  Monte Carlo estimators of $\int_{\mse} \ff \, \rmd \mu$, $f \in \mrl^2_0(\mu)$, defined by 
\begin{equation}
\lim_{T\to\infty}  T \, {\rm Var}_{\mu}\left(T^{-1} \int_0^T \ff(X_t,V_t) \, \rmd t \right) / \normmu{ \ff }^2  \le \left. 2A  \middle/ \alpha \right.  \eqsp,
\end{equation}
where $(X_t,V_t)_{t \geq 0}$ is a trajectory of a PDMP process of generator $\mcl$ with $(X_0,V_0)$ distributed according to $\mu$. 
 For a class of problems of, say, increasing dimension $d\to\infty$, weak dependence of $A$ and $\alpha$ on $d$ indicates scalability of the method. It is worth pointing out that the result above is equivalent to the existence of $A\ge 1$ and $\alpha > 0$ such that for any measure $\rho_0\ll\mu$ such that $\normpiLine{\rmd\rho_0/\rmd\mu}<\infty$
 \begin{multline}
   \label{eq:5}
   \tvnorm{\rho_0 P_{t}-\mu} = \int_{\mse} \abs{\rmd(\rho_0 P_{t})/\rmd\mu-\bigone} \rmd \mu \leq \normpi{\rmd(\rho_0 P_{t})/\rmd\mu-\bigone}\\
   \leq A \rme^{-\alpha t} \normpi{\rmd\rho_0/\rmd\mu-\bigone} \eqsp,
 \end{multline}
where for $t \geq 0$ the probability measure $\rho_0 P_{t}$ on $\mse$ is such that for $(x,v)\in E$ and any measurable function $f$ such that the integrals exists, $\rho_0 P_{t}f(x,v)=\int_\mse P_tf(y,w)\rmd \rho_0(y,w)$ and  the leftmost inequality is standard and a consequence of the Cauchy-Schwarz inequality. Our hypocoercivity result therefore also allows characterization of convergence to equilibrium of PDMPs in various scenarios and regimes, leading in particular to the possibility to compare performance of algorithms started from the same initial distribution. Establishing similar results for different metrics may be a useful complement to our characterization of algorithmic computational complexity and is left for future work.

In \cite{mouhot:nuemann:2006,Villani09}, 
convergence of the type \eqref{eq:decay estimate} is established using an appropriate $\rmH^1$-norm associated with $\mu$. 
The method which was developed in these papers is closely related to hypoellipticity theory \cite{hormander:1967,eckman:hairer:2003,herau:nier:2004} for Partial Differential Equation and in particular the kinetic Fokker-Planck equation. Convergence for linear Boltzman equations was first derived in \cite{herau:2006,mouhot:nuemann:2006}. Since then, several works have extended and completed these results \cite{Dolbeault15,hankwan:leautaud:2015,achleitner:arnold:carlen:2016,bouin:hoffman:mouhot:2017,evans2017hypocoercivity,monmarche2017note}.

\subsection*{Notation and conventions}
Denote by $(\bfe_i)_{i \in \iint{1}{d}}$ the canonical basis of $\rset^d$ and $\Idd$ the $d$-dimensional identity matrix. The Euclidean norm on $\rset^d$ or $\rset^{d \times d}$ is denoted by $\absLigne{\cdot}$, and is associated with the usual Frobenius inner product ${\rm Tr}(\Phi^\top \Gamma)$ for any $\Phi,\Gamma$ in $\rset^d$ or $\rset^{d \times d}$.

Let $\msm$ be a smooth submanifold of $\rset^n$, for $n \in \nset$. 
For any $k \in \nset$, denote by $\mrC^k(\msm,\rset^m)$ the set of $k$-times differentiable functions from $\msm$ to $\rset   ^m$, $\mrCb^k(\msm,\rset^m)$ stands for the subset of bounded functions in $\mrC^k(\msm,\rset^m)$ with bounded differentials up to order $k$. $\mrC^k(\msm)$ and $\mrCb^k(\msm)$ stand for $\mrC^k(\msm,\rset)$ and $\rmCb^k(\msm,\rset)$ respectively.

For $f :\msx \to \rset$ and $i \in \iint{1}{d}$, $x \mapsto \partial_{x_i}f(x)$ stands for the partial derivative of $f$ with respect to the $i^{\text{th}}$-coordinate, if it exists. Similarly, for $f:\msx \to \rset$, $i,j \in \iint{1}{d}$, denote by $\partial_{x_i,x_j} f = \partial_{x_i} \partial_{x_j} f$ when $\partial_{x_i} \partial_{x_j} f$ exists.  For $f =(f_1,\ldots,f_m) \in \mrC^1(\msx,\rset^m)$, $\nax f$ stands for the gradient of $f$ defined for any $x \in \msx$ by $\nax f(x) = (\partial_{x_j} f_i(x))_{i\in \{1,\ldots,m\}, \, j \in \{1,\ldots,d\}} \in \rset^{d \times m}$. For ease of notation, we also denote by $(\nabla_x,\domain(\nabla_x))$ the densely defined closed extension of $(\nabla_x,\rmCb^1(\msx))$ on $\rml^2(\pi)$, see \cite[p. 88]{yoshida:1980}.
For any $f \in \rmC^k(\msx, \rset^m)$, $k \in \nset$ and $p \geq 0$, define $$\norm{f}_{k,p} = \sup_{x \in \msx} \, \sup_{(i_1,\ldots,i_k) \in \{1,\ldots,d\}^k} \defEns{\normLigne{\partial_{x_{i_1}, \ldots,x_{i_k}} f(x)}/(1+\norm{x}^p)} \eqsp.$$
We set for $k \geq 0$,
\begin{equation}
 \rmC_{\poly}^k(\msx, \rset^m) = \ensemble{f \in \rmC^k(\msx,\rset^m)}{\inf_{p \geq 0} \norm{f}_{k,p} < \plusinfty} \eqsp,
\end{equation}
and $ \rmC_{\poly}^k(\msx)$ simply stands for $ \rmC_{\poly}^k(\msx, \rset)$.
For any $f \in \rmC^2(\msx,\rset)$, we let $\Delta_x f$ denote the Laplacian of $f$. $\Id$ stands for the identity operator. For two self-adjoint operators $(\mca,\domain(\mca))$ and $(\mcbb,\domain(\mcbb))$ on a Hilbert space $\msh$ equipped with the scalar product $\psH{\cdot}{\cdot}$ and norm $\normH{\cdot}$, denote by $\mca \succeq \mcbb$ if $\lang \ff, \mca \ff \rang \geq \lang \ff, \mcbb \ff \rang$ for all $\ff \in \domain(\mca) \cap \domain(\mcbb)$. Then, define $(\mca \mcbb, \domain(\mca \mcbb))$ with domain, if not specified, $\domain(\mca\mcbb)= \domain(\mcbb) \cap \{\mcbb^{-1}\domain(\mca) \}$.
For a bounded operator $\mca$ on $\msh$, we let $\normopHLine{ \mca }= \sup_{\ff \in \msh, f \neq 0} \normHLine{\mca\ff}/\normHLine{\ff}$. $\Pi$ is said to be an orthogonal projection if $\Pi$ is a bounded symmetric operator $\msh$ and $\Pi^2=\Pi$. An unbounded operator $(\mca,\domain(\mca))$ is said to be symmetric (respectively anti-symmetric) is for any $f,g \in \domain(\mca)$, $\psH{\mca f }{g} = \psH{f}{\mca g}$ (respectively $\psH{\mca f }{g} = -\psH{f}{\mca g}$). If $\mca$ is densely defined, $\mca$ is said to be self-adjoint if $\mca = \mca^{\star}$. If in addition $\mca$ is closed, $\msc \subset \domain(\mca)$ is said to be a core for $\mca$ if the closure of $\restriction{\mca}{\msc}$ is $\mca$.  Denote by $\bigone_\msf$ the constant function equals to $1$ from a set $\msf$ to $\rset$. For any unbounded operator $(\mca,\domain(\mca))$, we denote by $\range(\mca) = \ensembleLigne{\mca f}{f \in \domain(\mca)}$ and $\kernel(\mca) = \ensembleLigne{f \in \domain(\mca)}{\mca f = 0}$. For any probability measure $\mathrm{m}$ on a measurable space $(\msm,\mathcal{F})$, we denote by $\mrl^2(\mrm)$ the Hilbert space of measurable functions $f$ satisfying $\int_{\msm} f^2 \rmd \mrm < \plusinfty$, equipped with the inner product $\langle f , g \rangle_{\mrm} = \int_{\msm} fg \, \rmd \mrm$,  and $\mrl^2_0(\mrm) = \{ f \in \mrl^2(\mrm) \, : \, \int_{\msm} f \rmd \mrm = 0\}$. We will use the same notation for vector and matrix fields $\Phi,\Gamma \in \left(\bbR^d\right)^\msm$ or $\left(\bbR^{d \times d}\right)^\msm$, \ie~$\langle \Phi ,  \Gamma \rangle_{\mrm} = \int_\msm \rm{Tr}\left(\Phi^\top \Gamma\right) \, \rmd \mrm$ and no confusion should be possible. When $\mrm=\mu$ we replace $\mrm$ with $2$ in this notation.
 For any $x \in \msm$ denote by $\updelta_{x}$ the Dirac distribution at $x$. 
We define the total variation distance between two probability measures $\mrm_1,\mrm_2$ on $(\msm,\mathcal{F})$ by $\tvnorm{\mrm_1-\mrm_2} = \sup_{\msa \in \mcf} \abs{\mrm_1(\msa)-\mrm_2(\msa)}$. For a square matrix $A$ we let ${\rm diag}(A)$ be its main diagonal and for a vector $v \in \rset^d$ we let ${\rm diag}(v)$ be the square matrix of diagonal $v$ and with zeros elsewhere. For $a,b \in \rset$ we let $a \wedge b$ denote their minimum.   For any $i,j \in \bbN$, $\delta_{i,j}$ denotes the Kronecker symbol which is $1$ if $i=j$ and $0$ otherwise. For any $n_1,n_2 \in \nset$, $n_1 < n_2$, we let $\sum_{n_2}^{n_1} = 0$. For any $x \in \rset$ we let $(x)_+=\max \{0, x\}$ be its positive part.


\section{Main results and organization of the paper}
\label{s:main results}

We now state our main results. In the following, for any densely defined operator $(\mcc,\domain(\mcc))$ we let $(\mcc^{\star},\domain(\mcc^{\star}))$ denote its $\Lmu$-adjoint. First we specify conditions imposed on the potential $U$.

\begin{assumption} \label{as:U} The potential $U \in \rmC^3_{\poly}(\msx)$ and satisfies 
\begin{enumerate}
\item \label{item:condition_hessian} there exists $c_1 \ge 0$ such that, for any $x \in \msx$,  $	\nax^2 U(x) \succeq - c_1 \Idd$;
\item \label{item:liminf} 
\begin{equation}
	\liminf_{|x| \to \infty} \left\{ | \nax U (x)|^2/2 - \Delta_x U(x) \right\} > 0 \eqsp.
\end{equation}
\end{enumerate}
\end{assumption}

From \cite{persson1960bounds,Bakry08}, \Cref{as:U}-\ref{item:liminf} is equivalent to assuming that $\pi$ satisfies a Poincar\'e inequality on $\msx$, that is the existence of $\Cp>0$ such that, for any $\ff \in \rmC^2(\msx)$ satisfying $\int_{\msx} \ff \rmd \pi = 0$,
\begin{equation} \label{eq:poincare_assumption}
\normmu{ \nax \ff }^2 \ge \Cp \normmu{ \ff }^2 \eqsp. 
\end{equation}
Further, \Cref{as:U}-\ref{item:liminf} also implies the existence of $c_2 > 0$ and $\varpi \geq 0$ such that for any $x \in \msx$,
\begin{equation}
 \label{eq:laplacian bound}
 \Delta_x U(x) \leq c_2 d^{1+\varpi}  + | \nax U(x) |^2/2 \eqsp.
\end{equation}
\Cref{as:U}-\ref{item:liminf} indeed implies that the quantity considered is bounded from below, the scaling in $d$ in front of $c_2$ will appear natural in the sequel. We have opted for this formulation of the assumption required of the potential to favour intuition and link it to the necessary and sufficient condition for geometric convergence of Langevin diffusions, but our quantitative bounds below will be given in terms of the Poincar\'e constant $\Cp$ for simplicity (see \cite[Section 4.2]{bakry:gentil:ledoux:2014} for quantitative estimates of $\Cp$ depending on potentially further conditions on $U$). \Cref{as:U}-\ref{item:condition_hessian} is realistic in most applications, can be checked in practice and has the advantage of leading to simplified developments. It is possible to replace this assumption with $\sup_{x \in \msd}\{\vert \nax^2U(x) \vert/ (1+\vert \nax U(x) \vert)\}<\infty$ and rephrase our results in terms of any finite upper bound of this quantity (see \cite[Sections~2 and 3]{Dolbeault15}).
Finally the Poincaré inequality \eqref{eq:poincare_assumption} implies by \cite[Proposition 4.4.2]{bakry:gentil:ledoux:2014} that there exists $s >0$ such that
\begin{equation}
 \label{eq:borne_moment}
 \int_{\rset^d}\rme^{s \abs{x}} \, \rmd \pi(x) < \plusinfty \eqsp.
\end{equation}

\begin{assumption}
\label{as:Fk} The family of vector fields $\{F_k : \msx \to \rset^d \, ; \, k \in \iint{0}{K}\}$ satisfies 
\begin{enumerate} 
\item for $k \in \iint{0}{K}$, $F_k \in \mrC^2(\msx,\rset^d)$;
\item \label{item:eq:decomposition U} for all $x \in \msx$, $	\nax U(x) = \sum_{k=0}^K F_k(x)$;
\item \label{item:Fk bounded} for all $k\in\iint{0}{K}$ there exists $a_k \ge 0$ such that for all $x \in \msx$,
 \begin{equation}
 \label{eq:def_a_k}
	|F_k|(x) \le a_k \left\{ 1 + |\nax U |(x) \right\} \eqsp.
\end{equation}
\end{enumerate}
\end{assumption}
This assumption is in particular trivially true for the Zig-Zag and the Bouncy Particle Samplers. In turn we assume the jump rates to be related to the family of vector fields $\{F_k : \msx \to \rset^d \, ; \, k \in \{1,\ldots,K\} \}$ through the following conditions.

  \begin{assumption} \label{as:intensities} There exist a  continuous  function $\varphi : \rset \to \rset_+$, $C_{\vphi} \geq 1$ and $c_{\vphi} \geq 0$ satisfying for any $s \in \rset$, 
    \begin{equation}
      \label{eq:prop_vphi}
      \varphi(s) - \varphi(-s) = s \eqsp, \qquad \text{and} \qquad \abs{s} \leq \varphi(s) + \varphi(-s) \leq  c_{\vphi} m_2^\half+ C_{\vphi} \abs{s} \eqsp,
    \end{equation}
such that for any $k \in \{1,\ldots,K\}$ and $(x,v) \in \mse$, $\lambda_k(x,v)= \vphi\big(\ps{v}{F_k(x)}\big)$. 
\end{assumption}
We note that the canonical choice $\vphi(s) = (s)_+$ satisfies these conditions and that the first condition of \eqref{eq:prop_vphi} is equivalent to $\vphi(s) - (s)_+ = \vphi(-s) -(-s)_+$, implying that $\vphi(s)\geq (s)_+$ for all $s\in\rset$ and therefore that the left hand side inequality in \eqref{eq:prop_vphi} is automatically satisfied. If we further assume the existence of $C,c \geq 0$ such that  for all $s \in \rset$, $ \vphi(s) \leq  c m_2^\half + C \, (s)_+$ then the second inequality is satisfied with $C_{\vphi} =  C$ and $c_{\vphi} = 2c$.  As remarked in \cite{andrieu:livingstone:2018}, the first condition of \eqref{eq:prop_vphi} holds for rates based on the choice
\[
\vphi(s)=-\log\left(\phi\big(\exp(-s)\big)\right) \eqsp,
\]
such that $\phi \colon \rset_+ \to [0,1]$ satisfies $r\phi(r^{-1}) = \phi(r)$ for all $r \in \rset_+ \setminus \{0\}$. The canonical choice corresponds to $\phi(r) = 1\wedge r$, but the (smooth) choice $\phi(r) = r/(1+r)$ is also possible. 

%


\begin{assumption}
\label{as:radial} Assume that $\msv$ and $\nu$ satisfy the following conditions.
 \begin{enumerate}
\item \label{item:as:radial:general} 
 $\msv$ is stable under bounces, i.e. for all $(x,v) \in \mse$ and $k \in \iint{1}{K}$, $v - 2 (v^\top \rmn_k(x))\, \rmn_k(x) \in \msv$, where $\rmn_k(x)$ is defined by \eqref{eq:def-bounces}. 
\item \label{item:as:radial:general_2} For any $\msa \in \mcv$,
  $x \in \msx$, we have
  $\nu \parenthese{\defEns{\Id-2\rmn_k(x) \rmn_k(x)^{\top}} \msa } =
  \nu(\msa)$, for any $k \in \iint{1}{K}$.
 \item \label{item:as:radial:general_3} For any bounded and measurable
 function $g :\rset^2 \to \rset$, $i,j \in\iint{1}{d}$ such that $i \neq j$,
 $\int_{\msv} g(v_i,v_j) \, \rmd \nu (v) = \int_{\msv} g(-v_1,v_2) \, \rmd
 \nu (v)$;
\item \label{assum:fourth_moment} $\nu$
 has finite fourth order marginal moment and for $i \in\iint{1}{d}$
 \begin{equation}
 \label{eq:def_m_4}
m_4 = (1/3) \normmu{ v_i^2}^2 = (1/3) \int_{\msv} v_i^4 \, \rmd \measv (v) < \plusinfty \eqsp,
\end{equation}
and for any $i,j,k,l\in\iint{1}{d} $ such that $\card(\{i,j,k,l\})>2$
\begin{equation}
  \int_{\msv}v_i v_j v_k v_l\, \rmd \measv( v) = 0 \eqsp.
\end{equation}
\end{enumerate}
\end{assumption}
Note that in the case where $\msv$ and $\nu$ are rotation invariant,
\ie~for any rotation $O$ on $\rset^d$, $O \msv = \msv$ and for any
$\msa \in \mcv$, $\nu(O \msa) = \nu(\msa)$, then
\Cref{as:radial}-\ref{item:as:radial:general}-\ref{item:as:radial:general_2}-\ref{item:as:radial:general_3}
are automatically satisfied.

By \Cref{as:radial}-\ref{item:as:radial:general_3}, we have
$\int_{\msv} v_1 v_2 \rmd \nu (x) = 0$ taking $g(v_1,v_2) = v_1v_2$ for
any $(v_1,v_2) \in \rset^2$ and therefore for any $i,j \in\iint{1}{d}$
such that $i \neq j$, $\int_{\msv} v_iv_j \, \rmd \nu (v) = 0$. In
addition, under \Cref{as:radial}-\ref{assum:fourth_moment}, from the
Cauchy-Schwarz inequality, we obtain that
\begin{equation}
 \label{eq:def_m_2_2}
 m_{2,2} = \normmu{ v_1 v_2 }^2 = \int_{\msv} v_1^2 v_2^2 \, \rmd \measv (v) < \infty \eqsp,
\end{equation}
and note that in the Gaussian case we have the relation $m_4 = m_{2,2} = m_2^2$.
Finally, under \Cref{as:radial}, for any $f,g \in \mrl^2(\mu)$ and $k \in \{1,\ldots,K\}$, $\psmu{\mcbb_k f}{g} = \psmu{ f}{\mcbb_k g}$, that is $\mcbb_k$ is symmetric on $\mrl^2(\mu)$. 

%

In this paper we consider operators $(\mcr_v,\domain(\mcr_v))$ on $\mrl^2(\mu)$ satisfying the following conditions. In the sequel, we identify $\mrl^2_0(\nu)$ as a subset of $\mrl^2_0(\mu)$. 
\begin{assumption}
\label{as:operator on velocities} 
\begin{enumerate}
\item \label{item:as:operator_on_velocities_self_adjoint} $\calR_v$ satisfies the detailed balance condition: $	\calR_v = \calR_v^{\star}$ and $\mrC^2_{\poly}(\mse) \subset \domain(\calR_v)$;
  \item \label{item:as:operator_on_velocities_projection}
For any $f \in \mrl^2(\pi)$ and $g \in \rmC^2_{\poly}(\mse)$ such that $fg \in \mrl^2(\mu)$ then $fg \in \domain(\mcr_v)$ and  $\mcr_v(fg) = f \mcr_v(g)$; in addition, $\mcr_v(\bigone_{\mse})=0$.
\item \label{item:refreshment} $\calR_v$ admits a spectral gap of size $1$ on $\tLnu$: for any $g \in \mrl^2_0(\nu)\cap \domain(\mcr_v)$, $\pspi{-\calR_v g}{g} \geq \normmu{g}^2$; in addition,  it holds for any $i \in \{1,\ldots,d\}$, $v_i \in \domain(\mcr_v)$ and $	-\calR_v (v_i ) = v_i $.
\end{enumerate}
\end{assumption}
Typically, $\mcr_v$ is of the form $\Id \otimes \tilde{\mcr}_v$ where
$(\tilde{\mcr}_v,\domain(\tmcr_v))$ is a self-adjoint operator on
$\mrl^2(\nu)$ with spectral gap equals $1$. Then, condition
\Cref{as:operator on
 velocities}-\ref{item:as:operator_on_velocities_projection} is equivalent to
$\tmcr_v(\bigone_{\msv}) = 0$, which implies that for any $g \in \domain(\tmcr_v)$, we have
\begin{equation}
	\int_{\msv} \tmcr_v g \, \rmd \nu = \psmu{\bigone_{\mse}}{\mcr_v g} = \psmu{\calR_v^{\star}(\bigone_{\msv})}{g} = \psmu{\calR_v(\bigone_{\msv})}{g} = 0 \eqsp,
\end{equation}
 so that the process associated with $\tmcr_v$ preserves the probability measure $\measv$.

Note that \Cref{as:operator on
 velocities}-\ref{item:as:operator_on_velocities_projection} implies
that $\calR_v \Piv = 0$, whereas \Cref{as:operator on
 velocities}-\ref{item:refreshment} implies that
$-\calR_v (v_1 \Piv) = v_1 \Piv$, where $\Piv$ is defined by
\eqref{eq:def_piv}. 
Assumption \Cref{as:operator on velocities} is satisfied when $\calR_v=\Pi_v$, or $\mcr_v = \Id \otimes \tmcr_v$ with $\tmcr_v$ the generator of the Ornstein-Uhlenbeck process defined for any $g \in \mrCb^2(\rset^d)$ by
 \begin{equation}
 \tmcr_v g = -\ps{\nabla_v g}{v} + \Delta_v g \eqsp. 
 \end{equation}
\begin{assumption}
\label{as:refreshment}
The refreshment rate $\lambdaref : \msx \to \bbR_+$ is bounded from below and from above as follows: there exist $\ulambda>0$ and $c_{\lambda} \geq 0$ such that for all $x \in \msx$, 
\begin{equation}
 0 < \ulambda \le \lambdaref(x) \le \ulambda (1+c_{\lambda}|\nax U(x)|) \eqsp.
\end{equation}	
\end{assumption}
Under the previous assumptions we can prove exponential convergence of the semigroup. 
\begin{theorem}
  \label{thm:hypocoercivity}
Assume that $\mcl_i$, $i \in \{1,2\}$ given by \eqref{eq:generator pdmp} or \eqref{eq:generator_r_d_2} satisfies \Cref{as:generator} with $\msc = \rmCb^2(\mse)$ and  \Cref{as:U},  \Cref{as:Fk}, \Cref{as:intensities}, \Cref{as:radial}, \Cref{as:operator on velocities} and \Cref{as:refreshment} hold.  Then there exist $A > 0$ and $\alpha > 0$ such that, for any $\ff \in \tLmu$, and $t \in \rset_+$,
\begin{equation}
 \label{eq:cv_expo_L2}
\normmu{ P_t \ff } \le A \, \rme^{-\alpha t} \normmu{ \ff } \eqsp.
\end{equation}
The constants $A$ and $\alpha$ are given in explicit form in \eqref{eq:def_alpha_C_eps} in  \Cref{thm:DMSmain} (Section \ref{sec:DMSabstract}), in terms of the constant appearing in \Cref{as:U}, \Cref{as:Fk}, \Cref{as:radial}, \Cref{as:operator on velocities} and \Cref{as:refreshment}, where $\epsilon$ can be taken to be $\epsilon_0$ given in \eqref{eq:epsilon_0}, $	\lambda_v = \underline{\lambda}$, $\lambda_x = \Cp / (1+\Cp) $
  and $R_0 = (4+2\sqrt{3}) \vee (\ulambda/2^{\half}) \vee \bar{R}_0 $ where 
\begin{align}
\label{eq:def_R_0_cas_gene}
\bar{R}_0 &= \frac{\sqrt{2m_{2,2}+3(m_{4}-m_{2,2})_{+}}}{m_2}\defEns{\frac {2^{1/2} (1+C_{\vphi})\kappavara}{\kappavarb} \sum_{k=1}^K a_k + \kappa_1 } \\
&\qquad \qquad \qquad \qquad + \frac{\ulambda}{2^\half} \defEns{ 1+ \frac {2 c_{\lambda}\kappavara } { \kappavarb}} + \frac{c_\vphi K}{2 ^\half}\eqsp,
 \end{align}

  $\kappavara = (1 + c_1/2)^{\half}$ and $\kappavarb^{-1} = \Cp^{-1}(1+4 c_2 d^{1+\varpi}  + 16 \Cp^2)^{\half}$.
\end{theorem}
\begin{proof}
  The proof is postponed to  \Cref{ss:proof-main-result}.
\end{proof}

 The following details the expected scaling behaviour with $d$ of $A$ and $\alpha$. The proof can be found in \Cref{sec:proof-theor-refthm:s}.
\begin{corollary}\label{thm:scaling-with-d} Consider the assumptions and notation of \Cref{thm:hypocoercivity}. Further suppose that there exists  $m_b>0$ satisfying 
\begin{equation} \label{eq:assumption-moments-velocity}
m_2^{-1}\sqrt{2m_{2,2}+3(m_{4}-m_{2,2})_{+}} \leq m_b \eqsp,
\end{equation}
which together with  $\Cp,c_1,c_2$
and $\norm{a}_{\infty}=\sup_{k \in \{1,\ldots,K\}} a_k$ are independent of $d$. Then $A \leq 3^{\half}$ and there exists  $C^{\alpha}(\Cp,c_1,c_2,\norm{a}_{\infty},m_b) >0$, independent of $d,\ulambda,c_\lambda$ and $C_\vphi,c_\vphi$, such that for $d$ large enough, 
 \begin{align}
\label{eq:borne_alpha_comp_calcul_dim}
\qquad    \alpha &> C^{\alpha}(\Cp,c_1,c_2,\norm{a}_{\infty},m_b) \, \ulambda \, m_2^{\half} \\
   &  \qquad \times \big[\{c_{\vphi}K \}\vee \{(1+C_{\vphi})d^{(1+\varpi)/2} K + 1\}\vee\{\ulambda(1+c_{\lambda} d^{(1+\varpi)/2})\}\big]^{-2} \eqsp. 
\end{align}
Thus, if $\ulambda$, $c_{\lambda}$, $C_{\vphi}$ and $c_{\vphi}$ are fixed, we get that $\alpha^{-1}$ is in general at most of order $\bigO(m_2^{-\half} d^{1+\varpi} K^2)$ if $K \geq 1$.
\end{corollary}

We now discuss the assumptions of the theorem, and application of its conclusion to various instances of PDMP-MC and two examples of potentials.  Assumption \Cref{as:U} is problem dependent and verifiable in practice, while \Cref{as:Fk}, \Cref{as:radial}, \Cref{as:operator on velocities} and \Cref{as:refreshment} are user controllable and we have already discussed standard choices satisfying these conditions. More delicate may be establishing that \Cref{as:generator} holds and that $\mrCb^2(\mse)$ is indeed a core for the generator $\mcl$.    As shown in \cite{durmus-invariant-2018arXiv180705421D}, BPS and ZZ are well defined Markov process whose generators admit $\mrCb^2(\mse)$ as a core and similar arguments can be used to establish that it is also a core for the RHMC. Further, it is not difficult to show that for the class of processes described earlier, for any $f\in \mrCb^2(\mse)$, $\psmu{\mcl f} {\bfone}=0$, therefore implying that $\mu$ is an invariant distribution and that \Cref{as:generator} holds. 

First we note that the spectral gap is indeed expected to be proportional to  $m_2^\half$, since if $(X_t,V_t)_{t \geq 0}$ is a PDMP with generator of  the form \eqref{eq:generator pdmp} or \eqref{eq:generator_r_d_2} for $m_2=1$, then $(X_{m^{\half} t},m^{\half} V_{m^{\half} t})_{t \geq 0}$ is a PDMP with generator of  the same form with $m_2=m$. We therefore set $m_2=1$ below, a condition satisfied when $\nu$ is the uniform distribution on the sphere $\sqrt d \, \sphere^{d-1}$or $\{-1,1\}^d$, or the $d$-dimensional zero-mean Gaussian distribution with covariance matrix $\Idd$, all of which also satisfy \eqref{eq:assumption-moments-velocity}. More generally, by \Cref{lem:moments-spherically-symmetric} in \Cref{app:radial}, property \eqref{eq:assumption-moments-velocity} is satisfied if $\nu$ is a spherically symmetric distribution
 on $\rset^d$ corresponding to random variables $V=B^{\half} W$
for $W$ uniformly distributed on the  hypersphere $ \sqrt d \, \bbS^{d-1}$ and $B$ a non-negative random variable independent of $W$ and of first and second order moments $\gamma_1$ and $\gamma_2$ respectively such that $\gamma_2^{1/2}/\gamma_1$ is upper bounded by a constant independent of the dimension.

By \cite[Proposition 5.1.3, Corollary 5.7.2]{bakry:gentil:ledoux:2014}, independence of $\Cp$ on $d$ is satisfied for strongly convex potentials $U$: \ie~whenever there exists $m>0$ such that $\nax^2U(x) \succeq m \Idd$ for any $x \in \rset^d$ which implies that one can take $\Cp=m$. This is the case for $U(x)=\sum_{i=1}^{d}\big(1+x_{i}^{2}\big)^{\beta}/2$ or $U(x)=(1+\vert x\vert^{2})^{\beta}$ with $\beta \geq 1$, for which \eqref{eq:laplacian bound} is also satisfied with $\varpi=0$ and $\varpi = 1-1/\beta$ respectively  (see \Cref{lemma:potential:independent} and  \Cref{lemma:potential:polynomial} in \Cref{sec:examples-potentials}). We note that from the Holley-Stroock perturbation principle \cite{Holley87}, uniformly bounded perturbations of a strongly convex potential lead to independence of $\Cp$ on $d$. For $\beta \in [1/2,1)$ $\Cp>0$, but is dependent on $d$, see \cite[Chapter 4]{bakry:gentil:ledoux:2014}. However recent progress in the precise quantitative  estimation of spectral gaps of certain probability measures \cite{Bobkov2003,BONNEFONT20162456} allows for the strong convexity property to be relaxed to simple convexity and beyond, but leads to a dependence of $\Cp$ on $d$ which can be characterised. 

Now further assume that $C_\vphi,c_\vphi$ and  that the refreshment rate are uniformly bounded in the position $x$, implying $c_{\lambda}=0$. Then by \Cref{thm:scaling-with-d}-\eqref{eq:borne_alpha_comp_calcul_dim}, there exists  $C^{\alpha}(\Cp,c_1,c_2,\norm{a}_{\infty},m_b,c_\vphi, C_\vphi)>0$ such that for $d$ sufficiently large
\begin{equation}
\alpha \geq C^{\alpha}(\Cp,c_1,c_2,\norm{a}_{\infty},m_b,c_\vphi, C_\vphi) \,\left\{ \big[\ulambda \big(1+K^2 d^{1+\varpi}\big)^{-\half}\big] \wedge \ulambda^{-1} \right\} \eqsp,
\end{equation}
from which we deduce the optimal scaling of the refreshment rate, namely $C_1^{\ulambda} \, \big(1+K^2 d^{1+\varpi}\big)^\half \leq \ulambda \leq C_2^{\ulambda} \, \big(1+K^2 d^{1+\varpi}\big)^\half$ for $C_1^{\ulambda}, C_2^{\ulambda}>0$ (which we denote $\Theta\big((1+K^2 d^{1+\varpi})^\half\big)$ hereafter to alleviate notation).  Using the description of RHMC, ZZ and BPS provided in the introduction we deduce the first three lines of Table \ref{tab:summary}, where $\alpha=\omega(s)$ is used as a short hand notation for $\alpha \geq C^{\alpha}(\Cp,c_1,c_2,\norm{a}_{\infty},m_b,c_\vphi, C_\vphi) s$ for $s \to 0$. The fourth line uses our specialised results of Section \ref{ss:scaling-ZZ}, showing that the conclusion of Theorem \ref{thm:scaling-with-d} is not optimal for ZZ. 

In \cite{bierkens2018arXiv180711358B}  scaling limits of particular functionals of the ZZ and BPS processes are studied, leading to quantitative estimates of the time required to achieve near independence at equilibrium. More specifically they consider the scenario where the target distribution is a centred normal distribution of covariance matrix $\Idd$ and focus on the angular momentum, the negative log-target density and the first coordinate of the process. Our more general results, obtained using a different argument, are in agreement after noticing that \cite{bierkens2018arXiv180711358B} considered the scenario $m_2=d^{-1}$ and using our earlier remark on the dependence of our estimate of the absolute spectral gap on $m_2^\half$.  In \cite{deligiannidis2018randomized} it is shown, again using an approach different from ours, that the RHMC has dimension free convergence rate in a scenario similar to ours.

\bgroup
\def\arraystretch{1.5}
\begin{table} 
\begin{centering}

\begin{tabular}{c|c|c||c|c|}
\multicolumn{1}{c}{} &\multicolumn{1}{c}{} &\multicolumn{1}{c}{} & \multicolumn{2}{c}{$U(x)=$} \tabularnewline
\multicolumn{1}{c}{} & \multicolumn{1}{c}{$\underline{\lambda}$} & \multicolumn{1}{c}{$\alpha$} & \multicolumn{1}{c}{\smaller $\frac{1}{2} \sum_{i=1}^{d}(1+x_{i}^{2})^{\beta} $}   & \multicolumn{1}{c}{\smaller $(1+\vert x\vert^{2})^{\beta}$}\tabularnewline
\cline{2-5} 
RHMC & $\Theta(1)$ & $\omega\big(\underline{\lambda}\wedge\underline{\lambda}^{-1}\big)$ &  & \tabularnewline
\cline{2-3} 
BPS & $\Theta\big(d^{(1+\varpi)/2}\big)$ & $\omega\big(d^{-(1+\varpi)/2}\big)$ & $\beta\geq1$   & $\beta\geq1$\tabularnewline
\cline{2-3} 
ZZ (crude) & $\Theta\big(d^{(3+\varpi)/2}\big)$ & $\omega\big(d^{-(3+\varpi)/2}\big)$ & $\varpi=0$ & $\varpi=1-1/\beta$ \tabularnewline
\cline{2-5} 
ZZ (Section \ref{ss:scaling-ZZ}) & $\Theta(1)$ & $\omega(1)$ & $\beta\geq1$ & $\beta=2$\tabularnewline
\cline{2-5} 
\end{tabular}
\vspace{.25cm}
\caption{Left hand side: summary of the dependence of $\alpha$ on $d$ for $C_{P},c_{1},c_{2},\|a\|_{\infty}$
constant, $m_{2}=1$ and optimal choice of $\underline{\lambda}$.
Right hand side: summary of application to two examples of potentials.}\label{tab:summary}
\par
\end{centering}
\end{table}

While nonreversibily of the processes considered here may be practically beneficial, it is only recently that the tools allowing our work have been developed \cite{villani2006hypocoercive,Villani09}. Our method of proof relies on the framework proposed recently in \cite{Dolbeault09,Dolbeault15,2017arXiv170806180B}  to study the solutions of the forward Kolmogorov equation associated with the linear kinetic process, but we study the dual backward Kolmogorov equation for a broader class of processes as is the case in \cite{GS2014,GS2015,GS2016} who provide the first rigorous derivation of the results of \cite{Dolbeault09,Dolbeault15,2017arXiv170806180B}. This, combined with the flexibility of the framework of \cite{Dolbeault15,2017arXiv170806180B} explains the differing inner product used throughout, which we have found to lead to simpler computations while yielding identical conclusions. The estimate \eqref{eq:decay estimate} (with constant $A=1$) would follow straightforwardly from a Gr\"onwall argument if the generator $\mcl$ of the semigroup was coercive, that is it satisfied $\psmu{\generator f}{f} \leq -a \normmu{f}^2$ for some $a>0$ and any $f$ in a core of $\generator$.  Unfortunately, the symmetric part of the generator corresponding to a PDMP is degenerate in general, in the sense that it has a nontrivial null space. Hence, the aforementioned coercivity clearly fails to hold. However, it is possible to equip $\mrl^2(\mu)$ with an equivalent scalar product derived from $\psmu{\cdot}{\cdot}$ with respect to which $\mathcal{L}$ is coercive. The constant $\alpha$ is then given by the coercivity bound, while the constant $A$ can be obtained from estimates relating the two equivalent scalar products.

The paper is organised as follows. In \Cref{sec:DMSabstract} we develop our framework for hypocoercivity suited to PDMP-MC processes, based on the ideas of \cite{Dolbeault15}. In addition to providing a rigorous framework we further optimize the constants involved, ultimately leading to  \Cref{thm:hypocoercivity}. The proofs of \Cref{thm:hypocoercivity} and its corollary are given in \Cref{sec:postponed-proofs}. In \Cref{ss:scaling-ZZ}, we specialize our results to the case of the Zig-Zag process for which better estimates are possible, leading to attractive scaling properties with the dimension $d$.  Various intermediate technical results have been moved to Appendices where, for completeness, we have also included classical facts from functional analysis.


\section{The DMS framework for hypocoercivity} \label{sec:DMSabstract}
As stated above our results rely on the ideas proposed by \cite{Dolbeault09,Dolbeault15,2017arXiv170806180B} for which a rigorous framework was subsequently given in \cite{GS2014,GS2015,GS2016,grothaus2017weak}. We derive here a novel proof, which borrows elements of \cite{GS2014,GS2015,GS2016,grothaus2017weak} but leads to a different set of conditions motivated by our application to PDMP-Monte Carlo methods. We further provide explicit and optimized estimates of the constants involved in terms of accessible characteristics of the process. We first present abstract results which form the core of all of our proofs and then establish more specific ones common to all the processes considered in this paper, implying some of the abstract conditions. More specific results relating to the Zig-Zag process are treated in  \Cref{ss:scaling-ZZ}.

\subsection{Abstract DMS results}

We let $\calS$ and $\calT$ be the
$\Lmu$-symmetric and $\Lmu$-anti-symmetric parts of a generator
$\calL$ satisfying \Cref{as:generator}, that is
\begin{equation}
 \label{eq:def_calS_calT}
\calS = (\calL+\calL^{\star})/2 \quad \text{and} \quad \calT = (\calL-\calL^{\star})/2 \eqsp, \eqsp \text{ defined on } \domain(\calS) = \domain(\calT) = \msc \eqsp.
\end{equation}
Consider the following additional assumption to \Cref{as:generator}.
\begin{assumptionA}
  \label{ass:stability_c_piv}
  $\Piv \msc \subset \msc$ where $\Pi_v$  is  defined by  \eqref{eq:def_piv} and $\msc$ is given in \Cref{as:generator}.
\end{assumptionA}

Note that since   $\Piv \msc \subset \msc$,  we  have $\msc \subset \domain(\mct \Piv)$ and the restriction of $\mct\Piv$ to $\msc$ exists. Under \Cref{as:generator} and
  \Cref{ass:stability_c_piv}, $\mct \Piv$ is a closable operator of closure $(\bar{\mct \Piv}, \domain(\bar{\mct \Piv}))$ since $\mct$
  is anti-symmetric and $\msc$ is dense. Although this result follows easily from standard theory of unbounded (anti)-symmetric operators on Hilbert space, a proof is given for completeness in \Cref{lem:closure_anti_symm}. We point out here the difference with the corresponding assumption of \cite{GS2014,GS2015,GS2016}, which justifies the development of a novel theoretical framework. Indeed, motivated by our applications and what is currently understood of their theoretical properties, our starting point is the restriction of $\mct \Piv$ to the core $\msc$ and we assume closability, while in \cite[p.3522]{GS2014}  or \cite[p. 155 condition D4]{GS2016} the authors  consider directly a closed extension of $(\mct \Piv,\msc)$, say $(\mct \Piv,\msdd)$. Applying the results of \cite{GS2014,GS2015,GS2016} would require showing that $\msc$ is a core for $(\mct\Piv,\msdd)$ which, to the best of our knowledge, appears to be very difficult for the processes we are interested in.

\Cref{lem:def_inverse_dms}  in \Cref{ss:spctral bounds} justifies the definition of the operator $\calA$,
\begin{equation}
\label{eq:defcalA}
	\calA = \left( m_2 \Id + (\calT \Piv)^{\star} (\overline{\calT \Piv}) \right)^{-1} (-\calT \Piv)^{\star}\eqsp, \quad \domain(\calA) = \domain((\calT \Pi_v)^{\star}) \eqsp,
      \end{equation}
      where $m_2$ is given by \eqref{eq:def_m_2} and $(\overline{\mct \Piv}, \domain(\overline{\mct \Piv}))$ and $((\mct\Piv)^{\star},\domain((\mct \Piv)^{\star}))$ are the closure and the adjoint of $(\mct \Piv,\msc)$ respectively. Key properties are that $\range(\calA) \subset \domain(\overline{\mct \Piv})$, $\mca$ is closable with $\bmca$ bounded, and $\overline{\mct \Piv} \mca$ is also closable of bounded closure. To show this result we adapt  \cite[Lemma 2.4]{GS2014} since their lemma assumes that $(\mct,\domain(\mct))$ is closed whereas, motivated by our applications, we assume $(\mct \Piv, \msc)$ to be a densely defined and closable operator instead. Below $\normopmu{ \cdot}$ refers to the operator norm associated to $\normmu{\cdot}$, as defined in the notation paragraph in the introduction.

 \begin{lemma}
   \label{lem:bounded_A}
   Let $(\mct,\domain(\mct))$ be an anti-symmetric  densely defined operator on $\rml^2(\mu)$.  Assume that there exists $\msdd \subset \domain(\mct \Piv) \cap \domain(\mct)$, such that  $(\mct \Piv,\msdd)$ is a densely defined closable operator. 
      \begin{enumerate}[label=(\alph*), wide, labelwidth=!, labelindent=0pt]
 \item \label{lem:bounded_A_item_a_0}
 The closure of $(\mct \Piv,\msdd)$, $(\overline{\mct \Piv}, \domain(\overline{\mct \Piv}))$ satisfies $\domain(\overline{\mct \Piv}) \subset \domain((\Piv \mct)^{\star})$ and for any $f \in \domain(\overline{\mct \Piv})$, $(\Piv \mct)^{\star} f= -\overline{\mct \Piv}f$, where $((\Piv \mct)^{\star},\domain((\Piv \mct)^{\star} ))$ is the adjoint of $(\Piv \mct,\msdd)$.
 \item  \label{lem:bounded_A_item_a}
The operator $\mca$ defined by \eqref{eq:defcalA} satisfies $\range(\mca) \subset \domain(\overline{\mct \Piv})$, is closable and its closure   $\bmca$ is a bounded operator on $\rml^2(\mu)$  with	$\normopmu{ \bmca} \le 1/ (2 m_2)^{\half}$ and  $\Piv \bmca = \bmca$  on $\Lmu$.
 \item  \label{lem:bounded_A_item_b} Assume in addition that 
  for any $f \in \msdd$, $\Piv \mct \Piv f=0$. Then, the operator $(\overline{\mct \Piv} \mca,\domain(\overline{\mct \Piv} \mca))$ is also closable and its closure $\mce$ is bounded and satisfies for any $f \in \mrl^2(\mu)$, $\normmu{ \mce f} \le \normmu{(\Id-\Piv)f}$.
 \end{enumerate}
 \end{lemma}
 \begin{proof}
 To establish this result, we make use of classical results on unbounded operators in Hilbert spaces which for completeness, are given in \Cref{ss:spctral bounds}.
   
 \begin{enumerate}[wide, labelwidth=!, labelindent=0pt, label= (\alph*)]
   \item Since $\mct$ is assumed to be anti-symmetric, we have for any $f \in \domain(\Piv \mct)$, $ g \in \msdd$,  $\psmu{\Piv \mct f }{g} = - \psmu{f}{\mct \Piv g}$ since $\Piv g \in \domain(\mct)$ as $\msdd \subset \domain(\mct \Piv)$. By definition of $(\mct \Piv)^{\star}$, we obtain that $\msdd \subset \domain((\Piv \mct)^{\star})$, and for any $f \in \msdd$, $ \mct \Piv f = -(\Piv \mct)^{\star} f$. Therefore $\{(f,  \mct \Piv f) \, : \, f \in \msdd\} \subset \{(f,-(\Piv \mct)^{\star} f) \, : \, f \in \domain((\Piv \mct)^{\star})\}$, and we obtain the desired result by definition of the operator $(\overline{\mct \Piv}, \domain(\overline{\mct \Piv}))$ since $-(\Piv \mct)^{\star}$ is closed by \cite[Theorem 5.1.5]{pedersen1995analysis}.  
   \item The fact  that $\range(\mca) \subset \domain(\overline{\mct \Piv})$, $\mca$ is closable and the bound follow directly from \Cref{lem:def_inverse_dms} and \Cref{prop:abstract bound}-\ref{prop:abstract bound_a}-\ref{prop:abstract bound_d}. We turn to the statement $\Piv \bmca = \bmca$.  By  \Cref{lem:def_inverse_dms}, the operator $\mcc= (m_2 \Id + (\calT \Piv)^{\star} (\overline{\calT \Piv}))^{-1}$ is well-defined, bounded and  $\range(\mcc) = \domain( (\mct \Piv)^{\star} (\overline{\mct \Piv}) )$. Therefore using \Cref{lem:grothaus_lemme_2_3}-\ref{lem:grothaus_lemme_2_3_a} (since $\mct \Piv$ is densely defined), we have for any $f \in \domain(\mct)$,
     \begin{equation}
       \label{eq:proof_lemma:op norms_0}
       \mca f = \mcc \Piv \mct f = m_2^{-1}\defEns{\Id- (\mct \Piv)^{\star} (\overline{\mct \Piv}) \mcc }\Piv \mct f  \eqsp,
     \end{equation}
 where the argument for the last equality can found in the proof of Proposition \ref{prop:abstract bound}.
 Therefore, by applying $\Piv$ to both sides and using \Cref{lem:grothaus_lemme_2_3}-\ref{lem:grothaus_lemme_2_3_b}, we deduce that for any $f \in \domain(\mct)$, $ \Piv \mca f = \mca f$. The proof is then concluded upon noting that  $\domain(\mct)$ is dense and $\Piv$ is continuous.
\item  
  
  
For any   $f \in \msdd$, since $\Piv \mct \Piv f=0$, \eqref{eq:proof_lemma:op norms_0} becomes
     \begin{equation}
         \mca f = \mcc \Piv \mct (\Id - \Piv) f = m_2^{-1}\defEns{\Id- (\mct \Piv)^{\star} (\overline{\mct \Piv}) \mcc }\Piv \mct (\Id - \Piv) f \eqsp.
     \end{equation}
      Therefore, we get for any $f\in\msdd$,
     \begin{align}
       &  m_2     \normmu{\mca f}^2\\
       &= \psmu{\Piv \mct (\Id - \Piv) f}{\mca f} - \psmu{(\mct \Piv)^{\star} (\overline{\mct \Piv}) \mcc \Piv \mct (\Id - \Piv) f }{\mca f}\\
                         & =  \psmu{ -( \mct \Piv)^{\star}(\Id - \Piv) f}{\mca f} - \psmu{(\mct \Piv)^{\star} (\overline{\mct \Piv}) \mcc \Piv \mct (\Id - \Piv) f}{\mca f}\\
                             & =-\psmu{ (\Id - \Piv) f}{(\overline{ \mct \Piv})\mca f} - \normmu{(\overline{ \mct \Piv}) \mca f}^2 \eqsp,
     \end{align}
     using successively that  $(\Id-\Piv)f \in \domain(\mct) $ since
     $f\in \msdd \subset \domain(\mct \Piv)$, \Cref{lem:grothaus_lemme_2_3} and
     $\mca f \in \domain(\overline{\mct \Piv})$. Using the Cauchy-Schwarz inequality we obtain that for any $f \in \msdd$, $\normmuLine{(\overline{\mct\Piv})\mca f } \leq \normmuLine{(\Id-\Piv)f}$. Using that 
     $\msdd$ is dense in $\Lmu$ together with the bounded linear transformation extension theorem \cite[Theorem I.7]{reed1972methods} concludes the proof.
   \end{enumerate}
 \end{proof}
The main result of \cite{Dolbeault15} can be formulated under the following abstract assumption, which we shall assume to hold from now on, and the proof of our main theorem relies on optimized estimates of the constants involved.
\begin{assumptionA}[DMS abstract conditions] \label{as:DMSabstract}
Let $\core$ be as in \Cref{as:generator}. Assume further that it satisfies \Cref{ass:stability_c_piv} and the following conditions   
\begin{enumerate}
\item \label{item:DMS-micro} there exists $\lambda_v>0$ satisfying for any $\ff \in \core$
\begin{equation}
-\psmu{ \calS \ff}{ \ff} \geq \lambda_v m_2^{\half} \normmu{ (\Id-\Piv) \ff }^2 \eqsp;
\end{equation}
\item \label{item:DMS-macro} there exists $\lambda_x\in\ooint{0,1}$ satisfying for any $\ff \in \core$
 \begin{equation}
 \label{eq:DMS_macro}
-\psmu{ \bmca\calT\Piv\ff}{\ff} \ge \lambda_x \normmu{ \Piv \ff }^2 \eqsp;
\end{equation}
\item \label{item:DMS-RSandRT} there exists $R_0 \geq 0$ satisfying for any $\ff \in \core$
 \begin{equation}
 \abs{\psmu{ \bmca \calT(\Id-\Piv)\ff}{\ff}+\psmu{ \bmca\calS\ff}{\ff}}\leq R_0\normmu{(\Id-\Piv )\ff} \normmu{\Piv \ff} \eqsp;
 \end{equation}
\item \label{item:DMS-proj}   for any $f \in \msc$, $\Piv \calT \Piv f =0 $; 
\item \label{item:DMS-kernel_calS} finally,
 $\range(\Pi_v) \subset \kernel(\calS^{\star})$.
\end{enumerate}
\end{assumptionA}

\begin{theorem} \label{thm:DMSmain}
Assume \Cref{as:generator}, \Cref{ass:stability_c_piv} and \Cref{as:DMSabstract}. 
\begin{enumerate}
\item \label{thm:DMSmain_item_a} Then, for any $\ff \in \mrl^2_0(\mu)$, $t \in \rset_+$ and $\epsilon\in (0, (2^{\half} \lambda_v)^{-1} \wedge \{ 4 \lambda_x /(4\lambda_x+R_0^2)\} )$
\begin{equation}
\normmu{ P_t \ff } \le A(\epsilon) \rme^{-\alpha(\epsilon) t} \normmu{ \ff } \eqsp,
\end{equation}
with 
\begin{equation}
\label{eq:def_alpha_C_eps}
\alpha(\epsilon) = \lambda_v m_2^\half \frac {\Lambda(\epsilon)}{1+2^{\half} \lambda_v \epsilon}>0 \quad \mbox{and} \quad A(\epsilon) = \sqrt{\frac {1+2^\half \lambda_v \epsilon}{1-2^\half \lambda_v \epsilon}} \eqsp,
\end{equation}
where
\begin{equation}
 \label{eq:def_lambda_0}
\Lambda(\epsilon)=\frac{1-\epsilon(1-\lambda_{x})-\sqrt{[1-\epsilon(1-\lambda_{x})]^{2}-4\epsilon\lambda_{x}(1-\epsilon)+\epsilon^{2}R_{0}^{2}}}{2} \eqsp. 
\end{equation}

\item \label{thm:DMSmain_item_b} Further, if $2^\half R_0 \geq \lambda_v $ then
 $\alpha \colon \big(0, 4 \lambda_x /(4\lambda_x+R_0^2) \big) \rightarrow
 \rset_+$ has a unique maximum at $\epsilon^{\star}$ such that $\alpha(\epsilon_0) < \alpha(\epsilon^{\star}) <3\alpha(\epsilon_0)$, with
\begin{equation}
  \label{eq:epsilon_0}
	\epsilon_0 = \frac {1+\lambda_x - (1-\lambda_x) \sqrt{\frac{R_0{^2}}{R_0{^2}+4 \lambda_x}}}{(1+\lambda_x)^2 + R_0{^2}} \in (0, (2^{\half} \lambda_v)^{-1} \wedge \{ 4 \lambda_x /(4\lambda_x+R_0^2)\} ) \eqsp,
\end{equation}
 so that $A(\epsilon_0)<\plusinfty$ is well defined.
In addition, if $R_0 \geq 2$ then $\epsilon_0 < 3 \lambda_x/(4 \lambda_x + R_0^2)$.
 \end{enumerate}
\end{theorem}

The main idea of \cite{Dolbeault15} behind the proof of \Cref{thm:DMSmain} is the introduction of an equivalent norm for $\varepsilon \in \rset_+$ (instead of the $\Lmu$ norm, which corresponds to $\varepsilon=0$)
\begin{equation} \label{eq:defcalH}
	\entropyA_{\veps}(\ff) = (1/2) \normmu{ \ff }^2 + \varepsilon \psmu{ \ff}{ \mcab \ff },
\end{equation}
for which $(P_t)_{t \geq 0}$ is exponentially contracting. More precisely, \cite[Theorem 2]{Dolbeault15} shows that for some $\varepsilon \in \oointLine{-(m_2/2)^{\half},(m_2/2)^{\half}}$  there exists $\alpha(\varepsilon) >0$ such that for any $\ff \in \rml^2_0(\mu)$, $\entropyA_{\veps}(P_t\ff) \leq \rme^{-\alpha(\varepsilon)t} \entropyA_{\veps}(\ff)$. Then, the convergence in $\mrl^2_0(\mu)$ follows by \Cref{lem:bounded_A}-\ref{lem:bounded_A_item_a} which implies that $\entropyA_{\veps}(\cdot)$ defines a norm which is equivalent to $\normmuLine{ \cdot }$: for $\varepsilon \in \oointLine{-(m_2/2)^{\half},( m_2/2)^{\half}}$ and for any $\ff \in \mrl^2(\mu)$, it holds 
\begin{equation} \label{eq:equinorm}
(	1-(m_2/2)^{-\half} \varepsilon) \normmu{ \ff }^2 \le 2\entropyA_{\veps}(\ff) \le (1+(m_2/2)^{-\half} \varepsilon) \normmu{ \ff }^2.
\end{equation}

Therefore, for a family $\big\{f_t \in \mrl^2_0(\mu) \big\}_{t\geq 0}$, exponential decay of $t\mapsto \entropyA_{\veps}(\ff_t)$ is equivalent to that of $t\mapsto \normmu{ \ff_t }^2$ , a property exploited in the following proof. We first establish the following results which give estimates of the functional $\{\Fscr_i \, : \, i \in \{1,2,3\}\}$ defined for any $g \in \domain(\generator)$ by
\begin{equation}
\label{eq:def_Fscr}
  \Fscr_1(g) = \psmu{\generator g}{g} \eqsp, \quad  \Fscr_2(g) = \psmu{\generator g}{ \mcab g} \eqsp, \quad \Fscr_3(g) = \psmu{\mcab \generator   g}{ g}\eqsp. 
 \end{equation}
 \begin{lemma}
   \label{lem:preli_proof_dms_bound_fscr}
   Assume that $\calL$ satisfies \Cref{as:generator}, \Cref{ass:stability_c_piv}, and \Cref{as:DMSabstract}. Then, for any $g \in \domain(\generator)$, we have
   \begin{equation}
     \label{eq:lem:preli_proof_dms_bound_fscr}
   \begin{aligned}
     \Fscr_1(g) \leq & -\lambda_v m_2^{\half}\normmu{(\Id-\Piv)g}^2 \eqsp, \quad  \Fscr_2(g) \leq \normmu{(\Id-\Piv)g}^2 \eqsp, \\
                                                       &   \Fscr_3(g) \leq -\lambda_x \normmu{\Piv g}^2 +R_0\normmu{(\Id-\Piv)g}\normmu{\Piv g} \eqsp.
\end{aligned}
\end{equation}
\end{lemma}
 \begin{proof}

   Note that since  $\core$ is a core for $\generator$ and $\bmca$ and $\Piv$ are bounded, we only need to show that \eqref{eq:lem:preli_proof_dms_bound_fscr} holds for all $g \in \core$. In addition, since $\mcab$ is an extension of $\mca$ by \Cref{lem:bounded_A}-\ref{lem:bounded_A_item_a}, and for any $g \in \core  \subset \domain((\mct \Piv)^{\star}) = \domain(\mca)$ from \Cref{lem:grothaus_lemme_2_3}-\ref{lem:grothaus_lemme_2_3_a} as $\Piv(\core) \subset \core =  \domain(\mct)$ by \Cref{ass:stability_c_piv}, we deduce
   \begin{equation}
     \label{eq:eq_mca_mcab}
     \mcab g = \mca g \eqsp.
   \end{equation}
   
   Using that $\mcs$ is symmetric, $\mct$ is anti-symmetric and $\core \subset \domain(\generator)\cap \domain(\generator^{\star})$, we get that for any $g \in \core$, $\Fscr_1(g) = \psmu{\mcs g}{g} \leq -\lambda_vm_2^{\half} \normmu{(\Id-\Piv)g}$ by \Cref{as:DMSabstract}-\ref{thm:DMSmain_item_a}.

  Second, using that  
   $\Piv \mcab = \mcab$ by \Cref{lem:bounded_A}-\ref{lem:bounded_A_item_a} and \eqref{eq:eq_mca_mcab},  we have  for any $g \in \msc$,
   \begin{equation}
     \Fscr_2(g) = \psmu{\Piv \mca g}{\mcs g} + \psmu{\Piv \mca g}{\mct g} = \psmu{\Piv \mca g}{\mct g} \eqsp,
   \end{equation}
   where the last equality follows from   $\range(\Piv) \subset \kernel(\mcs^{\star})$.
   In addition, since $\Piv$ is symmetric, $\Piv \mct \Piv g = 0$, $\range(\mca) \subset \domain(\overline{\mct \Piv})\subset\domain((\Piv \mct)^{\star})$ by \Cref{lem:bounded_A}-\ref{lem:bounded_A_item_a_0}-\ref{lem:bounded_A_item_a}, so $(\Piv \mct)^{\star} \mca = -\overline{\mct \Piv} \mca$ by \Cref{lem:bounded_A}-\ref{lem:bounded_A_item_a_0} and $\normmu{\overline{\mct \Piv} \mca g} \leq \normmu{(\Id-\Piv)g}$ by \Cref{lem:bounded_A}-\ref{lem:bounded_A_item_b}, we obtain for any $g \in \core$,
   \begin{multline}
     \Fscr_2(g) = \psmu{\mca g}{\Piv \mct (\Id-\Piv)g} = \psmu{(\Piv \mct)^{\star} \mca g}{(\Id-\Piv)g}\\ = -\psmu{\overline{\mct \Piv} \mca g}{(\Id - \Piv) g} \leq \normmu{(\Id-\Piv)g}^2 \eqsp. 
   \end{multline}

   Finally, using \Cref{as:DMSabstract}-\ref{item:DMS-macro}-\ref{item:DMS-RSandRT} we have that for any $g \in \msc \subset \domain(\mcl) \cap \domain(\mcl^{\star}) \cap \domain(\calT \Piv)$,
   \begin{align}
     \Fscr_3(g) &= \psmu{\mcab \mct \Piv g}{g} + \psmu{\mcab \mct (\Id-\Piv) g}{g} + \psmu{\mcab \mcs g}{g} \\
     & \leq - \lambda_x\normmu{\Piv g}^2 + R_0 \normmu{(\Id-\Piv)g}\normmu{\Piv g} \eqsp.
   \end{align}
 \end{proof}

\begin{proof}[Proof of Theorem \ref{thm:DMSmain}]
The first part of the proof follows along the same lines as \cite[Theorem 2.18]{GS2014}.
  Let $\ff \in \mrl^2_0(\mu)$ and $\varepsilon >0$. For ease of notation, set for any $t \geq 0$, $\ff_t = P_t \ff$. From the Dynkin formula \cite[Proposition 1.5]{ethier:kurtz:1986}, for any $t >0$ $f_t \in \domain(\generator)$  
 and ${\rmd}\ff_{t}/{\rmd}t=\mcl\ff_{t}$. Therefore, for any $t >0$,
 \begin{equation}
   \label{eq:derivation_proof_dms}
   -\frac {\rmd}{\rmd t} \entropyA_{\veps}(\ff_t) = - [\Fscr_1(\ff_t) + \varepsilon\defEns{\Fscr_2(\ff_t)+\Fscr_3(\ff_t)}] \eqsp,
 \end{equation}
 where $\{\Fscr_i \, : \, i \in \{1,2,3\}\}$ are  defined in \eqref{eq:def_Fscr}. Then by \Cref{lem:preli_proof_dms_bound_fscr}, we obtain that for any $t >0$,
 \begin{align}
 \nonumber
  -\frac {\rmd}{\rmd t} \entropyA_{\veps}(\ff_t) &\ge \lambda_v m_2^{\half} \normmu{ (\Id-\Piv) \ff_t }^2 \\
 \nonumber
 & \quad  + \varepsilon \left[ \lambda_x \normmu{ \Piv \ff_t }^2 - \normmu{ (\Id-\Piv) \ff_t }^2 - R_0 \normmu{ (\Id-\Piv) \ff_t } \normmu{ \Piv \ff_t } \right]\\
 \nonumber
   &=\begin{pmatrix} \normmu{ \Piv \ff_t} \\ \normmu{ (\Id-\Piv) \ff_t} \end{pmatrix}^\top 
 \begin{pmatrix} \varepsilon \lambda_x & - \varepsilon R_0/2 \\ - \varepsilon R_0/2 \quad   & \lambda_v m_2^{\half} - \varepsilon \end{pmatrix}
                                                                                       \begin{pmatrix} \normmu{ \Piv \ff_t} \\ \normmu{ (\Id-\Piv) \ff_t} \end{pmatrix} \\
   & \geq \Lambda_0(\varepsilon) \normmu{\ff_t}^2 \eqsp, 
\end{align}
where 
\begin{align}
  2 \Lambda_0(\varepsilon) &= \lambda_v m_2^{\half} -\varepsilon (1 - \lambda_x) \\
&  \qquad - \sqrt{(\lambda_v m_2^{\half} -\varepsilon (1-\lambda_x))^2 - [4\varepsilon\lambda_{x}(\lambda_{v}m_{2}^{\half}-\varepsilon)-\varepsilon^{2}R_{0}^{2}] } \eqsp,
\end{align}
is the smallest eigenvalue of the symmetric matrix, positive for $0\leq \varepsilon \leq 4 \lambda_x \lambda_v m_2^{\half} /(4\lambda_x+R_0^2)$
from \Cref{lem:derivativeLambda0} in \Cref{sec:optim-varepsilon} (as $\lambda_x \leq 1$ by \Cref{as:DMSabstract}-\ref{item:DMS-macro}). Using \eqref{eq:equinorm}, we get
\begin{equation}
 \label{eq:inequality of D}
 -\frac {\rmd}{\rmd t} \entropyA_{\veps}(\ff_t) \geq \frac{2\Lambda_0(\varepsilon)}{1+(m_2/2)^{-\half} \varepsilon} \entropyA_{\veps}(\ff_t)\eqsp.
\end{equation}
From Grönwall's lemma and \eqref{eq:equinorm}, we obtain for $0\leq \varepsilon \leq (m_2/2)^{\half} \wedge\{ 4 \lambda_x \lambda_v m_2^{\half} /(4\lambda_x+R_0^2)\}$, $\normmu{f_t} \leq C_0(\varepsilon) \rme^{-\alpha_0(\varepsilon)t} \normmu{f_0}, \text{ where }$
\begin{equation}
 \alpha_0(\varepsilon)= \frac{\Lambda_0(\varepsilon)}{1+(m_2/2)^{-\half} \varepsilon} \; \text{ and } \; C_0(\varepsilon) = \sqrt{\frac {1+( m_2/2)^{-\half}\varepsilon}{1-(m_2/2)^{-\half}\varepsilon}} \eqsp.
\end{equation}
 For notational simplicity we let $\epsilon = \varepsilon/(\lambda_v m_2^\half)$ and note that with the definitions in \eqref{eq:def_alpha_C_eps}-\eqref{eq:def_lambda_0}, for $\epsilon < 4 \lambda_x /(4\lambda_x+R_0^2)$, $\alpha(\epsilon)=\alpha_0(\varepsilon)>0$ and $\lambda_v m_2^\half \Lambda(\epsilon)=\Lambda_0(\varepsilon)>0$,  and for $\epsilon \leq (2^{\half} \lambda_v)^{-1}$ the two norms are equivalent and $A(\epsilon)=C_0(\varepsilon)$ is well defined. This concludes the proof of \ref{thm:DMSmain_item_a}.

From \Cref{prop:bound-on-Phi-epsilon-star} and associated notation in \Cref{sec:optim-varepsilon}, $\epsilon \mapsto \alpha(\epsilon)$ has a unique, but intractable, maximum, $\epsilon^{\star}\in (0, 4 \lambda_x /(4\lambda_x+R_0^2) )$. However from \Cref{lem:second-derivative-Lambda0}-\ref{lem:second-derivative-Lambda0_b} and \Cref{prop:bound-on-Phi-epsilon-star} the unique maximum $\epsilon_0 \in (\epsilon^{\star}, 4 \lambda_x /(4\lambda_x+R_0^2))$ of $\epsilon \mapsto \Lambda(\epsilon)$, defined by \eqref{eq:def_eta_star}, provides us with a tractable proxy such that $\alpha(\epsilon_0) < \alpha(\epsilon^{\star}) <3\alpha(\epsilon_0)$. In addition, since $\lambda_x \leq 1$ and for $2^\half R_0 \geq \lambda_v$ we get

\begin{equation}
	\epsilon_0 < \frac {(1+\lambda_x)} {(1+\lambda_x)^2 + R_0^2} \leq (2R_0)^{-1} \le (2^{\half}\lambda_v)^{-1} \eqsp,
\end{equation}
 which implies that $A(\epsilon_0)$ is well defined (and the two norms equivalent). The last statement follows from \Cref{lem:second-derivative-Lambda0}-\ref{lem:second-derivative-Lambda0_c} in \Cref{sec:optim-varepsilon}.

\end{proof}

The following lemma provides us with simple estimates of $\alpha(\epsilon_0)$ and $A(\epsilon_0)$ defined in \Cref{thm:DMSmain}.
 \begin{lemma}\label{lemma:behaviour_rate_R_0} Let $\epsilon \mapsto \alpha(\epsilon), A(\epsilon)$ and $\epsilon_0$ be as in \Cref{thm:DMSmain} and let $\lambda_x \in \ooint{0,1}$. Then  
 \begin{enumerate}

\item for any $R_0 \geq 4 + 12^{\half}$,
\begin{equation}
\label{eq:borne_eps_0_R_0}
  \lambda_x/(1 +R_0^{2}) \leq \epsilon_0 \leq 2/(4 +R_0^2) \leq 1/(4 R_0)\eqsp,
 \end{equation}
 
 \item for any $R_0 \geq (4+12^{\half}) \vee (\lambda_v/2^{\half})$,
 \begin{equation} 
A(\epsilon_0) \leq 3^\half \quad \text{and} \quad  \lambda_v \lambda_x m_2^\half \epsilon_0/8 \leq \alpha(\epsilon_0) \leq 4 \lambda_v \lambda_x m_2^\half \epsilon_0 \eqsp.
 \end{equation}
 \end{enumerate}
 \end{lemma}

 \begin{proof}
   The proof is postponed to \Cref{sec:proof:behaviour_R_0}.
 \end{proof}

\subsection{DMS for PDMP: generic results} \label{sec:DMSforPDMP}

\begin{proposition} \label{prop:adjoint-calL}
Assume that $\mcl_i$, $i \in \{1,2\}$, defined by \eqref{eq:generator pdmp} or \eqref{eq:generator_r_d_2}, with $\mcbb_k$ given in \eqref{eq:def-bounces}, satisfies \Cref{as:generator} with $\msc = \rmCb^2(\mse)$ together with  \Cref{as:U},  \Cref{as:Fk}, \Cref{as:intensities}, \Cref{as:radial}  \Cref{as:operator on velocities} and \Cref{as:refreshment}. Then the $\Lmu$-adjoint of $\calL_i$ for $i \in \{1,2\}$ defined by \eqref{eq:generator pdmp} or \eqref{eq:generator_r_d_2}  is given for any $\ff \in \mrCb^{2}(\mse)$ by

\begin{equation}
\calL^{\star}_i \ff = - v^\top \nax \ff + \updelta_{i,2} m_2 F_0^\top \nav \ff + \sum_{k=1}^K \varphi\big(-v^\top F_k\big)[ (\mcbb_k-\Id)\ff ]+ m_2^{\half}\lambdaref \calR_v \ff \eqsp.
\end{equation}

\end{proposition}
\begin{proof}
We only consider the case $i=2$ since the proof for $i=1$ follows along the same lines. In addition, since $\mcr_v$ is self-adjoint by \Cref{as:operator on velocities} and $\rmCb^2(\mse) \subset \domain(\calR_v)$, we can consider the case $\lambdaref(x) = 0$ for any $x \in \msx$. 
Based on \eqref{eq:generator pdmp}-\eqref{eq:generator_r_d_2}, using that for any $k \in \{1,\ldots,K\}$, $\mcbb_k$ is symmetric on $\mrl^2(\mu)$, for any $(x,v) \in \mse$, $\mcbb_k \lambda_k(x,v) = \lambda_k(x,-v)$ and by integration by part, for any $f,g \in \mrCb^2(\mse)$, we obtain 
\begin{align}
 & \psmu{g}{\mcl f} = \psmu{- v^\top \nax g + (v^\top \nax U)g + m_2 F_0^\top \nav g - (v^\top F_0)g }{f} \\
&  \phantom{\psmu{g}{\mcl f}} \qquad \qquad  + \psmu{ \textstyle\sum_{k=1}^K (\mcbb_k - \Id) [\lambda_k(x,v)g] }{f} \\
  & = \psmu{- v^\top \nax g + [v^\top (\nax U-F_0)] g + m_2 F_0^\top \nav g}{f} \\
  & \phantom{\psmu{g}{\mcl f}} \qquad \qquad + \psmu{\textstyle\sum_{k=1}^K\{\lambda_k(x, -v) \mcbb_k g - \lambda_k(x,v) g\} }{f} \\
  & = \psmu{\mcl_i^{\star} g}{f} + \psmu{ [v^\top (\nax U-F_0)] g  + g \textstyle\sum_{k=1}^K\{\lambda_k(x,-v) - \lambda_k(x,v)\}  }{f} \eqsp.
\end{align}
Using that $\sum_{k=0}^{K} F_k= \nax U$ by \Cref{as:Fk}-\ref{item:eq:decomposition U} and that $\lambda_k(x,v) - \lambda_k(x,-v) = \ps{v}{F_k(x)}$  for any $k \in \{1,\ldots,K\}$ and $(x,v) \in \mse$ by \Cref{as:intensities},  concludes the proof.

\end{proof}

The following provides expressions for the $\Lmu$-symmetric and $\Lmu$-anti-symmetric parts of $\calL$ for all the PDMP processes considered in this paper. Define $\lambda_k^{\even} : \mse \to \rset_+$ for any $(x,v) \in \mse$ and $k \in \{1,\ldots,K\}$ by
 \begin{equation}
   \label{eq:def_lambda_even}
   \lambda^{\even}_k(x,v) = \lambda_k(x,v) + \lambda_k(x,-v) \eqsp.
 \end{equation}

\begin{proposition}
 \label{prop:dms_pdmp_1} 
Assume that $\mcl_i$, $i \in \{1,2\}$, defined by \eqref{eq:generator pdmp} or \eqref{eq:generator_r_d_2}, with $\mcbb_k$ given in \eqref{eq:def-bounces}, satisfies \Cref{as:generator} with $\msc = \rmCb^2(\mse)$ together with  \Cref{as:U},  \Cref{as:Fk}, \Cref{as:intensities}, \Cref{as:radial}  \Cref{as:operator on velocities} and \Cref{as:refreshment}.
 Let $\mcs$ and $\mct_i$ be the symmetric and anti-symmetric parts of $\mcl_i$ respectively,  defined by \eqref{eq:def_calS_calT}.
\begin{enumerate}[wide, labelwidth=!, labelindent=0pt]
\item \label{item:sym skewsym}
    Then for any $\ff \in \mrCb^2(\mse)$, $\mct_i \ff = \tmct_if$ and $\mcs \ff= \tmcs f$ where $\tmct_i$ and $\tmcs$ are the operators defined for any $g \in \rmC^2_{\poly}(\mse)$ by 

    \begin{align}
      \label{eq:def_t_i}
      \tmct_i g &= v^\top \nax g - \delta_{i,2} m_2 F_0^\top \nav g + \frac 1 2 \sum_{k=1}^K (v^\top F_k) \, (\mcbb_k-\Id) g \eqsp, \\
            \label{eq:def_s}
	\tmcs g &= \frac 1 2 \sum_{k=1}^K  \lambda_k^{\even} \, (\mcbb_k-\Id) g + m_2^{\half} \lambdaref \calR_v g \eqsp.
\end{align}

\item \label{item:prop:dms_pdmp_1_b} 
  $\mcs$ satisfies    \Cref{as:DMSabstract}-\ref{item:DMS-kernel_calS}.
\item \label{item:prop:dms_pdmp_1_c} $\rmC^1_{\poly}(\mse) \subset \domain(\mct_i^\star) \cap \domain(\mcs^\star)$ and for any $f \in \rmC^1_{\poly}(\mse)$, $\mct_i^{\star}f = -\tmct_i f$ and $\mcs^{\star} f = \tmcs f$.  
\end{enumerate}
\end{proposition}
Note that the symmetric parts of $\mcl_i$ for $i \in \{1,2\}$ are the same and equal to $\mcs$. 
\begin{proof}
  \ref{item:sym skewsym} follows from \Cref{prop:adjoint-calL} and the definitions of $\calS$ and $\calT$ in \eqref{eq:def_calS_calT}.  \ref{item:prop:dms_pdmp_1_b} is a direct consequence of the first result and the definition of $(\mcs^{\star},\domain(\mcs^{\star}))$. Simple integration by parts and definitions of $(\mcs^{\star},\domain(\mcs^{\star}))$, $(\mct_i^{\star},\domain(\mct_i^{\star}))$ imply \ref{item:prop:dms_pdmp_1_c}. 
  
\end{proof}

 We define the directional derivative operator 
\begin{equation}
  \label{eq:direct_derivat}
\text{ for any }  \ff \in \domain(\calD) = \rmCb^{1}(\mse) \eqsp, \,
 \calD \ff(x,v) = v^{\top} \nabla_x \ff(x,v) \eqsp.
 \end{equation}
 The operators $(\calD, \rmCb^1(\mse))$ and  $(\calD \Piv, \rmCb^1(\mse))$ are densely defined on $\mrl^2(\mu)$ and closable. The proof is similar to that for the operator $\nabla_x$ and is omitted, see for example \cite[p. 88]{yoshida:1980}. 
 Note that by \eqref{eq:def_t_i}, a simple computation gives that for any $f \in \rmCb^2(\mse)$ and $i \in \{1,2\}$, since $\Pi_v f \in \rmCb^2(\mse)$,
 \begin{equation}
   \label{eq:mct_mcd_rmc_2}
   \mct_i \Pi_v f = \mcd \Pi_v f \eqsp.
 \end{equation}
 
 \begin{lemma}
   \label{lem:relation_mct_mcd}
Assume that $\mcl_i$, $i \in \{1,2\}$, defined by \eqref{eq:generator pdmp} or \eqref{eq:generator_r_d_2}, with $\mcbb_k$ given in \eqref{eq:def-bounces}, satisfies \Cref{as:generator} with $\msc = \rmCb^2(\mse)$ together with  \Cref{as:U},  \Cref{as:Fk}, \Cref{as:intensities}, \Cref{as:radial}  \Cref{as:operator on velocities} and \Cref{as:refreshment}.
Then, with $\mct_i$ the anti-symmetric part of $\mcl_i$  defined by \eqref{eq:def_calS_calT} and the operator $\mca_i$ defined by \eqref{eq:defcalA} relative to $\mct_i$, it holds:
  \begin{enumerate}
  \item    \label{lem:relation_mct_mcd_a}   $\mct_i$ satisfies \Cref{ass:stability_c_piv} and \Cref{as:DMSabstract}-\ref{item:DMS-proj} with $\core=\mrCb^2(\mse)$ and  $((\overline{\mct_i \Piv}),\domain(\overline{\mct_i \Piv})) =  ((\overline{\mcd \Piv}), \domain(\overline{\mcd \Piv}))$;
  \item    \label{lem:relation_mct_mcd_a_b} $\rmCb^2(\mse) \subset \domain((\mct_i \Piv)^{\star} \overline{\mct_i \Piv})$ and for any $f \in \rmCb^2(\mse)$, $  (\mct_i \Piv)^{\star} \mct_i \Piv f = m_2 \nabla_x^{\star} \nabla_x \Piv f$;
  \item \label{lem:relation_mct_mcd_b} $\{m_2 \Id +  (\mct_i \Pi_v)^{\star} \overline{\mct_i \Piv}\}^{-1} \Piv  =  m_2^{-1} \{\Id +  \nabla_x^{\star} \nabla_x \}^{-1} \Piv $ on $\Lmu$;
  \item \label{lem:relation_mct_mcd_c} $\mca^{\star}_i = m_2^{-1} (\overline{\mcd \Pi_v}) \{\Id +  \nabla_x^{\star} \nabla_x \}^{-1} \Piv$ and for any $f \in \rmCb^{2}(\mse)$, there exists a unique  function  $u \in \rmC_{\poly}^3(\msx)$, such that $m_2^{-1}\{\Id+ \nabla_x^{\star} \nabla_x \}^{-1} \Piv f = u$ and
    \begin{equation}
      \label{eq:mca_i_f_u}
      \mca^{\star}_i f = -\ps{v}{\nabla_x u} = -m_2^{-1} (\overline{\mcd \Pi_v}) \{\Id +  \nabla_x^{\star} \nabla_x \}^{-1} \Piv f \eqsp.
    \end{equation}
  \end{enumerate}
\end{lemma}
\begin{proof}
  \ref{lem:relation_mct_mcd_a} First note that $\rmCb^2(\mse)$ is a core for $(\overline{\mcd \Piv},\domain(\overline{\mcd \Piv}))$ since for any $f \in \rmCb^1(\mse)$, there exists a sequence of functions $(f_n)_{n \in \nset}$ such that for any $n \in \nset$, $f_n \in \rmCb^2(\mse)$, $\lim_{n \to \plusinfty} \normmu{f-f_n}= 0$ and $\lim_{n\to \plusinfty}\normmu{\nabla_xf - \nabla_x f_n} = 0$. 
Then the proof is completed upon using \eqref{eq:direct_derivat} and \eqref{eq:mct_mcd_rmc_2}. 

\ref{lem:relation_mct_mcd_a_b} By \eqref{eq:mct_mcd_rmc_2}, we have for any $f \in \rmCb^2(\mse)$, that $\mct_i \Piv f = \ps{v}{\nabla_x \Piv f}$. 
  It suffices then to verify  that with $g : (x,v) \mapsto \ps{v}{\nabla_x (\Piv f)(x)}$,  then $g  \in \domain((\mct_i \Piv )^{\star})$ and $(\mct_i \Piv)^{\star} g = m_2 \nabla^{\star}_x \nabla_x \Piv f$, \ie~for any $h \in \domain(\mct_i \Piv)$, we have $\psmu{\mct_i \Piv h}{g} = m_2 \psmu{h}{\nabla^{\star}_x \nabla_x \Piv f}$. But  $\domain(\mct_i \Piv) = \core =  \rmCb^2(\mse)$ by assumption and definition see \eqref{eq:def_calS_calT}. Then
using \eqref{eq:direct_derivat}, \eqref{eq:mct_mcd_rmc_2} and an integration by part we obtain for any $h \in \rmCb^2(\mse)$,
\begin{align}
  \label{eq:11}
  \psmu{\mct_i \Piv h}{g} &= \psmu{\mcd \Pi_v h}{g} = \int_{\mse} \{\ps{v}{ \nabla_x (\Pi_vh)(x)}\}\{\ps{v}{\nabla_x (\Piv f)(x)}\} \rmd \mu(x,v) \\
& = m_2 \int_{\msx} \ps{\{\nabla_x (\Pi_vh)(x)\}}{\nabla_x(\Piv f)(x)} \rmd \pi(x) \\
&  = m_2 \int_{\msx} (\Pi_vh)(x)\nabla_x^{\star} \nabla_x(\Piv f)(x) \rmd \pi(x) \\
& = m_2 \int_{\msx} \Piv \parentheseDeux{h\nabla_x^{\star} \nabla_x(\Piv f)}(x) \rmd \pi(x) \\
  &  = m_2 \int_{\mse}  h(x,v)\nabla_x^{\star} \nabla_x(\Piv f)(x) \rmd \mu(x,v) \eqsp,
\end{align}
where we have used the definition of $\Pi_v$ \eqref{eq:def_piv} in the last step.

  
  \ref{lem:relation_mct_mcd_b}
  Note that we only need to show that $\{ m_2 \Id + (\mct_i \Pi_v)^{\star} \overline{\mct_i \Piv}\}^{-1}$ and   $m_2^{-1} \{\Id + \nabla_x^{\star} \nabla_x \}^{-1}$   are equal on a dense subset of $\rml^2(\pi)$ since they are bounded. We now show that this statement is true  choosing the subset
  $ m_2 \{\Id + \nabla_x^{\star} \nabla_x      \}(\rmC^3_{\poly}(\msx))$.  First, for any $h \in \rmCb^2(\mse)$, we have using \ref{lem:relation_mct_mcd_a}, \ref{lem:relation_mct_mcd_a_b} and the definition \eqref{eq:direct_derivat} that
      \begin{equation}
        \label{eq:relation_mct_mcd_b_1}
        \{m_2 \Id +  (\mct_i \Pi_v)^{\star} \overline{\mct_i \Piv} \}  h=         \{m_2 \Id +  (\mct_i \Pi_v)^{\star} \mct_i \Piv \}  h=  m_2 \{\Id +  \nabla_x^{\star} \nabla_x \Piv \}   h \eqsp.
      \end{equation}
      Second, for any $g \in \rmC^3_{\poly}(\msx)$, there exists a
      sequence $(g_n)_{n \in \nset}$ such that for any $n \in \nset$,
      $g_n \in \rmCb^2(\msx)$, $(g_n)_{n\in\nset}$,
      $(\nabla_x g_n)_{n \in \nset}$ and
      $(\nabla_x^2 g_n)_{n \in \nset}$ converge in \-~$\mrl^2(\pi)$ to
      $g$, $\nabla_x g$ and $\nabla_x^2g$ respectively, which
      implies that the sequences 
      $\{ [m_2 \Id + (\mct_i \Pi_v)^{\star} \overline{\mct_i \Piv} ]
      g_n\}_{n \in \nset}$ and $\{ m_2 [\Id + (\nabla_x)^{\star} \nabla_x ] g_n\}_{n \in
        \nset}$ are $\rml^2(\pi)$ convergent. Therefore, since
      $ \{m_2 \Id +  (\mct_i \Pi_v)^{\star} \overline{\mct_i \Piv} \}$
      and $m_2 \{\Id +  \nabla_x^{\star} \nabla_x \}$ are clos\-ed, we
      get that $\rmC^3_{\poly}(\msx)$ is included in the domain of
      these two operators and \eqref{eq:relation_mct_mcd_b_1} holds for any $h \in \rmC^3_{\poly}(\msx)$. 
      \cite[Theorem 2]{pardoux2001} or \cite[Lemma
      17]{brosse:durmus:moulines:sabanis:2019}\footnote{Note that the
        result is stated for functions $f \in \rmC^3_{\poly}(\rset^d)$
but  the proof can be easily extended to
        $f \in \rmC^2_{\poly}(\msx)$} show that for any $f \in \mrCb^2(\msx)$,
      there exists $u \in \rmC^3_{\poly}(\msx)$ such that
      $m_2\{ \Id + \nabla_x^{\star} \nabla_x\} u = f$. Therefore,
      it holds that
      \[ \rmCb^2(\msx) \subset m_2 \{\Id +  \nabla_x^{\star} \nabla_x
      \}(\rmC^3_{\poly}(\msx)), \] so the subset  $ m_2 \{\Id +  \nabla_x^{\star} \nabla_x
      \}(\rmC^3_{\poly}(\msx)) $ is dense in $\mrl^2(\pi)$. In addition, since we have shown that the operators $\{ m_2 \Id + (\mct_i \Piv)^{\star}\overline{\mct_i\Pi_v} \}$ and   $m_2\{ \Id + \nabla_x^{\star} \nabla_x\}$ coincide on $\rmC_{\poly}^3(\msx)$, 
      $\{ m_2 \Id + (\mct_i \Piv)^{\star}\overline{\mct_i\Pi_v} \}^{-1}$ and   $m_2^{-1}\{ \Id + \nabla_x^{\star} \nabla_x\}^{-1}$ coincide on $m_2 \{\Id +  \nabla_x^{\star} \nabla_x
      \}(\rmC^3_{\poly}(\msx)) $.

      \ref{lem:relation_mct_mcd_c} As $\mca_i$ is bounded, it is
      sufficient to show that the operators $\mca_i^{\star}$ and
      $m_2^{-1} (\overline{\mcd \Pi_v}) \{\Id + \nabla_x^{\star}
      \nabla_x \}^{-1} \Piv$ coincide on a dense subset of
      $\rml^2(\mu)$. First, for all $f,g \in \rmCb^2(\mse)$, we get
      that $\psmu{\mca_i g}{f}= \psmu{\Pi_v\mca_i g}{f}$ by
      \Cref{lem:bounded_A}-\ref{lem:bounded_A_item_a}. Now using the
      definition of $\mca_i$ \eqref{eq:defcalA}, that $\Piv$ and
      $\{m_2 \Id +  (\mct_i \Pi_v)^{\star} \overline{\mct_i \Piv}\}^{-1}$ are bounded and      self-adjoint, since $\Piv$ is an orthogonal projection and  by \Cref{prop:abstract bound}-\ref{prop:abstract bound_a}-\ref{lem:relation_mct_mcd_b},  we
      get for any $f \in \rmCb^2(\mse)$,
      \begin{align}
       \psmu{\mca_i g}{f} &= m_2^{-1} \psmu{(-\Piv \mct_i)^{\star}g}{\{\Id + \nabla_x^{\star} \nabla_x \}^{-1} \Pi_v f} \\
&= m_2^{-1}  \psmu{ \mct_i \Piv g}{\{\Id + \nabla_x^{\star} \nabla_x \}^{-1} \Pi_v f}\eqsp,
      \end{align}
      where we have used \Cref{lem:grothaus_lemme_2_3}-\ref{lem:grothaus_lemme_2_3_a} for the last equality and  $\domain(\mct_i) = \rmCb^2(\mse)$.
            \cite[Theorem 2]{pardoux2001} or \cite[Lemma
      17]{brosse:durmus:moulines:sabanis:2019} show that there exists $u \in \rmC^3_{\poly}(\msx)$ satisfying 
      $m_2\{ \Id + \nabla_x^{\star} \nabla_x\} u = \Pi_v f$ and therefore, we get that
      \begin{equation}
          \psmu{\mca_i g }{f} =   \psmu{ \mct_i \Piv g}{u} = -\psmu{g}{ \ps{v}{\nabla_x u}} \eqsp,
        \end{equation}
        using an integration by part for the last identity. This
        result shows that for any $f \in \rmCb^2(\mse)$, we have that
        $\mca_i^{\star} f= - \ps{v}{\nabla_x u}$. In addition, for any
        $g \in \rmC_{\poly}^1(\mse)$, there exists a sequence
        $(f_n)_{n \in \nset}$ such that $f_n \in \rmCb^1(\mse)$ and $\lim_{n\to \plusinfty} \normmu{g-f_n}=0$, $\lim_{n \to \plusinfty} \normmu{\nabla_xg - \nabla_x f_n} = 0$. Therefore we get that $\rmC_{\poly}^1(\mse) \subset \domain(\overline{\mcd \Piv})$ and for any $g \in \rmC^1_{\poly}(\mse)$, $\overline{\mcd \Piv} g(x,v)= \ps{v}\nabla_xg(x,v)$ for any $(x,v) \in \mse$. Therefore, we get the desired conclusion that 
        $ \mca_i^{\star} f= - \ps{v}{\nabla_x u}=  - m_2^{-1} (\overline{\mcd \Pi_v}) \{\Id +  \nabla_x^{\star} \nabla_x \}^{-1} \Piv f $, which completes the proof.
\end{proof}
Establishing \Cref{as:DMSabstract}-\ref{item:DMS-micro} (referred to as microscopic coercivity in \cite{Dolbeault15}) for the processes considered is fairly straightforward in the present framework.
\begin{proposition}
  \label{prop:microscopic coercivity}
Assume that $\mcl_i$, $i \in \{1,2\}$ given by \eqref{eq:generator pdmp} or \eqref{eq:generator_r_d_2}, where $\mcbb_k$ is defined in \eqref{eq:def-bounces} satisfies \Cref{as:generator} with $\msc = \rmCb^2(\mse)$. Assume in addition that  \Cref{as:U},  \Cref{as:Fk}, \Cref{as:intensities}, \Cref{as:radial}  \Cref{as:operator on velocities} and \Cref{as:refreshment} hold.
Let $\mcs$ be the symmetric part of $\mcl_i$  defined by \eqref{eq:def_calS_calT}. Then \Cref{as:DMSabstract}-\ref{item:DMS-micro} is satisfied with $\lambda_v = \ulambda$ and $\core = \mrCb^2(\mse)$.
\end{proposition}

\begin{proof}

From \Cref{as:operator on velocities}-\ref{item:refreshment} and \Cref{as:refreshment}, it holds that for any $f \in \rmCb^2(\mse)$, we have 
\begin{equation}
\label{eq:microscopic coercivity_proof_1}  
	-\psmu{\lambdaref m_2^\half \calR_vf}{f} \geq \ulambda m_2^{\half} \psmu{(\Id - \Piv)f}{f} \eqsp.
 \end{equation}
 In addition,  any $f \in \rmCb^2(\mse)$ satisfies $\max_{k \in \{1,\ldots,K\}} \normmu{v^{\top} F_k f} < \plusinfty$ by \Cref{as:U},\eqref{eq:borne_moment} and \eqref{eq:def_a_k}, then by \Cref{as:intensities} for any $k \in \{1, \ldots,K\}$, $\sup_{k \in \{1,\ldots,K\}}\normmu{ \lambda^{\even}_k f} < \plusinfty$. Therefore, using the Cauchy-Schwarz inequality, that $\mcbb_k$ is a symmetric involution on $\rmL^2(\mu)$ by \Cref{as:radial}, and  $\mcbb_k \lambda^{\even}_k = \lambda^{\even}_k$ by definition \eqref{eq:def_lambda_even}, we obtain for any $k \in \{1,\ldots,K\}$ and $f \in \rmCb^2(\mse)$,
\begin{align}
\psmu{\lambda_k^{\even} \, \mcbb_k f}{f} &\leq \normmuLine{(\lambda_k^{\even})^\half f}\normmuLine{(\lambda_k^{\even})^\half \calB_k f} = \normmuLine{(\lambda_k^{\even})^{\half} f}^2 \eqsp. 
\end{align}
As a result, we deduce $\psmu{\lambda_k^{\even} \, (\Id - \mcbb_k)f}{f} \geq 0$. Combining this result and \eqref{eq:microscopic coercivity_proof_1} in the expression for $\mcs$ given in \eqref{eq:def_s} in \Cref{prop:dms_pdmp_1} completes the proof. 
\end{proof}

The following lemma establishes equivalence between \Cref{as:DMSabstract}-\ref{item:DMS-macro} and the Poincar\'e inequality \Cref{as:U} , which allows one to refer to the expansive body of literature on the topic and implies dependence on the properties of the potential $U$ only.
 \begin{proposition} \label{lem:linkPoincare-DMSabstract}
   Assume that $\mcl_i$, $i \in \{1,2\}$ given by \eqref{eq:generator pdmp} or \eqref{eq:generator_r_d_2}, where $\mcbb_k$ as in \eqref{eq:def-bounces} satisfies \Cref{as:generator} with $\msc = \rmCb^2(\mse)$. Assume in addition that  \Cref{as:U},  \Cref{as:Fk}, \Cref{as:intensities}, \Cref{as:radial}  \Cref{as:operator on velocities} and \Cref{as:refreshment} hold.
Let $\mct_i$ be the anti-symmetric part of $\mcl_i$ defined by \eqref{eq:def_calS_calT} and $\mca_i$ be defined by \eqref{eq:defcalA} relative to $\mct_i$. Then,
 \Cref{as:DMSabstract}-\ref{item:DMS-macro}, \ie~\eqref{eq:DMS_macro}, holds with
 \begin{equation}
 \label{eq:def_lambda_x}
\lambda_x = \Cp/(1+\Cp) \eqsp. 
\end{equation}
\end{proposition}

\begin{proof}

  From the assumed Poincar\'e inequality \eqref{eq:poincare_assumption} we have for any $\ff \in \rmCb^1(\mse)$
\begin{equation}
\normmu{m_{2}^{-\half}\calD\Piv \ff}^{2}=\normmu{\nax\Piv \ff}^{2}\geq \Cp\normmu{\Piv \ff}^{2} \eqsp.
\end{equation}
Then, by definition of $\overline{\mcd \Piv}$ this inequality holds also for any $f \in \domain(\overline{\mcd \Piv})$ replacing $\calD\Piv \ff$ by $\overline{\calD\Piv} \ff$.  Therefore, we obtain since $(\mcd \Piv)^{\star \star} = \overline{\mcd \Piv}$ that 
for any $\ff \in \domain((\mcd \Piv)^{\star}\overline{\mcd \Piv})$,
\begin{equation}
 \label{eq:lem_proof_poincare_dms_as}
\psmu{ \ff }{m_{2}^{-1}(\calD\Piv)^{\star}\overline{\calD\Piv} \ff }\geq\Cp\normmu{\Piv \ff}^{2} \eqsp.
\end{equation}
In addition by \cite[Theorem~5.1.9]{pedersen1995analysis}, $(\overlineb{\mcd\Piv})^{\star} \overline{\mcd\Piv}$ is a self-adjoint operator. These results and \eqref{eq:lem_proof_poincare_dms_as} imply that $\spec(m_{2}^{-1}(\overlineb{\mcd\Piv})^{\star} \overline{\mcd\Piv})\subseteq \coint{\Cp,\infty}$
by  \cite[Theorem~4.3.1]{davies:1995}.

On the other hand, since by \Cref{lem:relation_mct_mcd}-\ref{lem:relation_mct_mcd_a}, $\overline{\mcd \Piv} = \overline{\overlineb{\mct_i \Piv}}$, we have  $(\mcd \Piv)^{\star} = (\overlineb{\mct_i \Piv})^{\star}$ and 
\begin{equation}
  \mca_i = -\big(m_{2}\Id+(\calD\Piv)^{\star}\overline{\calD\Piv}\big)^{-1}(\calD\Piv)^{\star} \eqsp.
\end{equation}
Therefore, for any $\ff \in \domain((\mcd \Piv)^{\star}\overline{\mcd \Piv})$,
\begin{align}
  - \mcab_i \, \overline{\calD\Piv} \ff &= - \calA_i \overline{ \calD\Piv} \ff =\big(m_{2}\Id+(\calD\Piv)^{\star}\overline{\calD\Piv}\big)^{-1}(\calD\Piv)^{\star}\overline{\calD\Piv} \ff \\
  &= \Phi\big(m_{2}^{-1}(\calD\Piv)^{\star}\overline{\calD\Piv}\big) \ff \eqsp,
\end{align}
where $\Phi(z)= z / (1+z)$. Since $\domain((\mcd \Piv)^{\star}\overline{\mcd \Piv})$ is a core for $\overline{\mcd \Piv}$  by \cite[Theorem 5.1.9.]{pedersen1995analysis}, from the spectral mapping theorem \cite[Theorem 2.5.1, Corollary 2.5.4]{davies:1995}, and the fact that
$\Phi\colon \coint{0,\infty}\rightarrow \ccint{0,1}$ is non-decreasing, we get that $- \bmca_i \,  \overline{ \calD\Piv}$ can be extended on $\mrl^2(\mu)$ as a self-adjoint bounded operator $\mce$ and  ${\spec}(\mce)\subseteq \coint{\Phi(\Cp),1}$.

Finally, from the fact that $\Piv$ is a projector, we deduce from \Cref{lem:bounded_A}-\ref{lem:bounded_A_item_a} that $ - \bar{\mca}_i \mct_i \Piv \ff = -\Piv \bar{\mca}_i \, \overline{\mcd \Piv} \Piv  \ff=  \Piv\mce\Piv \ff$ for any $f \in \rmCb^2(\mse) \subset \domain(\overline{\mcd \Piv})$  and therefore, we get that for any $f \in \rmCb^2(\mse)$
\begin{equation}
- \psmu{\Piv \ff}{\bmca\mct_i\Piv \ff } =\psmu{\Piv \ff}{\mce \Piv \ff } \geq\frac{\Cp}{1+\Cp}\normmu{\Piv \ff }^{2}=\lambda_{x}\normmu{\Piv \ff}^{2} \eqsp,
\end{equation}
which concludes the proof.
\end{proof}


\Cref{as:DMSabstract}-\ref{item:DMS-RSandRT} is usually a more involved condition to check. For $\ff \in \Lmu$ denote by
\begin{equation}
 \label{eq:def_u}
 u_f = m_2^{-1} ( \Id +\nax^{\star} \nax)^{-1} \Piv f  \eqsp.
\end{equation}
In the scenarios considered here, condition \Cref{as:DMSabstract}-\ref{item:DMS-RSandRT} relies on estimates of $\normmu{u_f}$, $\normmu{\nabla_x u_f}$ and $\normmu{\nabla_x^2 u_f}$ which are obtained by noticing that by definition $u_f$ is solution of the following partial differential equation
\begin{equation}
m_2 ( \Id+\nax^{\star} \nax)u_f=\Piv \ff \eqsp.
\end{equation}
In the next section, we show how general, but potentially rough, estimates can be obtained, while in \Cref{ss:scaling-ZZ} we show how tighter bounds can be obtained in specific scenarios where we can take advantage of the structure at hand, in particular when interested in the scaling properties of the algorithm with $d$.


\subsection{Computation of $R_{0}$ in the general setting}
\label{app:bound RT RS}
In all this section, we consider $u_f$ defined for any $f \in \mrl^2(\mu)$ by \eqref{eq:def_u}. Recall that from \Cref{lem:relation_mct_mcd}-\ref{lem:relation_mct_mcd_c}, if $f \in \mrCb^2(\mse)$ then $u_f \in \rmC^3_{\poly}(\rset^d)$ and satisfies \eqref{eq:mca_i_f_u}.

\begin{lemma} \label{lemma:RS bounded}
Assume that $\mcl_i$, $i \in \{1,2\}$ given by \eqref{eq:generator pdmp} or \eqref{eq:generator_r_d_2}, where $\mcbb_k$ is given in \eqref{eq:def-bounces}, satisfies \Cref{as:generator} with $\msc = \rmCb^2(\mse)$. Assume in addition that  \Cref{as:U},  \Cref{as:Fk}, \Cref{as:intensities}, \Cref{as:radial}  \Cref{as:operator on velocities} and \Cref{as:refreshment} hold.
Let  $\mcs$ be the symmetric part of  $\calL_{i}$  defined by \eqref{eq:def_calS_calT} and the operator $\mca_i$ defined by \eqref{eq:defcalA} relative to $\mct_i$. 
\begin{enumerate}[label=(\alph*)]
\item \label{lem:bound_a_s_item_a}
For any $f \in \rmCb^2(\mse)$, 
\[
\vert\psmu{\overline{\mca}_i\mcs(\Id-\Piv)f}{ f}\vert\leq\normmuLine{(\Id-\Piv)f}\normmuLine{(\Id-\Piv)\tmcs \mca_i^{\star} f} \eqsp,
\]
where $\tmcs$ is given by \eqref{eq:def_s}. 
\item \label{lem:bound_a_s_item_b} 
 For any $f \in \rmCb^2(\mse)$,
 \begin{equation}
 \label{eq:lem:bound_a_s_item_b_1}
 \normmuLine{(\Id-\Piv)\tmcs \mca_i^{\star} f}=\normmuLine{\bfG^{\top}\nax u_f} \eqsp,
 \end{equation}
with $\bfG$ given for any $(x,v) \in \mse$ by

\begin{equation}
 \label{eq:def_bfG}
 \bfG(x,v)=\sum_{k=1}^{K}\lambda_k^{\even}(x,v) \big(\rmn_{k}^{\top}(x)v\big) \rmn_{k}+m_2^{\half}\lambdaref(x) v \eqsp,
\end{equation}

and $u_f,\{\lambda_k^{\even} : \mse \to \rset_+ \, : \, k \in \{1,\ldots,K\}\} $ are  defined by \eqref{eq:def_u} and \eqref{eq:def_lambda_even} respectively.
In addition

\begin{align}
 \normmuLine{\bfG^{\top}\nabla_{x}u_f}&\leq   m_2 \big( \normmuLine{ \lambdaref \nax u_f } +  c_\vphi K \normmuLine{\nax u_f}\big)
  \\
   \label{eq:lem_bound_a_019}
&\qquad \qquad   + C_\vphi \sqrt{2 m_{2,2} + 3(m_4-m_{2,2})_+} \sum_{k=1}^K \normmuLine{F_k^\top \nax u_f }  \eqsp.
\end{align}

\end{enumerate}
\end{lemma}

\begin{proof}
 We only consider the case $i=2$ since the case $i=1$ is obtained by taking $F_0=0$.

 \ref{lem:bound_a_s_item_a} By  \Cref{lem:bounded_A}-\ref{lem:bounded_A_item_a}, $\bmca_i$ is a bounded operator. Therefore, we have for any  $f \in \rmCb^2(\mse)$ that $\psmuLigne{\overline{\mca}_i \mcs(\Id -\Piv) f}{f} = \psmuLigne{ \mcs(\Id -\Piv) f}{\mca_i^\star f}$. Then, by  \Cref{lem:relation_mct_mcd}-\ref{lem:relation_mct_mcd_c}, we have that $\mca_i^\star f = -v^\top \nabla_x u_f$, with $u_f \in \rmC_{\poly}^3(\mse)$. This result,    \Cref{prop:dms_pdmp_1}-\ref{item:prop:dms_pdmp_1_c}, and the fact that $\Id-\Piv$ is an orthogonal projector imply that
 \begin{equation}
\psmu{\overline{\mca}_i\mcs(\Id-\Piv)f}{ f} =\psmu{(\Id-\Piv)f}{(\Id-\Piv) \tmcs \mca_i^{\star} f} \eqsp. 
\end{equation}

The proof is completed upon using the Cauchy-Schwarz inequality.

\ref{lem:bound_a_s_item_b}  Notice that

 \begin{align}
 \nonumber
 \tmcs \mca_2^{\star} f &= - \left( \frac 1 2 \sum_{k=1}^K  \lambda_k^{\even} (\mcbb_k - \Id) + m_2^{\half} \lambdaref \calR_v \right) v^\top \nax u_f\\
 \label{eq:bound_a_s_2}
 & =  \sum_{k=1}^K \lambda_k^{\even} \, (v^\top \rmn_{k}) (\rmn_{k}^\top \nax u_f) + m_2^{\half} \lambdaref v^\top \nax u_f 
 = \bfG^{\top}\nax u_f\eqsp,
 \end{align}
 where we have used \Cref{as:operator on velocities}-\ref{item:refreshment} for the last equality. Combining \eqref{eq:bound_a_s_2} and the fact  that $\Piv \tmcs \mca_2^{\star} f = 0$ completes the proof of \eqref{eq:lem:bound_a_s_item_b_1}.

We now show  \eqref{eq:lem_bound_a_019} for any $f \in \rmCb^2(\mse)$. But it  is a direct consequence of the triangle inequality, the definition of $\{\lambda^{\even}_k : \mse \to \rset_+ \, ; \, k \in \{1,\ldots,K\}\}$ given in \eqref{eq:def_lambda_even}, \Cref{as:intensities}, the Cauchy-Schwarz inequality, \Cref{lemma:velocities norm bound} and the identity  $F_k = \rmn_k \abs{F_k}$ for any $k \in \{1,\ldots,K\}$:
 \begin{equation}
 \begin{aligned}
   \normmuLine{ \mcs \mca_2^{\star} f } &\le  m_2^{\half}\normmuLine{ \lambdaref v^\top \nax u_f }\\
   & \qquad + \sum_{k=1}^K \defEns{C_\vphi\normmuLine{ (v^\top \rmn_{k})^2 \, F_k^\top \nax u_f } +c_\vphi m_2^\half \normmuLine{ (v^\top \rmn_{k}) \, \rmn_k^\top \nax u_f } }\\
   &=  m_2\normmuLine{ \lambdaref \nax u_f } +m_2 c_\vphi K \normmuLine{\nax u_f}\\
   & \qquad + C_\vphi \sqrt{2 m_{2,2} + 3(m_4-m_{2,2})_+} \sum_{k=1}^K \normmuLine{F_k^\top \nax u_f }  \eqsp.
 \end{aligned}
 \end{equation}
\end{proof}

\begin{lemma} \label{lemma:RT bounded}
Assume that $\mcl_i$, $i \in \{1,2\}$ given by \eqref{eq:generator pdmp} or \eqref{eq:generator_r_d_2}, where $\mcbb_k$ is given in \eqref{eq:def-bounces}, satisfies \Cref{as:generator} with $\msc = \rmCb^2(\mse)$. Assume in addition that  \Cref{as:U},  \Cref{as:Fk}, \Cref{as:intensities}, \Cref{as:radial}  \Cref{as:operator on velocities} and \Cref{as:refreshment} hold.
Let  $\mct_i$ be the anti-symmetric part of  $\calL_{i}$  defined by \eqref{eq:def_calS_calT} and the operator $\mca_i$ defined by \eqref{eq:defcalA} relative to $\mct_i$. Then,
\begin{enumerate}[label=(\alph*)]
\item \label{lem:bound_a_t_item_a} For any $f \in \rmCb^2(\mse)$, we get
 \[
\vert\psmu{\overline{\mca}_i\mct_i(\Id-\Piv)f}{ f}\vert\leq\normmuLine{(\Id-\Piv)f}\normmuLine{(\Id-\Piv)\tmct_i \mca_i^{\star} f} \eqsp,
\]
where $\tmct_i$ is given in \eqref{eq:def_t_i}. 
\item \label{lem:bound_a_t_item_b} For any $f\in \rmCb^2(\mse)$
 \begin{equation}
 \label{eq:lem:bound_a_t_item_b}
\normmuLine{(\Id-\Piv)\tmct_i \mca_i^{\star} f}= 2m_{2,2}\normmuLine{\bfM}^{2}+3(m_{4}-m_{2,2})\normmuLine{{\rm diag}(\bfM)}^{2} \eqsp,
 \end{equation}
with
\begin{equation}
 \label{eq:def_bfM}
 \bfM =\nabla_{x}^{2}u_f +\sum_{k=1}^{K} (F_{k}^{\top} \nabla_{x}u_f) \rmn_{k}\rmn_{k}^{\top} \eqsp, 
\end{equation}
and $u_f$ defined by \eqref{eq:def_u}. 
\end{enumerate}
\end{lemma}

\begin{remark}
\label{rem:RT bounded}
 A general, but potentially rough, bound on the right hand side of \eqref{eq:lem:bound_a_t_item_b} can be obtained as follows. From the fact that $\normmu{{\rm diag}(\bfM)}\leq\normmu{\bfM}$, it holds that 
\[
\normmuLine{(\Id-\Piv)\tmct_i \mca_i^{\star} f}\leq \sqrt{2m_{2,2}+3(m_{4}-m_{2,2})_{+}}\normmuLine{\bfM}
\]
where from the triangle inequality and the property $\vert \rmn_{k}(x)\rmn_{k}(x)^{\top} \vert=1$
\[
\normmuLine{\bfM}\leq\normmuLine{\nabla_{x}^{2}u_f}+\sum_{k=1}^{K}\normmuLine{F_{k}^{\mathsf{\top}}\nabla_{x}u_f} \eqsp.
\]
\end{remark}

\begin{remark} Specific scenarios lead to simplifications of these bounds and the bounds in Lemma \ref{lemma:ZZ RS bounded}:
\begin{enumerate}
\item from \Cref{lem:moments-spherically-symmetric} in \Cref{app:radial}, for radial distributions $m_4=m_{2,2}$ leading to a simplification of this bound,
\item further if $\nu$ is the centred normal distribution of covariance $m_2 \Idd$, then $m_{2,2}=m_2^2$, leading to further simplifications,
\item if $K=0$, and hence $F_0 = \nax U$, the scenario considered by \cite{Dolbeault15}, then one finds that the bound depends on $\normmu{\nax^2 u_f}$ only.
\end{enumerate}
\end{remark}

\begin{proof}
 We proceed as in the proof of \Cref{lemma:RS bounded}. 
 We only consider the case $i=2$ since the case $i=1$ is obtained by taking $F_0=0$.

 \ref{lem:bound_a_s_item_a} By  \Cref{lem:bounded_A}-\ref{lem:bounded_A_item_a}, $\bmca_2$ is a bounded operator. Therefore, we have for any  $f \in \rmCb^2(\mse)$ that $\psmuLigne{\overline{\mca}_2 \mct_2(\Id -\Piv) f}{f} = \psmuLigne{ \mct_2(\Id -\Piv) f}{\mca_2^\star f}$. Then, by  \Cref{lem:relation_mct_mcd}-\ref{lem:relation_mct_mcd_c}, we have that $\mca_2^\star f = -v^\top \nabla_x u_f$, with $u_f \in \rmC_{\poly}^3(\mse)$. This result,    \Cref{prop:dms_pdmp_1}-\ref{item:prop:dms_pdmp_1_c}, the fact that $\Id-\Piv$ is an orthogonal projector and $F_{k}=\rmn_{k}\vert F_{k}\vert$, imply that for any
 \begin{equation}
\psmu{\overline{\mca}_2\mct_2(\Id-\Piv)f}{ f} =-\psmu{(\Id-\Piv)f}{(\Id-\Piv) \tmct_2 \mca_2^{\star} f} \eqsp, 
\end{equation}
The proof is completed upon using the Cauchy-Schwarz inequality. 

\ref{lem:bound_a_t_item_b} Notice that for any $(x,v) \in \mse$,
\begin{align}
  -  \tmct_2\mca_2^{\star} f(x,v) & =v^{\top}\nabla_{x}^{2}u_f(x)v-m_{2}F_{0}^{\top}(x)\nabla_{x}u_f(x)\\
  & \qquad -\sum_{k=1}^{K}(v^{\top}F_{k}(x))\big(\rmn_{k}(x)\rmn_{k}(x)^{\top}v\big)^{\top}\nabla_{x}u_f(x)\\
  \label{eq:bounded_RT_1_general}
 & =v^{\top}\bfM(x) v-m_{2}F_{0}^{\top}(x)\nabla_{x}u_f(x) \eqsp.
\end{align} 
By \eqref{eq:bounded_RT_1_general}, we obtain that for any $f \in \rmCb^2(\mse)$,
$(x,v) \in \mse$, 
\begin{equation}
  - (\Id-\Piv) \tmct_2\mca_i^{\star} f(x,v) = v^{\top}\bfM(x) v - m_2 \trace(\bfM(x)) \eqsp.
\end{equation}
Combining this result and \Cref{lemma:velocities norm bound}, 
we deduce
\begin{align}
\normmuLine{ (\Id-\Piv) \tmct_2 \mca_i^{\star} f}^{2} & =2m_{2,2}\normmuLine{\bfM}^{2}+3(m_{4}-m_{2,2})\normmuLine{{\rm diag}(\bfM)}^{2}\\
 & \leq\big[2m_{2,2}+3(m_{4}-m_{2,2})_+\big]\normmuLine{\bfM}^{2} \eqsp,
\end{align}
which completes the proof.
\end{proof}

\begin{remark}
\label{rem:RS bounded}
 Combining \Cref{coro:bound regularization} and \Cref{coro:bound Fk} in \Cref{sec:ellipt-regul-estim}, by definition of $u_f$ in \eqref{eq:def_u} and using \Cref{as:refreshment}, we obtain that
\begin{align}
m_2 \normmuLine{\nax u_f} &\leq 2^{-\half}\normmuLine{ \Piv f},\\
	 \sum_{k=1}^K \normmuLine{ F_k^\top \nax u_f } &\le \frac {2^{\half}\kappavara}{m_2 \kappavarb} \sum_{k=1}^K a_k \normmuLine{ \Piv f } \eqsp, \\
	 m_2 \normmuLine{\lambdaref \nax u_f } & \le \ulambda \defEns{ 2^{-\half}+ \frac {2^{\half}c_{\lambda}\kappavara } { \kappavarb}} \normmuLine{ \Piv f } \eqsp.
\end{align}

\end{remark}


%
%
%
%

      \section{Postponed proofs}
      \label{sec:postponed-proofs}
      \subsection{Proof of Theorem \ref{thm:hypocoercivity}} \label{ss:proof-main-result}

In this section we prove that \Cref{ass:stability_c_piv} and  \Cref{as:DMSabstract} holds for the dynamics described in \Cref{s:main results} in order to obtain \Cref{thm:hypocoercivity} as a consequence of the abstract \Cref{thm:DMSmain}.
 Under the assumptions of the theorem, we can set $\msc$ to be
 $\mrCb^2(\mse)$. \Cref{ass:stability_c_piv} and \Cref{as:DMSabstract}-\ref{item:DMS-proj} hold by \Cref{lem:relation_mct_mcd}-\ref{lem:relation_mct_mcd_a}.  \Cref{as:DMSabstract}-\ref{item:DMS-micro}
 follows from \Cref{prop:microscopic coercivity} with
 $\lambda_v = \ulambda$. \Cref{as:DMSabstract}-\ref{item:DMS-macro}
 follows from \Cref{lem:linkPoincare-DMSabstract} 
 with $\lambda_x =
 \Cp/(1+\Cp)$. \Cref{as:DMSabstract}-\ref{item:DMS-kernel_calS} follows from \Cref{prop:dms_pdmp_1}-\ref{item:prop:dms_pdmp_1_b}.  
 We are left with checking
 \Cref{as:DMSabstract}-\ref{item:DMS-RSandRT}.
By  \Cref{lemma:RS
 bounded}-\ref{lem:bound_a_t_item_b}, \Cref{lemma:RT bounded}-\ref{lem:bound_a_t_item_b}, \Cref{rem:RT bounded},  we get setting $m = \sqrt{2m_{2,2}+3(m_{4}-m_{2,2})_{+}}$, for any $f \in \rmCb^2(\mse)$ that
\begin{align}
&\normmu{\tmcs \mca_i^{\star} f}  + \normmu{(\Id-\Piv) \tmct_i \mca_i^{\star} f }    \leq
                m \defEns{\normmuLine{\nabla_{x}^{2}u_f}+(1+C_{\vphi}) \sum_{k=1}^{K}\normmuLine{F_{k}^{\mathsf{\top}}\nabla_{x}u_f}} \\
  &\phantom{\normmu{\tmcs \mca_i^{\star} f}  + \normmu{(\Id-\Piv) \tmct_i \mca_i^{\star} f }    aaaa}+ 	 m_2 \normmuLine{\lambdaref \nax u_f } + m_2 c_\vphi K \normmuLine{\nax u_f} \\
 &    \leq \parentheseDeux{\frac{m}{m_2}\defEns{\frac {2^{1/2}(1+C_{\vphi})\kappavara}{\kappavarb} \sum_{k=1}^K a_k + \kappavara } + \frac {\ulambda}{2^\half} \left\{ 1 + \frac{2 c_\lambda \kappavara}{\kappavarb }  \right\} + \frac{c_\vphi K}{2^\half} } \normmu{\Piv f} \eqsp,  
\end{align}
where we have used that 
 $\normmu{\nax^2 u_f} \leq m_2^{-1} \kappavara \normmu{ \Piv f } $ by 
 \Cref{prop:bound hessian} in \Cref{sec:ellipt-regul-estim} and 
 \Cref{rem:RS bounded}, with 
 $\kappavara$ and $\kappavarb$ given in \eqref{eq:def_kappa_a} and \eqref{eq:def_kappab} respectively. The proof of \Cref{as:DMSabstract}-\ref{item:DMS-RSandRT} is then completed using \Cref{lemma:RT bounded}-\ref{lem:bound_a_t_item_a} and \Cref{lemma:RS bounded}-\ref{lem:bound_a_s_item_a}.

\subsection{Proof of \Cref{lemma:behaviour_rate_R_0}}
\label{sec:proof:behaviour_R_0}
\begin{proof}[Proof of \Cref{lemma:behaviour_rate_R_0}]
 Fix $\lambda_x \in \ooint{0,1}$.
\begin{enumerate}[wide, labelwidth=!, labelindent=0pt] \item Using that $t \mapsto (1+t)/\big[(1+t)^2+R_0^2\big]$  is nondecreasing on $\ooint{0,1}$ since $R_0 \geq 4$, we obtain that for any $R_0 \geq 4+2 \sqrt{3}$, \eqref{eq:borne_eps_0_R_0} is satisfied.
\item Since for any $a >0$, $s \mapsto (s+a)/(s-a)$ for $s>a$ is nonincreasing, we deduce from above that for $R_0 \geq (4+2 \sqrt{3}) \vee (\lambda_v/2^{\half})$, 
 \begin{equation}
 A(\epsilon_0)^2 \leq  \frac{4 R_0 +2^\half \lambda_v}{4 R_0 - 2^\half \lambda_v} \leq \frac{2^{3/2} \lambda_v +2^\half \lambda_v}{2^{3/2} \lambda_v - 2^\half \lambda_v} < 3^{\half}\eqsp.
 \end{equation}
 For the second part of the statement, first note that
 \begin{equation}
 \Lambda(\epsilon) = 2^{-1}[1-\epsilon(1-\lambda_x)]\big[1-\big(1-\epsilon b_{\Lambda}(\epsilon)\big)^{\half}\big] \eqsp,
 \end{equation}
 where $b_{\Lambda}(\epsilon)=\big[  4 \lambda_x (1-\epsilon)- \epsilon R_0^2 \big]/[1-\epsilon(1-\lambda_x)]^2  \in \ccint{0,\epsilon^{-1}}$ for $\epsilon \leq (2^\half\lambda_v)^{-1} \wedge \{ 4 \lambda_x /(4\lambda_x+R_0^2)\}$.  Using that  for any $a \in \ccint{0,1}$, $a/2 \leq 1-(1-a)^{1/2} \leq a$  we deduce that for $\epsilon \leq (2^\half\lambda_v)^{-1} \wedge \{ 4 \lambda_x /(4\lambda_x+R_0^2)\}$,
 \begin{equation}
   4^{-1}[1-\epsilon(1-\lambda_x)] \epsilon b_{\Lambda}(\epsilon) \leq \Lambda(\epsilon) \leq 2^{-1}[1-\epsilon(1-\lambda_x)] \epsilon b_{\Lambda}(\epsilon)  \eqsp. 
 \end{equation}
Further for $R_0 \geq (4+2 \sqrt{3}) \vee (\lambda_v/2^{\half})$ we have $\epsilon_0 \leq (2^\half\lambda_v)^{-1} \wedge \{ 3 \lambda_x /(4\lambda_x+R_0^2)\}$ from \Cref{thm:DMSmain}-\ref{thm:DMSmain_item_b}, leading to 
  \[\lambda_x /[1-\epsilon_0(1-\lambda_x)]^2 \leq b_{\Lambda}(\epsilon_0) \leq 4\lambda_x /[1-\epsilon_0(1-\lambda_x)]^2 \eqsp,
  \] 
  and consequently, using \eqref{eq:borne_eps_0_R_0},
  \begin{equation}
 \epsilon_0 \lambda_x/4 \leq  \Lambda(\epsilon_0) \leq 2\lambda_x \epsilon_0/[1-2(1-\lambda_x)/(4+R_0^2)] \leq 4 \lambda_x \epsilon_0 \eqsp,
\end{equation}
where we have used that $\lambda_x \leq 1$ for the last inequality. Finally we note that from \eqref{eq:borne_eps_0_R_0}
\[
\frac{1}{2} \leq \frac{1}{1 + 2^{3/2} \lambda_v/(4+R_0^2)} \leq \frac{1}{1+2^\half \lambda_v \epsilon_0} \leq 1 \eqsp,
\]
where the leftmost inequality follows from the fact that for $2^\half R_0 \geq \lambda_v$
\[
\frac{2^{3/2} \lambda_v}{4+R_0^2} \leq \frac{2^{3/2} \lambda_v}{4+2^{-1} \lambda_v^2} \leq 1 \eqsp.
\]
\end{enumerate}
\end{proof}

\subsection{Proof of  Theorem \ref{thm:scaling-with-d}}
\label{sec:proof-theor-refthm:s}
\begin{proof}[Proof of Theorem \ref{thm:scaling-with-d}]
  Since $\lambda_v = \ulambda$ and $R_0 \geq (4+2 \sqrt{3}) \vee (\ulambda/2^{\half})$ by \Cref{thm:hypocoercivity}, from  \Cref{thm:DMSmain} and \Cref{lemma:behaviour_rate_R_0},  $A < 3^{\half}$  while with $\lambda_x=\Cp/(1+\Cp )$
  \begin{equation}
    \label{eq:borne_alpha_discussion}
   \ulambda \lambda_x m_2^\half \epsilon_0/6 \leq \alpha(\epsilon_0) 
   \quad \text{with} \quad \lambda_x/(1 +R_0^{2}) \leq \epsilon_0 \leq 2/(4 +R_0^2) \eqsp .
  \end{equation}
By \eqref{eq:def_R_0_cas_gene}, if $c_1,c_2, \norm{a}_{\infty},m_b$ are fixed, there exist  $C^R_1(\Cp,c_1,c_2, \norm{a}_{\infty},m_b) >0$, independent of $d,\ulambda$, $c_\lambda$, $C_{\vphi}$ and $c_{\vphi}$ such that
 \begin{equation}
\bar{R}_0 \leq  C^R_1(\Cp,c_1,c_2, \norm{a}_{\infty},m_b) \bar{R}_1 \eqsp,
 \end{equation}
 where $\bar{R}_1= c_\vphi K + (1+C_{\vphi}) d^{(1+\varpi)/2} K+\ulambda(1+c_{\lambda} d^{(1+\varpi)/2})$. Combining this bound with \eqref{eq:borne_alpha_discussion} concludes the proof.
\end{proof}


\section{The Zig-Zag sampler--optimization}
\label{ss:scaling-ZZ}
In this section, we specify our results in the case of the Zig-Zag sampler for which better estimates can be obtained, leading to better scaling properties with respect to $d$. The Zig-Zag process corresponds to the instantiation of \eqref{eq:generator pdmp} for which $F_0=0$, $K=d$, $F_i(x) = \partial_{x_i} U(x) \bfe_i$, $\rmn_{i}(x)=\bfe_{i}$, $\lambdaref(x) = \ulambda >0$\footnote{which corresponds to $c_{\lambda} = 0$ in \Cref{as:refreshment}}, for $i \in \{1,\ldots,d\}$ and $x \in \msd$,  and $\calR_v = \Pi_v-\Id$. The corresponding generator takes the simplified form, for $f \in \mrCb^2(\mse)$ and any $(x,v) \in \mse$ 

\begin{align}
  \generator f(x,v) &= v^{\top} \nabla_xf(x) + \sum_{i=1}^d \vphi\big(v_i \partial_{x_i} U (x)\big)\big[f\big(x,(\Id - 2 \bfe_i \bfe_i^{\top})v\big) - f(x,v)\big]\\
   \label{eq:def_generator_zz}
  & \qquad \qquad \qquad \qquad \qquad \qquad  + \lambdaref(x) m_2^\half \mcr_v f(x,v) \eqsp,
\end{align}
where $\vphi : \rset \to \rset_+$ is a continuous function and satisfies \eqref{eq:prop_vphi} in \Cref{as:intensities}.

In the next two subsections we first consider general velocity distributions and then show how our results can be specialized to the scenario where $\msv = \{-m_2^\half,+m_2^\half \}^d$ for $m_2>0$ and $\nu$ is the uniform distribution on $\msv$.

\subsection{General velocity distribution}

\begin{theorem}
 \label{theo:ZZ spectral gap}
Consider the Zig-Zag process with generator defined by \eqref{eq:def_generator_zz} with $\lambdaref=\ulambda$, $\mcr_v = \Piv-\Id$ and  $\vphi : \rset \to \rset_+$ is a continuous function satisfying \eqref{eq:prop_vphi} in \Cref{as:intensities}. Assume  \Cref{as:generator} with $\core=\rmCb^2(\mse)$, \Cref{as:U}, \Cref{as:Fk}, \Cref{as:radial}, \Cref{as:operator on velocities}, \Cref{as:refreshment} hold and that there exists $c_3\geq 0$ such that for any $g \in \Lpi^d$
\begin{equation} \label{eq:ZZ-cond-hessian-U}
\psmu{g}{\big[\nax^{2}U-{\rm diag}(\nax^{2}U)\big]g}\geq-c_3 \normmu{g}^{2} \eqsp.
\end{equation}
Then, \Cref{thm:DMSmain} holds with $\lambda_x$ as in \eqref{eq:def_lambda_x}, $\lambda_v=\ulambda$ and
\begin{equation}
\label{eq:R_0_zz}
 R_0 = \frac{(6  m_4)^\half (2+C_\vphi)}{m_2} \parenthese{\left(1+c_1/2\right)^\half+1+ (c_3/2)^\half} +  \frac{\ulambda   + c_\vphi}{2^\half}    \eqsp.
\end{equation}
\end{theorem}
\begin{remark}
 \label{rem:ZZ spectral gap}
 From \Cref{as:U} we have for any $g \in \Lpi ^ d$
\[
\psmu{ g}{\nax^{2}Ug}\geq-c_{1}\normmu{g}^{2}
\]
 and therefore \eqref{eq:ZZ-cond-hessian-U} holds if there exist $\bar{c}_1>0$ such that for any $g\in \Lpi ^ d$,
\[
\psmu{ g}{{\rm diag}(\nax^{2}U)g}\leq \bar{c}_1\normmu{g}^{2} \eqsp,
\]
which is itself implied by $\bar{c}_1 \Id \succeq {\rm diag}(\nax^{2}U(x))$ for all $x \in \msd$, since the matrix ${\rm diag}(\nax^{2}U(x))$ is symmetric. Note that this is the case when for all $x \in \msd$, $\vert {\rm diag}(\nax^{2}U(x)) \vert \leq \bar{c}_1$ or $\vert \nax^{2}U(x) \vert \leq \bar{c}_1$, for example.
\end{remark}
The proof is very similar to that of \Cref{thm:hypocoercivity} and follows from the application of \Cref{thm:DMSmain} and the following lemmas whose proofs can be found in \Cref{ssec:post_proof_zig_zag}. 

\begin{lemma}
\label{lemma:ZZ RS bounded}
Consider the Zig-Zag process with generator $\generator$ defined by \eqref{eq:def_generator_zz} with $\lambdaref=\ulambda$, $\mcr_v = \Piv-\Id$ and  $\vphi : \rset \to \rset_+$ is a continuous function satisfying \eqref{eq:prop_vphi} in \Cref{as:intensities}.
Assume \Cref{as:generator} with $\core= \rmCb^2(\mse)$, \Cref{as:U}, \Cref{as:Fk}, \Cref{as:radial}, \Cref{as:operator on velocities}, \Cref{as:refreshment} and \eqref{eq:ZZ-cond-hessian-U} hold. Let $\mcs$ and  $\mct$ be the symmetric and anti-symmetric parts of  $\generator$ respectively and $\mca$ the operator defined by \eqref{eq:defcalA} relative to $\mct$.
 Then for any $f \in \rmCb^2(\mse)$,
\begin{align}
&  \normmuLine{(\Id-\Piv)\tmcs \mca^{\star} f} \leq  \big(\ulambda  + c_\vphi \big) m_2 \normmu{ \nax u_f } \\
  & \qquad \qquad  (6 m_4)^\half C_\vphi   \left( \normmu{\nax^2 u_f} +  \normmu{\nax^{\star} \nax u_f }+ c_3^\half \normmu{\nax u_f} \right)  \eqsp,
\end{align}
where $u_f$ is given by \eqref{eq:def_u}.
\end{lemma}

\begin{lemma}
\label{lemma:ZZ RT bounded}
Consider the Zig-Zag process with generator $\generator$ defined by \eqref{eq:def_generator_zz} with $\lambdaref=\ulambda$, $\mcr_v = \Piv-\Id$ and  $\vphi : \rset \to \rset_+$  a continuous function satisfying \eqref{eq:prop_vphi} in \Cref{as:intensities}.
Assume \Cref{as:generator} with $\core= \rmCb^2(\mse)$, \Cref{as:U}, \Cref{as:Fk}, \Cref{as:radial}, \Cref{as:operator on velocities}, \Cref{as:refreshment} and \eqref{eq:ZZ-cond-hessian-U} hold. Let  $\mct$ be the anti-symmetric part of  $\generator$ and $\mca$ the operator  defined by \eqref{eq:defcalA} relative to $\mct$. 
 Then for any $f \in \rmCb^2(\mse)$
\begin{align}
  & \normmu{(\Id-\Piv) \tmct \mca^{\star}  f} \\
  & \qquad \qquad \leq [6(4m_4-m_{2,2})]^\half \left(\normmu{\nax^{2}u_f}+\normmu{\nax^{*}\nax u_f}+c_3^\half \normmu{\nax u_f}\right) \eqsp,
\end{align}
where  $u_f$ is defined by \eqref{eq:def_u}.
\end{lemma}

\begin{proof}[Proof of \Cref{theo:ZZ spectral gap}]
Checking  \Cref{ass:stability_c_piv} and \Cref{as:DMSabstract}-\ref{item:DMS-micro}-\ref{item:DMS-macro}-\ref{item:DMS-proj}-\ref{item:DMS-kernel_calS} is identical to the work done in the proof of \Cref{thm:hypocoercivity} with the constants $\lambda_v = \ulambda$ and $\lambda_x$ given by \eqref{eq:def_lambda_x}.
We are left with checking \Cref{as:DMSabstract}-\ref{item:DMS-RSandRT}.
By the improved bounds from  \Cref{lemma:ZZ RS bounded} and \Cref{lemma:ZZ RT bounded}, we have for any $f \in \rmCb^2(\mse)$, 
\begin{align}
&  \normmu{\tmcs \mca^{\star} f} +\normmu{(\Id-\Piv) \tmct \mca^{\star} f }\leq \big(\ulambda  + c_\vphi \big) m_2 \normmu{\nax u_f}  \\
 & \qquad \qquad  + (6  m_4)^\half (2 + C_\vphi)\left(\normmu{\nax^{2}u_f}+\normmu{\nax^{*}\nax u_f}+c_3^\half \normmu{\nax u_f}\right)  \eqsp.
\end{align}
Using \Cref{prop:bound hessian} and \Cref{coro:bound Fk}, we obtain that for any $f \in \rmCb^2(\mse)$, 
\begin{align}
&\normmu{(\Id-\Piv)\tmcs \mca^{\star} f} +\normmu{(\Id-\Piv) \tmct \mca^{\star} f } \\
&\leq \defEns{\frac{(6  m_4)^\half (2+C_\vphi)}{m_2} \parenthese{\left(1+c_1/2\right)^\half+1+ (c_3/2)^\half} + \frac{\ulambda   +  c_\vphi }{2^\half}} \normmu{\Piv f} \eqsp, 
\end{align}
The proof is then completed by \Cref{lemma:RS bounded}-\ref{lem:bound_a_s_item_a} and  \Cref{lemma:RT bounded}-\ref{lem:bound_a_t_item_a}.
\end{proof}

We discuss in the following the dependence on the dimension of the convergence rate $\alpha(\epsilon_0)$ and the constant $A(\epsilon_0)$ given by \Cref{thm:DMSmain} based on the constant provided by \Cref{theo:ZZ spectral gap}. Similarly to the general case, we need to impose some conditions on $m_2$ and $m_4$. Here, we assume that $m_4^{1/2}/m_2$ does not depend on $d$, which holds in the case where $\nu$ is the uniform distribution on $\msv = \{-1,1\}^d$ or the $d$-dimensional zero-mean Gaussian distribution with covariance matrix $\Idd$.

In the case where $\pi$ is the \iid~product of one-dimensional distributions $\pi_i$ on $(\rset,\mcb{\rset})$ associated with potentials $U_i : \rset \to \rset$ satisfying \Cref{as:U}, \ie~for any $x \in \msx$, $U(x) = \sum_{i=1}^d U_i(x_i)$, $\nax^2 U(x) = \diag(\nax^2U(x))$ for any $x \in \msx$ and therefore \eqref{eq:ZZ-cond-hessian-U} holds with $c_3=0$. Then, the convergence rate $\rateConv(\veps_0)$ and the constant $A(\veps_0)$ in \Cref{thm:DMSmain} do not depend on the dimension but only on the constants $c_1$, $c_2$, $\ulambda$, $c_\lambda$ and $\Cp$ associated to each $U_i$.

Consider now the case where the potential $U$ is strongly convex and gradient Lipschitz, \ie~there exist $m,L >0$ such that $m \Idd \preceq \nax^2 U(x) \preceq L \Idd$ for any $x \in \msx$. Then, since for any $i \in \{1,\ldots,d\}$ and $x \in \msx$, $\partial_{x_i,x_i} U(x) = \ps{\bfe_i}{\nax^2 U(x) \bfe_i} \leq L$ by assumption, \Cref{rem:ZZ spectral gap} implies that \eqref{eq:ZZ-cond-hessian-U} holds for $c_3 = L-m$. In addition, \Cref{as:U} holds with $c_1=0$ and $c_2= L$ and by \cite[Proposition 5.1.3, Corollary 5.7.2]{bakry:gentil:ledoux:2014}, $U$ satisfies \eqref{eq:poincare_assumption} with $\Cp=m$. Then, the convergence rate $\rateConv(\veps_0)$ and the constant $A(\veps_0)$ in \Cref{thm:DMSmain} do not depend on the dimension but only on $L$, $m$, $\ulambda$ and $\olambda$. In addition, we observe that the larger $L-m$ is, the larger $R_0$ given in \eqref{eq:R_0_zz} is, which in turn make the convergence rate $\rateConv(\veps_0)$ worse since it is of order  $\bigO(1/R_0^2)$ as $R_0 \to \plusinfty$ by \Cref{lemma:behaviour_rate_R_0}. This result is expected in the Gaussian case $U(x) = x^{\top} \Sigma x$ for any $x \in \msx$, since $L-m$ is the diameter of the set of eigenvalues of $\Sigma$ which is a characterization of the conditioning of the problem.

\subsection{\texorpdfstring{$d$}{d}-dimensional Radmacher distribution}
\label{ss:scaling-radmacher} 
We now consider the case $\msv = \{-m_2^{\half},+m_2^{\half}\}^d$ and $\nu$ is the uniform distribution on $\msv$ which corresponds to the original setting of the Zig-Zag process. This process has been proved to be ergodic~\cite{BRZ2018} even in the absence of refreshment, that is $\lambdaref=0$. 
We note that in this scenario $m_4 = m_2^2/ 3$ and $m_{2,2} = m_2^2$ which leads to simplified expressions for the bounds in  \Cref{lemma:ZZ RS bounded} and \Cref{lemma:ZZ RT bounded} upon revisiting their proofs. However this has no qualitative impact. In this section we show that hypocoercivity holds with our techniques for $\lambdaref(x)=0$ for ``most of $\msd$'' for a particular type of partial refreshment update. 

Consider the scenario where $\calR_v$ is a mixture of the bounces $\{ \calB_k, k=1,\ldots,d \}$, for any $f \in \mrl^2(\mu)$, $(x,v) \in \mse$,
\begin{equation} \label{eq:radmacher-def-Rv-mixture}
	\lambdaref \calR_v f(x,v) = \sum_{k=1}^d \lambda_{{\rm ref}, k}(x) \big[f\big(x, v - 2 v_k \bfe_k \big) - f(x,v)\big] \eqsp,
\end{equation}
with $\lambda_{{\rm ref}, k}\colon \msd \to \rset_+$ for $k \in \{1, \ldots, d\}$ satisfying \Cref{as:refreshment}, and $\lambdaref = \sum_{k=1}^d \lambda_{{\rm ref}, k}$, that is when the process refreshes, $k \in \{1,\ldots,d\}$ is chosen at random with probability proportional to $(\lambda_{{\rm ref}, 1},\ldots,\lambda_{{\rm ref}, d})$ and the component $v_k$ of $v$ is updated to $-v_k$.

\begin{proposition}
\label{propo:zz_micro}
Consider the Zig-Zag process with generator $\generator$ and refreshment operator as in \eqref{eq:def_generator_zz} and \eqref{eq:radmacher-def-Rv-mixture} respectively, with  $\vphi : \rset \to \rset_+$ is a continuous function satisfying \eqref{eq:prop_vphi} in \Cref{as:intensities}.
Assume \Cref{as:generator} with $\core= \rmCb^2(\mse)$, \Cref{as:U}, \Cref{as:Fk}, \Cref{as:radial}, \Cref{as:operator on velocities}, \Cref{as:refreshment} and \eqref{eq:ZZ-cond-hessian-U} hold. Let $\mcs$  be the symmetric  part of  $\generator$ defined by \eqref{eq:def_calS_calT}.
\begin{enumerate}[wide, labelwidth=!, labelindent=0pt]
\item \label{propo:zz_micro_a} the symmetric part of the generator is given for any $f \in \rmCb^{2}(\mse)$, $(x,v) \in \mse$ by
  \begin{align}
&    \calS f(x,v) =  \sum_{k=1}^d \left\lbrace \frac{\vphi\big(v_k\partial_{x_k} U(x)\big) + \vphi\big(-v_k\partial_{x_k} U(x)\big) }{2} \right. \\
      & \qquad \qquad \qquad \qquad\qquad \qquad  +  m_2^\half\lambda_{{\rm ref}, k}(x) \Bigg\} \big[f\big(x,v - 2 v_k \bfe_k\big) - f(x,v)\big] \eqsp;
\end{align}
\item \label{propo:zz_micro_b} the microscopic coercivity condition \Cref{as:DMSabstract}-\ref{item:DMS-micro} is satisfied, i.e. for any $f \in \rmCb^{2}(\mse)$, $(x,v) \in \mse$

 \begin{align}
 \label{eq:zz_micro_1}
 -\psmu{ \calS \ff}{ \ff} 
 &\geq \lambda_v m_2^{\half} \normmu{ (\Id-\Piv) \ff }^2 \\
 \quad \text{with} \quad \lambda_v &= \min_{k \in \{1,\ldots,d \}, x \in \msd} \left\lbrace \frac{ |\partial_{x_k} U(x)|} 2 + \lambda_{{\rm ref}, k}(x) \right\rbrace\eqsp. 
\end{align}

\end{enumerate}
\end{proposition}
 \begin{remark}
In other words \Cref{as:DMSabstract}-\ref{item:DMS-micro} holds if for any $\varepsilon>0$, for all $k \in \{1,\ldots,d\}$, $\lambda_{{\rm ref}, k}$ vanishes everywhere, except on $\{x \in \msx \, :\, \exists k \in \{1,\ldots,d\} \mid \absLigne{\partial_{x_k} U}(x) < \varepsilon \}$. We also note that a similar result holds for the case where $\mcr_v = \Piv-\Id$, that is \Cref{as:DMSabstract}-\ref{item:DMS-micro} holds whenever $\lambdaref$ vanishes everywhere, except on $\{x \in \msx \, :\, \exists k \in \{1,\ldots,d\}, \, \mid \absLigne{\partial_{x_k} U}(x) < \varepsilon \}$ for $\varepsilon >0$.
\end{remark}
\begin{proof}
The first statement is a direct application of \Cref{prop:dms_pdmp_1}-\ref{item:sym skewsym}.  For the second statement, using that $\nu$ is the uniform distribution on $\msv = \{-m_2^\half, m_2^\half\}^d$, from the polarization identity and since $\vphi$ satisfies \Cref{as:intensities}, we get for any $f \in \rmCb^{2}(\mse)$, setting $\vphi^{\even}(s) = \vphi(s)+\vphi(-s)$,
\begin{align}
-\psmu{ \calS \ff}{ \ff} &= \frac{1}{2} \int_{\mse} \sum_{k=1}^d \defEns{ \frac{\vphi^{\even}(v_k\partial_{x_k} U(x)) }{2} +  m_2^\half\lambda_{{\rm ref}, k}(x) } \\
&\qquad \qquad \qquad \qquad   \times \big[f(x,v) - f\big(x,(\Id - 2 \bfe_k \bfe_k^{\top})v\big) \big]^2 \, \rmd \mu (x,v)\\
& \geq ( \lambda_v m_2^{\half}/2) \int_{\mse} \sum_{k=1}^d \big[ f(x,v) - f\big(x,(\Id - 2 \bfe_k \bfe_k^{\top})v\big)\big]^2 \, \rmd \mu (x,v) \eqsp,
  \label{eq:proof_zz_micro_1}
\end{align}
where $\lambda_v$ is defined in \eqref{eq:zz_micro_1}.
Now by the Poincar{\'e} inequality for any $g \in \rml^2_0(\nu)$, see \eg~\cite[p. 52]{odonnell:2014}, it holds that 
\begin{equation}
 \label{eq:proof_zz_micro_2}
 (1/2) \int_{\msv}\sum_{k=1}^d \big[ g(v) - g\big((\Id - 2 \bfe_i \bfe_i^{\top})v\big) \big]^2 \, \rmd \nu(v) \geq \int_{\msv}\sum_{k=1}^d g^2(v) \, \rmd \nu(v) \eqsp.
\end{equation}
Now since for any $f \in \rmCb^{2}(\mse)$, $ \psmu{ \calS \ff}{ \ff} = \psmu{ \calS (\Id-\Piv)\ff}{ (\Id-\Piv)\ff}$ and for any $x \in \msx$, $v \mapsto (\Id-\Piv)f(x,v) \in \mrl^2_0(\nu)$, then combining \eqref{eq:proof_zz_micro_1} and \eqref{eq:proof_zz_micro_2} and using Fubini's theorem concludes the proof of \eqref{eq:zz_micro_1}.
\end{proof}

\subsection{Postponed proofs}
\label{ssec:post_proof_zig_zag}

\begin{proof}[Proof of \Cref{lemma:ZZ RS bounded}]
  We use \Cref{lemma:RS bounded} and its notation, where $K=d$, for $k \in \{1,\ldots,d\}$, $F_k = \partial_{x_k} U$ and $\rmn_k = \sign(\partial_{x_k} U) \bfe_k$. In this setting and by \eqref{eq:def_bfG}, it follows that for any $(x,v)\in \mse$,
\[
\bfG(x,v)= \sum_{k=1}^d \lambda^\even_k(x,v) v_k  \bfe_k + \ulambda m_2^{1/2} v \eqsp.
\]
By the triangle inequality and since $\int_{\msv} g(v_i) g(v_j)  v_iv_j {\rm d}\nu(v)= 0$ for $i,j \in \{1,\ldots,d\}$, $i \neq j$, and any even measurable bounded function $g: \rset \to \rset$ by \Cref{as:radial}-\ref{item:as:radial:general_3}, we get 
\begin{align}
&  \normmu{\bfG^{\top} \nax u_f} \\
  &  \leq \normmu{\textstyle{\sum}_{k=1}^d \{\vphi(v_k \partial_{x_k} U) + \vphi(-v_k \partial_{x_k} U)\} v_k  \partial_{x_k} u_f} + \ulambda m_2 \normmu{ \nax u_f }  \\
  \label{eq:proof_lem_19_date_16_04}
  &  = \parentheseDeux{\sum_{k=1}^d   \normmu{ \{\vphi(v_k \partial_{x_k} U) + \vphi(-v_k \partial_{x_k} U)\}  v_k  \partial_{x_k} u_f}^2}^{\half} + \ulambda m_2 \normmu{ \nax u_f } \eqsp.
\end{align}
Then by \Cref{as:intensities}, \Cref{as:radial}-\ref{item:as:radial:general_3}, the triangle inequality (on $\mrl^2(\mu)^d$) and since for any $i \in \{1,\ldots,d\}$, 
$\int_{\msv}v_i^4 \rmd \nu (v) = 3 m_4$ by
\Cref{as:radial}-\ref{assum:fourth_moment} we obtain
\begin{align}
&  \parentheseDeux{\sum_{k=1}^d   \normmu{ (c_{\vphi} m_2^\half + C_{\vphi} \abs{v_k \partial_{x_k} U})  v_k  \partial_{x_k} u_f}^2}^{\half}  \\
  &\qquad \qquad  \leq c_{\vphi}  m_2^\half  \parentheseDeux{\sum_{k=1}^d  \normmu{  v_k  \partial_{x_k} u_f}^2 }^{\half} +  C_{\vphi}  \parentheseDeux{ \sum_{k=1}^d \normmu{\abs{v_k \partial_{x_k} U}  v_k  \partial_{x_k} u_f}^2}^{\half}  \\
  & \qquad \qquad  \leq c_{\vphi}  m_2 \normmu{\nabla_x u_f} +  C_{\vphi}(3m_4)^{\half} \parentheseDeux{\sum_{k=1}^d \normmu{ \partial_{x_k} U  \partial_{x_k} u_f}^2}^{\half} \eqsp. 
\end{align}
Plugging this result in \eqref{eq:proof_lem_19_date_16_04}, we get 
\begin{equation}
  \normmu{\bfG^{\top}\nax u_f}  \leq 
  \label{eq:bound_a_s_zz_1} (  c_{\vphi}  + \ulambda )m_2 \normmu{\nabla_x u_f} +  C_{\vphi}(3m_4)^{\half} \parentheseDeux{\sum_{k=1}^d \normmu{ \partial_{x_k} U  \partial_{x_k} u_f}^2}^{\half} \eqsp. 
\end{equation}

To bound the sum we
note that for $k\in\{1,\ldots,d\}$
$\partial_{x_{k}}U\partial_{x_{k}}u_f =
\partial_{x_{k}}^{2}u_f+\partial_{x_{k}}^{*}\partial_{x_{k}}u_f$ by
\Cref{lem:adjoint-nax}-\ref{lem:adjoint-nax_a_0}, which together with
the fact $(a+b)^2 \leq 2(a^2+b^2)$ leads to 
\begin{equation}
\normmu{\partial_{x_{i}}U\partial_{x_{i}}u_f}^2 \leq 2\big( \normmu{\partial_{x_{i}}^{2}u_f}^2+\normmu{\partial_{x_{i}}^{*}\partial_{x_{i}}u_f}^2 \big) \eqsp.
\end{equation}
Then, using that for $a,b\geq0$ $\sqrt{a+b} \leq \sqrt{a} +\sqrt{b}$ twice and \eqref{eq:ZZ-bound-sum-partial-star-partial}, we deduce
\begin{align}
 \nonumber
 \left(\sum_{k=1}^d\normmu{\partial_{x_{k}}U\partial_{x_{k}}u_f}^2\right)^\half & \leq 2^{\half} \defEns{ \sum_{k=1}^d \parenthese{ \normmu{\partial_{x_{k}}^{2}u_f}^2+\normmu{\partial_{x_{k}}^{*}\partial_{x_{k}}u_f}^2 } }^\half\\
 \nonumber
 & \leq 2^{\half} \defEns{\left( \sum_{k=1}^d \normmu{\partial_{x_{k}}^{2}u_f}^2 \right)^\half + \left( \sum_{k=1}^d \normmu{\partial_{x_{k}}^{*}\partial_{x_{k}}u_f}^2 \right)^\half}\\
 \label{eq:bound_a_s_zz_2}
 & \leq 2^{\half} \parenthese{\normmu{\nax^2 u_f}+\normmu{\nax^{*}\nax u_f}+c_3^\half \normmu{\nax u_f} } \eqsp. 
\end{align}
Then combining \eqref{eq:bound_a_s_zz_1} and \eqref{eq:bound_a_s_zz_2} completes the proof by \Cref{lemma:RS bounded}-\ref{lem:bound_a_s_item_b}.
\end{proof}

For $a,b \in \rset^d$ ($A,B \in \rset^{d \times d}$), we denote by $a \odot b \in \rset^d$ ($A \odot B \in\rset^{d \times d}$) the Hadamard product between $a$ and $b$ defined for any $i \in \{1,\ldots,d\}$ ($i,j \in \{1,\ldots,d\}$) by $(a\odot b)_i = a_i b_i$ ($(A\odot B)_{i,j} = A_{i,j}B_{i,j}$).

\begin{proof}[Proof of \Cref{lemma:ZZ RT bounded}]
 We use \Cref{lemma:RT bounded} and its notations, where $K=d$, for $k \in \{1,\ldots,d\}$, $F_k = \partial_{x_k} U \bfe_k$ and $\rmn_k = \sign(\partial_{x_k} U) \bfe_k$. In this setting and by \eqref{eq:def_bfM}, it follows that 
\[
\bfM(x)=\nax^{2}u_f(x)+{\rm diag}\big(\nax u_f\odot\nax U\big),
\]
Since $\normmu{\bfM}^{2}=\normmu{{\rm diag}(\bfM)}^{2}+\normmu{\bfM-{\rm diag}(\bfM)}^{2}$, we obtain
\begin{align}
 \nonumber
  &2m_{2,2}\normmu{\bfM}^{2}+3(m_{4}-m_{2,2})\normmu{{\rm diag}(\bfM)}^{2}  \\
  &=2m_{2,2}\normmu{\bfM-{\rm diag}(\bfM)}^{2}+(3m_{4}-m_{2,2})\normmu{{\rm diag}(\bfM)}^{2} \\
 \label{eq:bound_a_t_zz_1}
 & \leq 2m_{2,2}\normmu{\nax^{2}u_f}^{2}+(3m_{4}-m_{2,2})\normmu{{\rm diag}(\bfM)}^{2}\eqsp.
\end{align}
We now bound $\normmu{{\rm diag}(\bfM)}^{2}$. 
First, we apply the triangle inequality and use \Cref{lem:adjoint-nax}-\ref{lem:adjoint-nax_a_0}, to deduce that 
\begin{align}
\nonumber
&  \normpi{\diag(\bfM)}^2 = \sum_{k=1}^d \normmu{2\partial_{x_{k}}^{2}u_f-\partial_{x_{k}}^{2}u_f+\partial_{x_{k}}U\partial_{x_{k}}u_f}^{2} \\
  & \phantom{\normpi{\diag(\bfM)}^2 } \leq \sum_{k=1}^d \left(2\normmu{\partial_{x_{k}}^{2}u_f}+\normmu{-\partial_{x_{k}}^{2}u_f+\partial_{x_{k}}U\partial_{x_k}u_f}\right)^{2} \\
 \label{eq:bound_a_t_zz_2} 
&   \phantom{\normpi{\diag(\bfM)}^2 } \leq \sum_{k=1}^d \parenthese{8\normmu{\partial_{x_{k}}^{2}u_f}^{2}+2\normmu{\partial_{x_{k}}^{*}\partial_{x_{k}}u_f}^{2} }\eqsp,
\end{align}
where we have used for the last inequality that $(a+b)^2 \leq 2a^2 +2b^2$ for any $a,b \in \rset$. 
By \Cref{lem:adjoint-nax}-\ref{lem:adjoint-nax_a_0}, \eqref{eq:commutation_nabla}, \eqref{eq:borne_moment} and the fact that $U \in \rmC^3_{\poly}(\msx)$ by \Cref{as:U}, using that same reasoning as to establish \eqref{eq:borne_nax_u_2}, it holds for any $k \in \{1,\ldots,d\}$,
\begin{align}
  \normmu{\partial_{x_k}^{\star} \partial_{x_k} u_f}^2 &= \normmu{\partial^2_{x_k}u_f}^2 + \psmu{\partial_{x_k} u_f}{ \partial_{x_k,x_k} U \, \partial_{x_k} u_f } \eqsp, \\
  \normmu{\nax^{*}\nax u_f}^{2} & =\normmu{\nax^{2}u_f}^{2}+\psmu{\nax u_f}{\nax^{2}U\nax u_f} \eqsp.
\end{align}
These identities and the condition \eqref{eq:ZZ-cond-hessian-U} imply
\begin{align}
 \nonumber
  \sum_{i=1}^{d}\normmu{\partial_{x_{i}}^{*}\partial_{x_{i}}u_f}^{2} &=\normmu{{\rm diag}\big(\nax^{2}u_f\big)}^{2}+\psmu{\nax u_f}{{\rm diag}\big(\nax^{2}U\big)\nax u_f}\\
  &\leq \normmu{\nax^{2}u_f}^{2}+\psmu{\nax u_f}{{\rm diag}\big(\nax^{2}U\big)\nax u_f}\\
  &\leq\normmu{\nax^{*}\nax u_f}^{2}-\psmu{\nax u_f}{\big(\nax^{2}U-{\rm diag}(\nax^{2}U)\big)\nax u_f}\\
  &\leq \normmu{\nax^{*}\nax u_f}^{2} + c_3 \normmu{\nax u_f}^2 \eqsp. \label{eq:ZZ-bound-sum-partial-star-partial}
\end{align}
Combining \eqref{eq:bound_a_t_zz_2} and \eqref{eq:ZZ-bound-sum-partial-star-partial}, we obtain
\begin{equation}
 \normmu{\diag(\bfM)}^2 \leq 8\sum_{k=1}^d \normmu{\partial_{x_{k}}^{2}u_f}^{2} + 2 (\normmu{\nax^{*}\nax u_f}^{2} + c_3 \normmu{\nax u_f}^2) \eqsp.
\end{equation}
From this inequality, \eqref{eq:bound_a_t_zz_1} and \Cref{lemma:RT bounded}-\ref{lem:bound_a_t_item_b}, we deduce
\begin{align}
  &\normmu{(\Id-\Piv) \tmct \mca^{\star} f}^{2}\leq 6(4m_4-m_{2,2})\normmu{\nax^{2}u_f}^{2}\\
  & \qquad \qquad \qquad \qquad \qquad  +2(3m_{4}-m_{2,2})\left(\normmu{\nax^{*}\nax u_f}^{2}+c_3\normmu{\nax u_f}^{2}\right)\\
& \phantom{\normmu{(\Id-\Piv) }} \leq 6(4m_4-m_{2,2}) \left(\normmu{\nax^{2}u_f}+\normmu{\nax^{*}\nax u_f}+c_3^\half\normmu{\nax u_f}\right)^2 \eqsp,
\end{align}
since for $a,b, c\geq 0$, $a^2+ b^2+c^2\leq (a+b+c)^2$.
\end{proof}




\section{Discussion and link to earlier work} \label{sec:discussion}
As pointed out earlier the scenario $K=0$ where $F_0=\nax U$ is considered in \cite{Dolbeault15} where the authors establish hypercoercivity but also in \cite[Theorem 3.9]{BouRabee17} where the authors establish geometric convergence, that is the existence of constants $A, \alpha > 0$ and a measurable function $V\colon\mse \to \Rset_+$  satisfying $\mu\big(\{V=\infty\}\big)=0$, such that for any $(x,v)\in \mse$ and $t \geq 0$,
\begin{equation} \label{eq:def-exponential-convergence}
\tvnorm{ P_{t}\big((x,v),\cdot\big)-\mu(\cdot)} \leq A V(x,v) \rme^{-\alpha t} \eqsp.
\end{equation}

Similar results have been obtained in \cite{deligiannidis2017arXiv170504579D} and \cite{durmus-geometric-2018arXiv180705401D} for the Bouncy particle sampler and in \cite{BRZ2018} for the Zig-Zag process. All these methods rely on guessing such a suitable Lyapounov function $V$ and establishing a so-called drift condition for this function, in conjunction with a minorization condition \cite{meyn2012markov}. Here we have established $\Lmu$-exponential convergence, or equivalently that there exists an absolute $\Lmu$-absolute spectral gap \cite[Proposition 22.3.2]{douc-moulines-priouret-soulier-2019} (by considering the skeleton of the process) and is therefore $\mu$-a.e. uniformly convergent by \cite[Proposition 22.3.3 and Proposition 22.3.5]{douc-moulines-priouret-soulier-2019}, that is \eqref{eq:def-exponential-convergence} holds with $V=\bfone$ and $\mu$-a.e..

An advantage of our approach is that it provides explicit and relatively simple bounds in terms of interpretable quantities which, we show, are informative, and is in contrast with those on minorization and drift conditions in most scenarios. One exception is the study of BPS on the torus carried out in \cite{durmus-geometric-2018arXiv180705401D} for $U=0$, using an appropriate coupling argument, which leads to a rate of convergence for the total variation distance with a favourable $\Theta(d^{1/2})$ scaling. Although we have shown that for the Zig-Zag sampler with Rademacher distribution $\lambdaref$ is not required to be bounded away from zero on $\msx$, the results of \cite{BRZ2018} hold with $\lambdaref=0$. It would be interesting to further investigate whether our results can be specialized to consider the scenario $\lambdaref=0$. 

Although we have shown that the theory developed in this paper covers numerous scenarios in a unified set-up, various possible extensions are possible. For example we have restricted this first investigation to deterministic bounces of the type given in \eqref{eq:def-bounces}, but there does not seem to be any obstacle to the extension of our results to the more general set-ups such as considered in \cite{Vanetti17,Wu17,Michel17}. In the same vein, great parts of our calculations could be used to consider distributions of the velocity $\nu$ that are neither Gaussian, nor the uniform distribution on the hypersphere. For $\nu$ of density proportional to $\exp(-\mathrm{K}(v))$ with $\mathrm{K}: \rset^d  \to \Rset$ the Liouville operator involved in the definition of \eqref{eq:generator_r_d_2} would take the form $\nabla_v \mathrm{K}(v)^\top \nax \ff(x,v) - m_2 F_0^\top \nabla_v \ff(x,v)$, leading to a different expression for $\calT$.  Such modified kinetic energies have been proposed to speed up the computation, introducing the Modified Langevin Dynamics for which convergence to equilibrium has been studied in~\cite{Redon16}.


\appendix
\section{Optimization and estimates of the rate of convergence $\alpha(\epsilon)$} \label{sec:optim-varepsilon}

We let $\rset_+^*=(0,\infty)$. Consider the functions $R,\tilde{\alpha} : \rset_+^* \to \rset_+^*$ given for any $\epsilon \geq 0$ by 
\begin{align}
 \label{eq:def_R}
 R(\epsilon) &=[1-\epsilon(1-\lambda_{x})]^{2}-4\epsilon\lambda_{x}(1-\epsilon)+\epsilon^{2}R_{0}^{2}\\ &=R_{1}^{2}\left(\epsilon-\frac{1+\lambda_{x}}{R_{1}^{2}}\right)^{2}+1-\frac{(1+\lambda_{x})^{2}}{R_{1}^{2}}>0 \eqsp, \\
  \label{eq:def_alpha_tilde}
  &\tilde{\alpha}(\epsilon) = \frac{ \Lambda(\epsilon)}{1+2^\half \lambda_{v}\epsilon} = \frac{ 1 - \epsilon(1-\lambda_x) - R^{\half}(\epsilon) }{2(1+2^\half \lambda_{v}\epsilon)} \eqsp,
\end{align}
where
\begin{equation}
 \label{eq:def_R_1}
 R_{1}^{2}=(1+\lambda_{x})^{2}+R_{0}^{2} \eqsp,
\end{equation}
and $\Lambda$ is given in \eqref{eq:def_lambda_0}. 
We show that optimizing $\epsilon \mapsto \Lambda(\epsilon)$ is a good enough proxy
for optimizing $\epsilon \mapsto \tilde{\alpha}(\epsilon)$, whose maximum is unique, but intractable. Since $\epsilon \mapsto \alpha(\epsilon)$ defined by \eqref{eq:def_alpha_C_eps} is proportional to $\epsilon \mapsto \tilde{\alpha}(\epsilon)$, the same conclusion holds for this function. 
\begin{lemma}
\label{lem:derivativeLambda0}Let $\Lambda\colon\mathbb{R}_+\to\mathbb{R}$
be defined by \eqref{eq:def_lambda_0}. Then with $\lambda_{x}\in(0,1)$ and $R_0 >0$,
\begin{enumerate}
\item \label{lem:derivativeLambda0_a} $\Lambda(\epsilon)\geq0$ for $\epsilon \in \ccint{0, 4\lambda_{x}/(4\lambda_{x}+R_{0}^{2})}$
and $\Lambda(0)=0$.
\item \label{lem:derivativeLambda0_b} $\Lambda$ has first order derivative
\[
\Lambda'(\epsilon)=-(1/2)\big[(1-\lambda_{x})R^{\half}(\epsilon)+\epsilon R_{1}^{2}-(1+\lambda_{x})\big]R^{-\half}(\epsilon) \eqsp,
\]
and $\Lambda'(0)=\lambda_{x}>0$. 
\item \label{lem:derivativeLambda0_c} $\Lambda\colon\mathbb{R}_{+}\to \mathbb{R}$ has a unique stationary point ($\Lambda'(\epsilon_{0})=0$)
\begin{equation}
 \label{eq:def_eta_star}
 \epsilon_{0}=\frac{(1+\lambda_{x})-(1-\lambda_{x})\parentheseDeux{R_{0}^{2}/(R_{0}^{2}+4\lambda_{x})}^{\half}}{(1+\lambda_{x})^{2}+R_{0}^{2}}>0 \eqsp,
\end{equation}
such that $\Lambda(\epsilon_{0})>0$.
\end{enumerate}
\end{lemma}

\begin{proof}
From \eqref{eq:def_lambda_0} we see that $\Lambda(\epsilon)\geq0$ requires
\[
0 \leq \epsilon\leq\frac{1}{1-\lambda_{x}}\wedge\frac{4\lambda_{x}}{4\lambda_{x}+R_{0}^{2}}=\frac{4\lambda_{x}}{4\lambda_{x}+R_{0}^{2}} \eqsp,
\]
where the equality follows from $\lambda_{x}>0$, which completes the proof of \ref{lem:derivativeLambda0_a}. The proof of \ref{lem:derivativeLambda0_b} is a simple calculation and is omitted. We now show \ref{lem:derivativeLambda0_c}. If we set $\Lambda'(\epsilon) = 0$, it implies that $\epsilon >0$ satisfies
\begin{equation}
  \label{eq:optim_conv_rate_1_1}
(1+\lambda_{x})-\epsilon R_{1}^{2}  =R^{\half}(\epsilon)(1-\lambda_{x}) \eqsp,
\end{equation}
and imposes the condition $(1+\lambda_{x})-\epsilon R_{1}^{2} \geq 0$ so 
\begin{equation}
 \label{eq:condition_eps_proof}
\epsilon \in \ccint{0,\frac{1+\lambda_{x}}{(1+\lambda_{x})^{2}+R_{0}^{2}}} \eqsp.
\end{equation}
Squaring both sides of \eqref{eq:optim_conv_rate_1_1} implies the following sequence 
of equalities using \eqref{eq:def_R}
\begin{align}
(1-\lambda_{x})^{2}R(\epsilon) & =\left[\epsilon R_{1}^{2}-(1+\lambda_{x})\right]^{2},\\
(1-\lambda_{x})^{2}\left[R_{1}^{2}\epsilon^{2}-2(1+\lambda_{x})\epsilon+1\right] & =R_{1}^{4}\epsilon^{2}-2R_{1}^{2}(1+\lambda_{x})\epsilon+(1+\lambda_{x})^{2} \eqsp,
\end{align}
which is equivalent by \eqref{eq:def_R_1} to 
\begin{align}
R_{1}^{2} \epsilon^{2}\left[(1-\lambda_{x})^{2}-R_{1}^{2}\right]-2\epsilon(1+\lambda_{x})\left[(1-\lambda_{x})^{2}-R_{1}^{2}\right]-4\lambda_{x} & =0\\
\left[(1+\lambda_{x})^{2}+R_{0}^{2}\right] \epsilon^{2}\left[-4\lambda_{x}-R_{0}^{2}\right]-2\epsilon(1+\lambda_{x})\left[-4\lambda_{x}-R_{0}^{2}\right]-4\lambda_{x} & =0\\
\left[(1+\lambda_{x})^{2}+R_{0}^{2}\right]\epsilon^{2}-2(1+\lambda_{x})\epsilon+4\lambda_{x}/(R_{0}^{2}+4\lambda_{x}) & =0 \eqsp.
\end{align}
The two strictly positive roots are 
\begin{align}
\epsilon_{\pm} & =\frac{(1+\lambda_{x})\pm\parentheseDeux{(1+\lambda_{x})^{2}-4\lambda_{x}\{(1+\lambda_{x})^{2}+R_{0}^{2}\}/(R_{0}^{2}+4\lambda_{x})}^{\half}}{(1+\lambda_{x})^{2}+R_{0}^{2}}>0,
\end{align}
where the inequality follows from $\lambda_{x}>0$ and $R_0>0$. Further 
\[
(1+\lambda_{x})^{2}\big(R_{0}^{2}+4\lambda_{x}\big)-4\lambda_{x}\big[(1+\lambda_{x})^{2}+R_{0}^{2}\big]=R_{0}^{2}\big[(1+\lambda_{x})^{2}-4\lambda_{x}\big] = R_0^2 [1-\lambda_x]^2 \eqsp,
\]
and since $\lambda_{x}\leq1$, this yields the simplified expression
for the two roots
\[
\epsilon_{\pm}=\frac{(1+\lambda_{x})\pm(1-\lambda_{x})\parentheseDeux{R_{0}^{2}/(R_{0}^{2}+4\lambda_{x})}^{\half}}{(1+\lambda_{x})^{2}+R_{0}^{2}} \eqsp\eqsp.\]
From the conditions on $\epsilon$ given by \ref{lem:derivativeLambda0_a} and \eqref{eq:condition_eps_proof}, and the fact
that $\lambda_{x}\leq1$, we retain $\epsilon_{0}=\epsilon_{-}$ only. The
last statement follows from the second statement and the fact that
$\Lambda'$ is continuous.
\end{proof}
The following lemma establishes in particular that $\epsilon_{0}$ is
a global maximum.
\begin{lemma}
\label{lem:second-derivative-Lambda0}  Let $\Lambda\colon\mathbb{R}_+^*\to\mathbb{R}$ be defined by \eqref{eq:def_lambda_0}. Then with $\lambda_{x}\in(0,1)$ and $R_0 >0$,
\begin{enumerate}
\item \label{lem:second-derivative-Lambda0_a} for any $\epsilon >0$, $\Lambda''(\epsilon)<0$ (implying concavity),
\item \label{lem:second-derivative-Lambda0_b} $\Lambda$  is maximized at $\epsilon_{0}$ defined by \eqref{eq:def_eta_star}
and $0<\epsilon_{0}\leq (4\lambda_{x})/(4\lambda_{x}+R_{0}^{2})$.
\item  \label{lem:second-derivative-Lambda0_c} If in addition $R_0 \geq 2$, $\epsilon_0 \leq 3 \lambda_x/(4\lambda_x+R_0^2)$.
\end{enumerate}
\end{lemma}

\begin{proof}
  \begin{enumerate}[wide, labelwidth=!, labelindent=0pt]
  \item 
  We differentiate $\epsilon\mapsto-2\Lambda(\epsilon)=-[1-\epsilon(1-\lambda_{x})]+R^{\half}(\epsilon)$
twice, yielding the first order derivative 
\[
\epsilon \mapsto (1-\lambda_{x})+(1/2)R'(\epsilon)R^{-\half}(\epsilon)
\]
and the second order derivative follows
\[
  \epsilon \mapsto 
  (1/4)R^{-3/2}(\epsilon)\left(2R''(\epsilon)R(\epsilon)-[R'(\epsilon)]^{2}\right) \eqsp\eqsp.\]
Now from \eqref{eq:def_R}, $R(\epsilon)=a\psi(\epsilon)$ with $\psi(\epsilon)=(\epsilon-b)^{2}+c$
with all constants $b,c$ non-negative. Further $\psi'(\epsilon)=2(\epsilon-b)$
and $\psi''(\epsilon)=2$ and therefore 
\begin{align}
2\psi''(\epsilon)\psi(\epsilon)-\psi'(\epsilon)^{2} & =4[(\epsilon-b)^{2}+c-(\epsilon-b)^{2}]=4c > 0 \eqsp,
\end{align}
which implies that $\Lambda''(\epsilon) \leq 0$ for any $\epsilon \geq 0$. 
\item 
From the concavity we deduce that $\epsilon_{0}$ is a maximum, and the
inequality on $\epsilon_{0}$ follows from the fact that this is required
for $\Lambda(\epsilon_{0})\geq0$. 
\item  Using that for any $s\geq 0$, $(1+s)^{\half} \leq 1+s/2$, and $4 \lambda_x \leq (1+\lambda_x)^2$, we get that
   \begin{multline}
    \epsilon_0 = R_0 \frac{(1+\lambda_x)(4\lambda_x/R_0^2+1)^{\half} - (1-\lambda_x)}{\big[(1+\lambda_x)^2+R_0^2\big](R_0^2 + 4 \lambda_x)^{\half}} \leq \frac{2 \lambda_x R_0 +2 \lambda_x (1+\lambda_x)/R_0}{\big[(1+\lambda_x)^2+R_0^2\big]^\half (R_0^2 + 4 \lambda_x)} \\\leq \frac{2 \lambda_x +2 \lambda_x (1+\lambda_x)/R_0^2}{R_0^2 + 4 \lambda_x} \eqsp.
  \end{multline}
The assumption $R_0 \geq 2$ completes the proof.
\end{enumerate}
\end{proof}
\begin{proposition}
\label{prop:bound-on-Phi-epsilon-star}The function $\tilde{\alpha}\colon\mathbb{R}_{+}\to\mathbb{R}_{+}$, defined by \eqref{eq:def_alpha_tilde},
has a unique maximizer $\epsilon^{\star}\in \ooint{0,\epsilon_{0}}$, where $\epsilon_0$ is given in \eqref{eq:def_eta_star}. In addition, if $2^\half R_{0}\geq\lambda_{v}$ 
then
\begin{equation}
  \label{eq:prop:bound-on-Phi-epsilon-star}
\tilde{\alpha}(\epsilon_{0})\leq\tilde{\alpha}(\epsilon^{\star})\leq3\tilde{\alpha}(\epsilon_{0})\eqsp.
\end{equation}
\end{proposition}

\begin{proof}
First note that for any $\epsilon \geq 0$,

\[
\tilde{\alpha}'(\epsilon)=\frac{\Psi(\epsilon)}{(1+2^\half\lambda_{v}\epsilon)^{2}}\eqsp,
\]
with
\begin{equation}
  \Psi(\epsilon)=\Lambda'(\epsilon)(1+2^\half \lambda_{v}\epsilon)-2^\half \lambda_{v}\Lambda(\epsilon)\eqsp.
\end{equation}
Then  from \Cref{lem:second-derivative-Lambda0}, for any $\epsilon \geq 0$
\begin{align}
  \label{eq:def_psi_optim}
 \Psi'(\epsilon)=(1+2^\half\lambda_{v}\epsilon)\Lambda''(\epsilon)<0 \text{ and }\Psi(\epsilon_0) = - 2^\half \lambda_{v}\Lambda(\epsilon_0) < 0.  
\end{align}
Together with $\Psi(0)=\Lambda'(0)=\lambda_{x}>0$, 
and the fact that $\epsilon\to\Psi(\epsilon)$ is continuous, we
deduce the existence and uniqueness of $\epsilon^{\star}\in(0,\epsilon_{0})$
satisfying $\tilde{\alpha}'(\epsilon^{\star})=0$, and maximizing $\tilde{\alpha}$
on $\mathbb{R}_{+}$. Further since $\tilde{\alpha}'(\epsilon^{\star})=0$ and
$\epsilon\mapsto\Psi(\epsilon)$ is non-increasing, using the first equality of \eqref{eq:def_psi_optim} and the definition of $\tilde{\alpha}$ given in \eqref{eq:def_alpha_tilde},  we deduce 
\[
\sup_{\epsilon\in[\epsilon^{\star},\epsilon_{0}]}\vert\tilde{\alpha}'(\epsilon)\vert\leq\frac{\vert\Psi(\epsilon_{0})\vert}{(1+2^\half\lambda_{v}\epsilon^{\star})^{2}}=2^\half\lambda_{v}\frac{1+2^\half\lambda_{v}\epsilon_{0}}{(1+2^\half\lambda_{v}\epsilon^{\star})^{2}}\tilde{\alpha}(\epsilon_{0})\eqsp,
\]
From Taylor's theorem, we obtain
\[
\tilde{\alpha}(\epsilon^{\star})-\tilde{\alpha}(\epsilon_{0})\leq(\epsilon_{0}-\epsilon^{\star})2^\half\lambda_{v}\frac{1+2^\half\lambda_{v}\epsilon_{0}}{(1+2^\half\lambda_{v}\epsilon^{\star})^{2}}\tilde{\alpha}(\epsilon_{0})\eqsp,
\]
from which we conclude that 
\[
\tilde{\alpha}(\epsilon_{0})\leq\tilde{\alpha}(\epsilon^{\star})\leq\left[1+(\epsilon_{0}-\epsilon^{\star})2^\half \lambda_{v}\frac{1+2^\half\lambda_{v}\epsilon_{0}}{(1+2^\half\lambda_{v}\epsilon^{\star})^2}\right]\tilde{\alpha}(\epsilon_{0})\eqsp.\]
Now if we use $2^\half R_{0}\geq\lambda_{v}$ we have by \eqref{eq:def_eta_star}
that 
\[
\lambda_{v}\epsilon_{0}<\frac{(1+\lambda_{x})\lambda_{v}}{(1+\lambda_{x})^{2}+R_0^2}\leq \lambda_v (2R_0)^{-1} \leq 2^{-\half} \eqsp,
\]
implying 
\[
(\epsilon_{0}-\epsilon^{\star})2^\half\lambda_{v}\frac{1+2^\half \lambda_{v}\epsilon_{0}}{(1+2^\half\lambda_{v}\epsilon^{\star})^{2}}\leq2^\half\lambda_{v}\epsilon_{0}(1+2^\half\lambda_{v}\epsilon_{0})\leq 2 \eqsp,
\]
which completes the proof of \eqref{eq:prop:bound-on-Phi-epsilon-star}.
\end{proof}


\section{Some results on closed operators on Hilbert spaces}
\label{ss:spctral bounds}

In this section we gather classical results concerning densely defined closed operators on a Hilbert space to which we repeatedly refer throughout the manuscript.

We start this section with a well-know result regarding the closure of anti-symmetric operators, for which a proof is given for completeness.
\begin{lemma}
  \label{lem:closure_anti_symm}
  Let $(\mct,\domain(\mct))$ be a densely defined anti-symmetric operator on a Hilbert space $\msh$, of inner product $\lang \cdot, \cdot \rang $ and induced norm $\norm{\cdot}$. In addition, let $\mca$ be a bounded operator on $\msh$. Then, $(\mct, \domain(\mct))$ and $(\mct \mca , \domain(\mct \mca))$ are closable operators. In particular, if $\mca \domain(\mct) \subset \domain(\mct)$, then $(\mct \mca , \domain(\mct))$ is closable.
\end{lemma}
\begin{proof}
  Since $\mct$ is densely defined, its adjoint
  $(\mct^{\star},\domain(\mct^{\star}))$ is well-defined and closed by
  \cite[Theorem 5.1.5]{pedersen1995analysis} and since $\mct$ is
  anti-symmetric, $(\mct^{\star},\domain(\mct^{\star}))$ is therefore
  a closed extension of $(\mct,\domain(\mct))$ which implies that
  $(\mct,\domain(\mct))$ is closable. Finally, it is easy to
  verify that $(\mct^{\star} \mca, \domain(\mct^{\star} \mca))$ is
  closed since $(\mct^{\star},\domain(\mct^{\star}))$ is and is an extension of $(\mct \mca, \domain(\mct \mca))$. This completes the proof. 
\end{proof}

\begin{proposition}
\label{prop:abstract bound}
Let $\calB$ be a closed and densely defined operator on a Hilbert space $\msh$ of inner product $\lang \cdot, \cdot \rang $, induced norm $\norm{\cdot}$ and operator norm $ \normop{\cdot}$.
\begin{enumerate}
\item \label{prop:abstract bound_a} $\Id+ \mcbb^{\star} \mcbb$ is a positive self-adjoint  operator on $\msh$ bijective from $\domain(\mcbb^{\star} \mcbb)$ to $\msh$. In addition, $(\Id+ \mcbb^{\star} \mcbb)^{-1}$ is a positive self-adjoint bounded operator on $\msh$ and $\mcbb (\Id+ \mcbb^{\star} \mcbb)^{-1}$ is a bounded operator.
\item \label{prop:abstract bound_b} For any $h \in \msh$,
 \begin{equation}
\| (\Id+\calB^{\star} \calB)^{-1} h \|^2 + 2 \, \| \calB (\Id+\calB^{\star} \calB)^{-1} h\|^2 \leq \| h \|^2 \eqsp.
 \end{equation}
\item \label{prop:abstract bound_c} $\mcbb^{\star} \mcbb (\Id+ \mcbb^{\star} \mcbb)^{-1}$ is a bounded operator on $\msh$ which satisfies
  \begin{equation}
    \normop{\mcbb^{\star} \mcbb (\Id+ \mcbb^{\star} \mcbb)^{-1}} \leq 1 \eqsp.
  \end{equation}
 \item \label{prop:abstract bound_d} The operator 
$((\Id + \calB^{\star} \calB)^{-1} \calB^{\star},\domain(\calB^{\star}))$  is closable, its closure is a bounded operator  and $\normopLine{\overline{(\Id + \calB^{\star} \calB)^{-1} \calB^{\star}}} \leq 1$. 
\end{enumerate}
\end{proposition}
\begin{remark}
\label{rem:abstract bound}
 Note that under the condition of \Cref{prop:abstract bound}, we get that $(\Id + \calB^{\star} \calB)^{-1} \calB^{\star}$ can be extended to a bounded operator and 
\begin{equation}
\normop{(\Id + \calB^{\star} \calB)^{-1} } \le 1 \eqsp, \quad \normop{ \calB (\Id + \calB^{\star} \calB)^{-1} }  \le 1/2^{\half} \eqsp. 
 \end{equation} 
\end{remark}
\begin{proof}
  \ref{prop:abstract bound_a} and \ref{prop:abstract bound_b} follow from \cite[Theorem 5.1.9]{pedersen1995analysis} and inspection of the proof.   We now show \ref{prop:abstract bound_c}.

 First note that $(\Id + \calB^{\star} \calB -\Id) (\Id + \calB^{\star} \calB)^{-1} = \Id - (\Id + \calB^{\star} \calB)^{-1}$, from which we deduce that it is a self-adjoint and bounded operator by the triangle inequality with norm less or equal than $2$. To prove the tighter upper bound we use \cite[Proposition 3.2.27 p. 99]{pedersen1995analysis} (twice), the identity for any $h \in \msh$
 \begin{align}
   & \abs{\psH{ \calB^{\star} \calB (\Id + \calB^{\star} \calB)^{-1} h }{ h } }\\
   &\qquad \qquad = \max\defEns{\|h\|^2 -\psH{ (\Id + \calB^{\star} \calB)^{-1}h}{ h }, \psH{ (\Id + \calB^{\star} \calB)^{-1}h}{ h }- \norm{h}^2} \eqsp, 
 \end{align}
 that $(\Id + \calB^{\star} \calB)^{-1}$ is positive and $\vvvert (\Id + \calB^{\star} \calB)^{-1} \vvvert \leq 1$ from the first statement. 
 
It remains to prove \ref{prop:abstract bound_d}.  Since $\mcbb$ is closed and densily defined, $\domain(\mcbb^{\star})$ is dense and therefore $\{(\Id + \calB^{\star} \calB)^{-1} \calB^{\star}\}^{\star}$ is closed and densely defined by \cite[Theorem 5.1.5]{pedersen1995analysis}. By \ref{prop:abstract bound_a}, we have for any $h_1 \in \domain(\mcbb^{\star})$ and  $h_2 \in \msh$, we have
  \begin{equation}
    \psmu{(\Id + \calB^{\star} \calB)^{-1} \calB^{\star}h_1}{h_2} = \psmu{h_1}{\mcbb (\Id + \calB^{\star} \calB)^{-1} h_2} \eqsp,
  \end{equation}
  which implies that $\{(\Id + \calB^{\star} \calB)^{-1} \calB^{\star}\}^{\star} = \mcbb (\Id + \calB^{\star} \calB)^{-1}$. Therefore, the operator  $\{(\Id + \calB^{\star} \calB)^{-1} \calB^{\star}\}^{**}$ is bounded on $\msh$. The proof then follows by \cite[Theorem 5.1.5]{pedersen1995analysis}  which implies that $(\Id + \calB^{\star} \calB)^{-1} \calB^{\star}$ is closable and
  \begin{equation}
    \overline{(\Id + \calB^{\star} \calB)^{-1} \calB^{\star}} = ((\Id + \calB^{\star} \calB)^{-1} \calB^{\star})^{**} \eqsp.
  \end{equation}
\end{proof}
A similar result can be obtained by using that $\mcbb$ is closable only, as a consequence of the following lemma.
\begin{lemma}
  \label{lem:def_inverse_dms}
Assume that $(\mcbb,\domain(\mcbb))$ is a densely defined closable operator. Let $(\overline{\mcbb}, \domain(\overline{\mcbb}))$   be the closure of $(\mcbb,\domain(\mcbb))$ and $m>0$.  Then,
 the conclusions of \Cref{prop:abstract bound} hold changing $\mcbb$ to $\overline{\mcbb}$. 
\end{lemma}
\begin{proof}
This result is a just a consequence of  \cite[Theorem 5.1.5]{pedersen1995analysis} which implies that $\mcbb^{\star}$ is densely defined, $\overline{\mcbb}=(\mcbb^{\star})^{\star}$ and $\mcbb^{\star} = \overline{\mcbb}^{\, \star}$.
\end{proof}

We would like to apply \Cref{prop:abstract bound} to the densely defined and closed operator $m^{-\half} \nabla_x$ for $m >0$, which does not fully fit in the framework of \Cref{prop:abstract bound} since it is an operator from $\mrl^2(\pi)$ to $\mrl^2(\pi)^d$.  This is easily fixed upon noting that the  operator $\nabla_x$ on $\mrl^2(\pi)$  can be extended as an operator on $\mrl^2(\pi)^d$ as follows: for any $\tilde{f} = (f_1,\ldots,f_d) \in \mrl^2(\pi)^d$, $f_1 \in \domain(\nabla_x)$, define $\mcbb \tilde{f} = m^{-\half} \nabla_x f_1$. Then, a direct consequence of \Cref{prop:abstract bound} applied to the operator $\mcbb$ for $m >0$, on $\mrl^2(\pi)^d$ is the following taking $\tilde{f} = (f,0,\ldots,0)$, for $f \in \mrl^2(\pi)$.
\begin{corollary}
 \label{coro:bound regularization}
 Let $m >0$. 
The operators $\nax (m \Id +\nax^{\star} \nax)^{-1}$ and $\nax^{\star} \nax (m \Id +\nax^{\star} \nax)^{-1}$ are bounded on $\mrl^2(\pi)^d$ with 
\begin{equation}
\normoppi{ \nax ( m \Id +\nax^{\star} \nax)^{-1} } \le 1/(2m)^{\half} \eqsp, \, \, \normoppi{\nax^{\star} \nax (m \Id +\nax^{\star} \nax)^{-1}} \leq 1 \eqsp.
\end{equation}
In addition, for any $\ff \in \mrl^2(\pi)$,
\begin{equation}
\label{eq:inequality lemma regul}
	\normmu{(m\Id +\nax^{\star} \nax)^{-1} \ff }^2 + (2/ m) \, \normmu{ \nax (m \Id +\nax^{\star} \nax)^{-1} \ff }^2 \le \{ \normmu{ \ff }/m\}^2 \eqsp,
 \end{equation}
 and
 \begin{equation}
 \label{eq:inequality lemma regul_b}
 \normmu{\nax^{\star} \nax (m \Id +\nax^{\star} \nax)^{-1}f} \leq \normmu{f} \eqsp.
 \end{equation}
\end{corollary}
We conclude this section by the following  results  which can be found in \cite{GS2014}.
     \begin{lemma}[\protect{\cite[Lemma 2.2]{GS2014}}] 
        \label{lem:grothaus_lemme_2_3}
        Let $(\mct,\domain(\mct))$ be a anti-symmetric operator on $\rml^2(\mu)
        $ and $\Pi$ be an orthogonal projection on $\mrl^2(\mu)$. Assume that there exists $\msdd \subset \domain(\mct)$ such that $\Pi(\msdd) \subset \domain(\mct)$ and $\msdd$ is dense in $\mrl^2(\mu)$. Then the following statements hold.
        \begin{enumerate}[label=(\alph*)]
        \item         \label{lem:grothaus_lemme_2_3_a} $\domain(\mct) \subset \domain((\mct \Pi)^{\star})$ and for any $f \in \domain(\mct)$, $(\mct \Pi)^{\star}f= -\Pi \mct f$.
        \item \label{lem:grothaus_lemme_2_3_b} For any $f \in \domain((\mct \Pi)^{\star})$, $\Pi(\mct \Pi)^{\star} f = (\mct \Pi)^{\star}f$.
        \end{enumerate}
      \end{lemma}

\section{Elliptic regularity estimates}
\label{sec:ellipt-regul-estim}


We preface this section with some complements on the adjoint of $\nabla_x$ seen as an operator on $\mrl^2(\pi)^d$.

\begin{lemma} \label{lem:adjoint-nax}
  Assume \Cref{as:U}. 
Consider the operator $(\nabla_x,\domain(\nax))$ from the Hilbert space $\mrl^2(\pi)$ to $\rml^2(\pi)^d$ endowed with the inner product defined by \eqref{eq:def-L2-inner-product}. Then it holds
\begin{enumerate}[label=(\alph*)]
\item \label{lem:adjoint-nax_a_0} for any $i \in \{1,\ldots,d\}$, the $\mrl^2(\pi)$-adjoint of $\partial_{x_i}$ is given for any $g \in \rmC^{1}_{\poly}(\msx)$ by 
\begin{equation}
	\partial_{x_i}^{\star} g = - \partial_{x_i} g + g \partial_{x_i} U \eqsp;
\end{equation}
\item \label{lem:adjoint-nax_a} the $\mrl^2(\pi)$-adjoint of $\nax$ is given for any $G \in \rmC^1_{\poly}(\msx,\rset^d)$ by 
\begin{equation}
	\nax^{\star}G = -\divx G + \nabla_x U^\top G \eqsp.
\end{equation}
\end{enumerate}
\end{lemma}
\begin{remark}
Note that \Cref{lem:adjoint-nax} implies that for any $g \in \rmC^{2}_{\poly}(\msx)$ and $G \in \rmC^{2}_{\poly}(\msx,\rset^d)$, we have 
\begin{equation}
 \label{eq:commutation_nabla}
\nax^{\star} \nax g =-\Delta_x g +\nax U^\top\nax g \text{ and }
\nax \nax^{\star} G =\nax^{\star}\nax G + \nax^{2}U G \eqsp,
\end{equation}
where we have defined $\nax^{\star}\nax G \in \rmC_{\poly}(\mse,\rset^d)$ for any $(x,v) \in \mse$ and $i \in \{1,\ldots,d\}$ by
\begin{equation}
 \label{eq:1}
\{ \nax^{\star}\nax G(x,v)\}_i = \nax^{\star} \partial_{x_i} G(x,v) = \sum_{j=1}^d - \partial_{x_j,x_i} G_j(x,v) + \partial_{x_j}U(x) \partial_{x_i}G(x,v) \eqsp.
\end{equation} 
\end{remark}
\begin{proof}
The proof just follows by integration by parts.
\end{proof}




\begin{proposition}
Let $m >0$ and assume \Cref{as:U}. 
\label{prop:bound hessian}
Then for any $f \in \rmCb^2(\mse)$,
\begin{equation}
 \label{eq:def_kappa_a}
	\normmuLine{ \nax^2 (m\Id +\nax^{\star} \nax)^{-1} \Piv f } \le \kappavara \normmuLine{\Piv f} \quad {\rm where} \quad \kappavara = (1 + c_1/(2m))^{\half} \eqsp.
\end{equation}
\end{proposition}

\begin{proof}
 
Let $f \in \rmCb^{2}(\mse)$ and consider $u =(m \Id+\nax^{\star} \nax)^{-1} \Piv f$. By \cite[Theorem 2]{pardoux2001}, $u \in \rmC^{3}_{\poly}(\msx)$. Therefore we obtain by \eqref{eq:commutation_nabla}, \eqref{eq:borne_moment} and the fact that $U \in \rmC^3_{\poly}(\msx)$ using \Cref{as:U}, 
\begin{align}
\normmuLine{ \nax^2 u }^2 &= \pspiLine{ \nax^2 u} { \nax^2 u } = \pspiLine{ \nax u}{ ( \nax^{\star} \nax)[ \nax u] }
\\ 
                          &= \pspiLine{ \nax u}{ (\nax \nax^{\star})[ \nax u] - \nax^2 U \nax u } \\
 \label{eq:borne_nax_u_2}
& = 	\normmuLine{\nax^{\star} \nax u}^2 -\pspiLine{ \nax u}{ \nax^2 U \nax u} \eqsp.
\end{align}
From the definition of $u$, using \Cref{coro:bound regularization} and \Cref{as:U}-\ref{item:condition_hessian} we conclude that
\begin{equation}
\normmuLine{\nax^2 u}^2 \leq \normmuLine{ \Piv f }^2 + c_1 \normmuLine{\nax u}^2 
 \leq \normmuLine{f }^2 + c_1 \normmuLine{\Piv f }^2 /(2m) \eqsp.
\end{equation}
\end{proof}

In order to bound terms of the form $\| F_k^\top \nax u \|$ in \Cref{app:bound RT RS} we need the following Lemma which is a quantitative version of \cite[Lemma 6]{Dolbeault15}.
Consider the function $W : \rset^d \to \rset_+$ defined for any $x \in \rset^d$ by
\begin{equation}
 \label{eq:def_W}
 W(x) = \defEns{1+\abs{\nax U(x)}^2}^{\half} \eqsp.
\end{equation}
\begin{lemma}[\protect{\cite[Lemma 6]{Dolbeault15}}]
\label{lemma:dms 6}
 Assume \Cref{as:U}. Then for any $\varphi \in \domain(\nax)$,
\begin{equation}
\normmu{ \nax \varphi } \geq \left[ 4\left(1 + c_2 d^{1+\varpi} /(4\Cp^2) \right)^{\half} \right]^{-1} \normmu{ \varphi \nax U } \eqsp,
\end{equation}
where $c_2$ and $\Cp$ are defined in \eqref{eq:laplacian bound} and \eqref{eq:poincare_assumption} respectively. 
As a corollary, it holds for any $\varphi \in \domain(\nax)$,
\begin{align}
 \label{eq:def_kappab}
  \normmu{ \nax \varphi } &\ge \kappavarb \normmu{ \varphi W } \eqsp,\\
  {\rm where \, \, \, } \kappavarb^{-1} &= \left( \Cp^{-2} + 16(1 + c_2 d^{1+\varpi} /(4\Cp^2)) \right)^{\half} \\ &= \Cp^{-1} \left(1+4 c_2 d^{1+\varpi}  + 16 \Cp^2 \right)^{\half} \ge \Cp^{-1} \eqsp.
\end{align}
\end{lemma}
\begin{proof}
Note that we only need to consider $\varphi \in \rmC^{\infty}_{\rmc}(\msx)$ since $\rmC^{\infty}_{\rmc}(\msx)$ is a core for $(\nax,\domain(\nax))$. 
First since $\nax U \in \mrl^2(\mu)$, for any $\varepsilon > 0$, we get
\begin{equation}
 \label{eq:lemma_dms_6_1}
2 \psmu{\varphi \nax U}{\nax \varphi} \leq \veps^{-1}\normmu{\nax \varphi}^2 + \veps \normmu{\varphi \nax U}^2 \eqsp. 
\end{equation}
We then bound from below the left-hand side. Using the \textit{carré
 du champ} identity, \ie~for any $f,g \in \rmC_{\poly}^{2}(\msx)$,
$\psmu{\nax f}{\nax g} = \psmu{\ps{\nax U}{\nax f} - \Delta_x f}{g}$, we get using that $\nax [\varphi^2] = 2 \varphi \nax \varphi$,
\begin{equation}
 2	\psmu{ \varphi \nax U}{ \nax \varphi } = \psmu{ \nax [\varphi^2]}{ \nax U } = \normmu{ \varphi \nax U }^2 - \psmu{ \varphi^2
 }{ \Delta_x U} \eqsp. 
\end{equation}
By \eqref{eq:laplacian bound} and \eqref{eq:poincare_assumption}, we obtain 
\begin{align}
  2	\psmu{ \varphi \nax U}{ \nax \varphi }& \geq \normmu{ \varphi \nax U }^2/2 - c_2 d^{1+\varpi}  \normmu{\varphi}^2 \\
  & \geq \normmu{ \varphi \nax U }^2/2 - (c_2 d^{1+\varpi} /\Cp^2) \normmu{\nax \varphi}^2 \eqsp. 
\end{align}
From this result and \eqref{eq:lemma_dms_6_1}, it follows that
\begin{equation}
\normmu{ \varphi \nax U }^2/2 - (c_2 d^{1+\varpi} /\Cp^2) \normmu{\nax \varphi}^2 \leq \veps^{-1}\normmu{\nax \varphi}^2 + \veps \normmu{\varphi \nax U}^2 \eqsp.
\end{equation}
Rearranging terms and setting $\varepsilon = 1/4$ completes the proof. The last statement is a direct consequence of the first one using the definition of $W$ in \eqref{eq:def_W}. 
\end{proof}

Putting this with Proposition~\ref{prop:bound hessian}, this implies the following.

\begin{corollary}
\label{coro:bound Fk}
Let $m >0$ and assume \Cref{as:U} and \Cref{as:Fk}. For any $f \in \Lmu$ and $k \in \{ 1, \ldots, K \}$, we have
\begin{align}
  \normmu{ F_k^\top \{\nax (m \Id +\nax^{\star} \nax)^{-1} \Piv f \} }& \leq 2^{\half} a_k 	\normmu{ W \{\nax (m \Id +\nax^{\star} \nax)^{-1} \Piv f \} } \\
 &  \leq \frac {2^{\half} a_k \kappavara} {\kappavarb} \normmu{ \Piv f} \eqsp,
 \end{align}
 where $a_k$, $W$, $\kappa_1$ and $\kappa_2$ are defined by \eqref{eq:def_a_k}, \eqref{eq:def_W}, \eqref{eq:def_kappa_a} and \eqref{eq:def_kappab} respectively.
\end{corollary}

\begin{proof}
Note first that since $\nax (m \Id +\nax^{\star} \nax)^{-1}$ is a bounded operator by  \Cref{coro:bound regularization}, it is sufficient by density to show this result for $f \in \rmCb^{2}(\mse)$. Let $f \in \rmCb^{2}(\mse)$ and $u = (m +\nax^{\star} \nax)^{-1} \Piv f$. By \cite[Theorem 2]{pardoux2001}, $u \in \rmC_{\poly}^3(\msx)$.
Second since for any $t,s \ge 0, \, s + t \le 2^{\half} \sqrt{s^2 + t^2}$, \Cref{as:Fk}-\ref{item:Fk bounded} implies for any $x \in \msx$, 
\begin{equation}
	|F_k| (x) \le a_k (1 + |\nax U|(x)) \le 2^{\half}a_k W(x) \eqsp. 
\end{equation}
Therefore using \Cref{lemma:dms 6} and \Cref{prop:bound hessian} successively, we obtain
\begin{align}
  \normmu{ F_k^\top \nax u } &\le \normmu{ \, |F_k| \, \nax u \, } \le 2^{\half} a_k \normmu{ W \nax u } = 2^{\half} a_k \parenthese{ \sum_{i=1}^d \normmu{ W \partial_{x_i} u }^2}^{1/2} \\ &\leq (2^{\half} a_k / \kappavarb) \parenthese{ \sum_{i=1}^d \normmu{ \nax \parentheseDeux{ \partial_{x_i} u} }^2}^{1/2} = (2^{\half} a_k / \kappavarb) \normmu{ \nax^2 u }\\
  &\le (2^{\half} a_k \kappavara / \kappavarb) \normmu{ \Piv f } \eqsp.
\end{align}
\end{proof}


\section{Supplementary material}
\subsection{Radial distributions}
\label{app:radial}

The following gathers standard results on spherically symmetric distributions on $\bbR^d$ for which we could not find a single reference. In particular we establish that \Cref{as:radial}-\ref{item:as:radial:general} and conditions required in Lemma \ref{lemma:velocities norm bound} are satisfied in this scenario.
\begin{lemma} \label{lem:moments-spherically-symmetric}
Let $d\geq2$. 
\begin{enumerate}
\item Assume $\measv$ is the uniform distribution on the unit hypersphere $\bbS^{d-1}$ ,
then 
\begin{enumerate}[wide, labelwidth=!, labelindent=0pt]
\item for $i,j,k,l\in\{ 1,\ldots, d\}$ such that $\card(\{i,j,k,l\})>2$, we have
$\int_{\sphere^{d-1}} v_{i}v_{j}v_{k}v_{l}\, \rmd \nu (v)=0$,
\item otherwise,
\begin{align}
  m_{2}&=\frac{1}{d} \eqsp, \qquad m_{2,2}=\int_{\sphere^{d-1}} v_{1}^2v_{2}^2\, \rmd \nu (v)=\frac{1}{d(d+2)}\\
   \text{ and } \qquad m_{4}& =\frac{1}{3}\int_{\sphere^{d-1}} v_{1}^{4} \, \rmd \nu(v)=\frac{1}{d(d+2)} \eqsp.
\end{align}
\end{enumerate}
\item For any spherically symmetric distribution $\nu$ \ie~corresponding
to random variables $V=B^{\half} W$ for $W$ uniformly distributed on
the unit hypersphere $\bbS^{d-1}$ and $B$ a non-negative random variable independent
of $w$ and of first and second order moments $\gamma_1$ and $\gamma_2$ respectively,
\begin{enumerate}
\item for $i,j,k,l\in\{ 1,\ldots, d\}$ such that $\card(\{i,j,k,l\})>2$, we have
$\int_{\rset^d} v_{i}v_{j}v_{k}v_{l}\, \rmd \nu (v)=0$,
\item otherwise,
\[
m_{2}=\frac{\gamma_1}{d}\eqsp, \qquad m_{2,2}=\frac{\gamma_2}{d(d+2)}\text{ and } \qquad m_{4}=\frac{\gamma_2}{d(d+2)} \eqsp. 
\]
\end{enumerate}
\end{enumerate}
\end{lemma}
\begin{remark} Naturally the zero-mean $d$-dimensional Gaussian distribution on $\rset^d$ with covariance matrix $\Idd$. corresponds to $B$ distributed according to $\chi^{2}(d)$, in which case $m_4=m_{2,2}=m_2^2$.
\end{remark}
\begin{proof}
We use the polar parametrization of the multivariate normal distribution.
Let 
\[
v(\phi)=\big(\cos\phi_{1},\sin\phi_{1}\cos\phi_{2},\ldots,\cos(\phi_{k})\prod_{i=1}^{k-1}\sin(\phi_{i}),\ldots,\prod_{i=1}^{d-1}\sin(\phi_{i})\big),
\]
 $\phi\in[0,\uppi]^{d-2}\times[0,2\uppi]$. The probability distribution
for $\phi$ ensuring uniformity of $v(\phi)$ on the surface of the
$d$-sphere has density
\[
f_{\sphere}(\phi)\propto\prod_{i=1}^{d-2}\sin^{d-i-1}(\phi_{i}) \1_{ [0,\uppi]^{d-2}\times [0,2\uppi]} (\phi) \eqsp,
\]
with respect to the Lebesgue measure on $\rset^{d-1}$. Let $\Phi$ be random variable with distribution $f_{\sphere}$.
Further let $B\sim\chi^{2}(d)$ be independent of $\Phi$ then
it is standard knowledge that $W=B^{\half}v(\Phi)$ follows the zero-mean $d$-dimensional Gaussian distribution on $\rset^d$ with covariance matrix $\Idd$.
Therefore, by construction, 
\begin{align}
  &\mathbb{E}\big[W_{i}W_{j}W_{k}W_{l}\big]=\mathbb{E}\big[B^{2}v_{i}(\Phi)v_{j}(\Phi)v_{k}(\Phi)v_{l}(\Phi)\big]\\
  &  \quad =\mathbb{E}\big[B^{2}\big]\mathbb{E}\big[v_{i}(\Phi) v_{j}(\Phi)v_{k}(\Phi) v_{l}(\Phi)\big]=d(d+2)\mathbb{E}\big[v_{i}(\Phi) v_{j}(\Phi) v_{k}(\Phi) v_{l} (\Phi)\big] \eqsp,
\end{align}
and the latter term vanishes when the leftmost term does. We also
deduce that
\begin{equation}
\mathbb{E}\big[W_{1}^{2}\big]\mathbb{E}\big[W_{2}^{2}\big] =\mathbb{E}\big[W_{1}^{2}W_{2}^{2}\big]
 =d(d+2)\mathbb{E}\big[v_{1}^{2} (\Phi) v_{2}^{2} (\Phi)\big] \eqsp,
\end{equation}
from which we obtain $\mathbb{E}\big[v_{1}^{2}(\Phi) v_{2}^{2} (\Phi)\big]$. Similarly
using properties of the moments of the normal distribution,
\begin{equation}
3\mathbb{E}\big[W_{1}^{2}\big]^{2} =\mathbb{E}\big[W_{1}^{4}\big]
 =d(d+2)\mathbb{E}\big[v_{1}^{4}(\Phi)\big] \eqsp,
\end{equation}
leading to the expression for $\mathbb{E}\big[v_{1}^{4} (\Phi)\big]$. The
last statement is straightforward.

\end{proof}


\section{Expectation of quadratic forms of the velocity}

This section provides expressions for second order moments of quadratic forms of $v$ for a large class of distributions for which we could not find adequate references.

 \begin{lemma}
\label{lemma:velocities norm bound} Let $M \in \bbR^{d \times d}$ be a symmetric matrix, $c \in \bbR$ and assume the distribution $\measv$ of $v$ is such that 
\begin{enumerate}[wide, labelwidth=!, labelindent=0pt]
\item   for any bounded and measurable function $f :\bbR^2 \to \bbR$, $i,j \in \{1,\ldots, d \}$ such that $i \neq j$, $\int f(v_i,v_j) \, \rmd \nu(v) = \int f(v_1,v_2) \, \rmd \nu(v)$
\item for $i,j,k,l\in\{1,\ldots,d\}$, we have
  \begin{equation}
    \int v_{i}v_{j}v_{k}v_{l}\, \rmd \nu (v)=0 \eqsp,
  \end{equation}
  whenever $\card(\{i,j,k,l\})>2$.
\end{enumerate}
\noindent Then
\begin{equation}
	\left\| v^\top M v - c \right\|_\nu^2 = 3(m_4 - m_{2,2}) \Tr(M \odot M) + \left(m_2 \Tr(M) - c \right)^2 + 2 m_{2,2} \Tr(M^2),
\end{equation}
where $\odot$ denotes the Hadamard product.
\end{lemma}

\begin{proof}
Using that $M$ is symmetric, and the expectation symbol for expectations with respect to $\nu$,
\begin{equation}
\begin{aligned}
  \bbE \left[ \left(\sum_{i,j=1}^d M_{ij} v_i v_j - c \right)^2 \right] &= \sum_{i,j,k,\ell=1}^d M_{ij} M_{k\ell} \bbE[v_i v_j v_k v_\ell] \\
  & \qquad \qquad - 2 c \sum_{i,j=1}^d M_{ij} \bbE[v_i v_j] + c^2
\end{aligned}
\end{equation}
where
\begin{equation}
\begin{aligned}
&	\sum_{i,j,k,\ell=1}^d M_{ij} M_{k\ell} \bbE[v_i v_j v_k v_\ell] = 3 m_4 \sum_{i=1}^d M_{ii}^2 + m_{2,2} \sum_{i\neq j} M_{ii} M_{jj} + 2 m_{2,2} \sum_{i\neq j} M_{ij}^2 \\
	&\qquad \qquad \qquad = (3 m_4 - 3 m_{2,2}) \sum_{i=1}^d M_{ii}^2 + m_{2,2} \sum_{i,j=1}^d \left( M_{ii} M_{jj} + 2 M_{ij}^2 \right) \\
	&\qquad \qquad \qquad  = (3 m_4 - 3 m_{2,2}) \Tr(M \odot M) + m_{2,2} \left( \Tr(M)^2 + 2 \Tr(M^2) \right).
\end{aligned}
\end{equation}
Therefore
\begin{equation}
\begin{aligned}
	\bbE \left[ \left(\sum_{i,j=1}^d M_{ij} v_i v_j - c \right)^2 \right] &= (3 m_4 - 3 m_{2,2}) \Tr(M \odot M) + m_{2,2} \Tr(M)^2  \\
	&\qquad + 2 m_{2,2} \Tr(M^2)- 2 c m_2 \Tr(M) + c^2\eqsp,
\end{aligned}
\end{equation}
which implies the desired result.
\end{proof}

\begin{corollary}
\label{coro:velocities norm bound}
Given a symmetric matrix $M \in \bbR^{d \times d}$ and a constant $c \in \bbR$,
\begin{equation}
	\left\| v^\top M v - m_2 \Tr(M) \right\|_\nu \le \sqrt{2m_{2,2} + 3 (m_4 - m_{2,2})_+} | M |.
\end{equation}
\end{corollary}


\subsection{Examples of potentials}
\label{sec:examples-potentials}
\begin{lemma} \label{lemma:potential:independent}
Assume that the potential $U$ is defined for any $x \in \msx$ by $U(x)=\sum_{i=1}^{d}\big(1+x_{i}^{2}\big)^{\beta}/2$,
for $\beta\geq1$. Then $U$ is strongly convex and  there exists $c_{2}>0$, dependent on $\beta$ only, such that  \eqref{eq:laplacian bound} is satisfied with $\varpi = 0$.
\end{lemma}

\begin{proof}

We have for $i,j\in\{1,\ldots,d\}$ and $x \in \msx$,
\[
\big[\nabla_{x}U(x)\big]_{i}=\beta x_{i}\big(1+x_{i}^{2}\big)^{\beta-1}\text{ and }\big[\nabla_{x}^{2}U(x)\big]_{i,j}=\beta[1+(2\beta-1)x_{i}^{2}]\big(1+x_{i}^{2}\big)^{\beta-2} \delta_{i,j} \eqsp,
\]
leading to $\nabla_x^2 U(x) \succeq \beta \Idd$,
and the strong convexity follows.   Using that $\beta \geq 1$ and for
any $s \geq 0$ and $c >0$,
$(1+s^2)^{\beta-2} s^{2} \leq (1+c^2)^{\beta-2}c^2\1_{\ccint{0,c}}(s)+(1+s^2)^{2\beta-2} s^{2}/(1+c^2)^{\beta}\1_{\ooint{c,\plusinfty}}(s)$ and $ (1+s^2)^{\beta-2} \leq  \{1\vee(1+c^2)^{\beta-2}\}  \1_{\ccint{0,c}}(s) + (1+s^2)^{2\beta-2} (s/c)^{2} \1_{\ooint{c,\plusinfty}}(s)$, we get for any $x \in \msx$,
  \begin{align}
    \Delta_x U(x) = \trace(\nabla_x^2U(x)) &= \beta \sum_{i=1}^d
    [1+(2\beta-1)x_{i}^{2}]\big(1+x_{i}^{2}\big)^{\beta-2} \\ &\leq
    \beta d \parentheseDeux{\{1\vee (1+c^2)^{\beta-2}\} + (2\beta-1)(1+c^2)^{\beta-2}c^2} \\
    &  \quad + \beta^{-1}\abs{\nabla_x U(x)}^2\parentheseDeux{c^{-2} + (2\beta-1)(1+c^2)^{-\beta}} \eqsp,
  \end{align}
 which with $c \geq (2 \beta^{-\half})\vee 2^{1/\beta}$ completes the proof.

\end{proof}

\begin{lemma} \label{lemma:potential:polynomial}
Assume that the potential $U$ is defined for any $x \in \msx$ by $U(x)=(1+\vert x\vert^{2})^{\beta}$ with  $\beta\geq1$. Then $U$ is strongly convex and 
there exists $c_{2}>0$, dependent on $\beta$ only, such that \eqref{eq:laplacian bound} is satisfied with $\varpi = 1-1/\beta$.
\end{lemma}

\begin{proof}
First, we have that 
\begin{equation}
  \label{eq:nabla_U_second_form_pot_ex}
\nabla_{x}U(x)  =2\beta(1+\vert x\vert^{2})^{\beta-1}x=2\beta U(x){}^{1-1/\beta}x,
\end{equation}
and
\begin{equation}
    \label{eq:nabla_2_U_second_form_pot_ex}
  \nabla_{x}^{2}U(x)=2\beta\left[(1-1/\beta)U^{-1/\beta}(x) \nabla_{x}U(x) x^{\top} +U^{1-1/\beta}(x) \Idd\right] \eqsp. 
\end{equation}
As a result, and since $\beta\geq1$,
\begin{equation}
(1-\beta^{-1})U^{-1/\beta}(x)\nabla_{x}U(x) x^\top  =2\beta U(x){}^{1-2/\beta}x x^\top\succeq0  \eqsp,\,  U^{1-1/\beta}(x)I  \succeq \Idd \eqsp,
\end{equation}
from which we conclude that for any $x\in\mathsf{X}$, $\nabla_{x}^{2}U(x)\succeq2\beta  \Idd$. It remains to show that \eqref{eq:laplacian bound}  holds.
First we have for any $x \in \msx$,
\begin{align}
\trace\left(\nabla_{x}^{2}U(x)\right) & =2(\beta-1) U^{-1/\beta} (x) x^\top \nax U(x)+2 \beta d\,U^{1-1/\beta}(x)\\
   & \leq  2(\beta-1)\vert\nabla_{x}U(x)\vert\frac{\vert x\vert}{1+\vert x\vert^{2}} +2 \beta d\,U^{1-1/\beta}(x) \eqsp.
\end{align}
Using that for any $s \geq 0$ and $a >0$, $2s \leq a^{-2}+(as)^2$, $(1+s^2)^{\beta-1} \leq (1+(2d / \beta)^{1/\beta})^{\beta-1} \1_{\ccintLigne{0,(2d / \beta)^{1/\beta}}}(s^2) + (2d / \beta)^{-1} s^{2 \beta }(1+s^2)^{\beta-1} \1_{\oointLigne{(2d / \beta)^{1/\beta},\plusinfty}}(s^2) \leq  (1+(2d / \beta)^{1/\beta})^{\beta-1} \1_{\ccintLigne{0,(2d / \beta)^{1/\beta}}}(s^2) + (2d / \beta)^{-1} s^{2}(1+s^2)^{2\beta-2} \1_{\oointLigne{(2d / \beta)^{1/\beta},\plusinfty}}(s^2)$,
\eqref{eq:nabla_U_second_form_pot_ex}-\eqref{eq:nabla_2_U_second_form_pot_ex}, we get for any $x \in \msx$, 
\begin{align}
  \trace\left(\nabla_{x}^{2}U(x)\right) & \leq 2(\beta-1) \vert\nabla_{x}U(x)\vert \vert x\vert(1+\vert x\vert^{2})^{-1} +2 \beta d\,U^{1-1/\beta}(x)\\
                                        &  \leq (\beta-1) (4\beta +\abs{\nabla_x U(x)}^2/(4\beta)) +2 \beta d [(1+(2d/\beta)^{1/\beta})^{\beta-1} \\
  & \qquad \qquad \qquad + \abs{\nabla_x U(x)}^2/(8 d \beta)] \\
  & \leq 4(\beta-1) \beta  +2^{\beta-1}\beta  d(1+(2d/\beta)^{1-1/\beta}) + \abs{\nabla_x U(x)}^2/2 \eqsp,
\end{align}
where we used in the last step which completes the proof,  that  $(a +b)^{\beta-1} \leq 2^{\beta-2}(a^{\beta-1}+ b^{\beta-1})$ for any $a,b \geq 0$, applying  Hölder inequality, since $\beta \geq 1$. 
\end{proof}

\section*{Acknowledgments}

JR would like to thank Pierre Monmarch\'e for showing him how  ZZ and BPS fall under a general framework. CA acknowledges support from EPSRC ``Intractable Likelihood: New Challenges from Modern Applications (ILike)'' (EP/K014463/1). All the authors acknowledge the support of the Institute for Statistical Science in Bristol. AD acknowledges support from the Chaire BayeScale ``P. Laffitte''.

\bibliographystyle{abbrv}
\bibliography{../Bibliography/bibliography}

\begin{thebibliography}{10}

\bibitem{achleitner:arnold:carlen:2016}
F.~Achleitner, A.~Arnold, and E.~A. Carlen.
\newblock On linear hypocoercive {BGK} models.
\newblock In {\em From particle systems to partial differential equations.
  {III}}, volume 162 of {\em Springer Proc. Math. Stat.}, pages 1--37.
  Springer, [Cham], 2016.

\bibitem{andrieu:livingstone:2018}
C.~Andrieu and S.~Livingstone.
\newblock {Peskun-Tierney} ordering for $\big(\mu,{Q}\big)-$self-adjoint
  {Markov} chain and process {Monte Carlo}.
\newblock 2018.

\bibitem{Bakry08}
D.~Bakry, F.~Barthe, P.~Cattiaux, and A.~Guillin.
\newblock A simple proof of the {P}oincar\'e inequality for a large class of
  probability measures including the log-concave case.
\newblock {\em Elect. Comm. in Probab.}, 13:60--66, 2008.

\bibitem{bakry:gentil:ledoux:2014}
D.~Bakry, I.~Gentil, and M.~Ledoux.
\newblock {\em Analysis and geometry of {M}arkov diffusion operators}, volume
  348 of {\em Grundlehren der Mathematischen Wissenschaften [Fundamental
  Principles of Mathematical Sciences]}.
\newblock Springer, Cham, 2014.

\bibitem{bhatnagar:gross:krook:1954}
P.~L. Bhatnagar, E.~P. Gross, and M.~Krook.
\newblock A model for collision processes in gases. i. small amplitude
  processes in charged and neutral one-component systems.
\newblock {\em Phys. Rev.}, 94:511--525, May 1954.

\bibitem{Bierkens16}
J.~Bierkens, P.~Fearnhead, and G.~Roberts.
\newblock The zig-zag process and super-efficient sampling for {B}ayesian
  analysis of big data.
\newblock {\em arXiv:1607.03188}, 2016.

\bibitem{bierkens2018arXiv180711358B}
J.~{Bierkens}, K.~{Kamatani}, and G.~O. {Roberts}.
\newblock {High-dimensional scaling limits of piecewise deterministic sampling
  algorithms}.
\newblock {\em ArXiv e-prints}, July 2018.

\bibitem{BRZ2018}
J.~Bierkens, G.~Roberts, and P.-A. Zitt.
\newblock Ergodicity of the zigzag process.
\newblock {\em arXiv:1712.09875}, 2018.

\bibitem{Bobkov2003}
S.~G. Bobkov.
\newblock {\em Spectral Gap and Concentration for Some Spherically Symmetric
  Probability Measures}, pages 37--43.
\newblock Springer Berlin Heidelberg, Berlin, Heidelberg, 2003.

\bibitem{BONNEFONT20162456}
M.~Bonnefont, A.~Joulin, and Y.~Ma.
\newblock Spectral gap for spherically symmetric log-concave probability
  measures, and beyond.
\newblock {\em Journal of Functional Analysis}, 270(7):2456 -- 2482, 2016.

\bibitem{BouRabee17}
N.~Bou-Rabee and J.~M.~a. Sanz-Serna.
\newblock Randomized {H}amiltonian {M}onte {C}arlo.
\newblock {\em Ann. Appl. Probab.}, 27(4):2159--2194, 2017.

\bibitem{BouchardCote15}
A.~{Bouchard-C{\^o}t{\'e}}, S.~J. {Vollmer}, and A.~{Doucet}.
\newblock {The Bouncy Particle Sampler: a non-reversible rejection-free Markov
  Chain Monte Carlo method}.
\newblock {\em ArXiv e-prints}, 2015.

\bibitem{2017arXiv170806180B}
E.~{Bouin}, J.~{Dolbeault}, S.~{Mischler}, C.~{Mouhot}, and C.~{Schmeiser}.
\newblock {Hypocoercivity without confinement}.
\newblock {\em ArXiv e-prints}, Aug. 2017.

\bibitem{bouin:hoffman:mouhot:2017}
E.~Bouin, F.~Hoffmann, and C.~Mouhot.
\newblock Exponential decay to equilibrium for a fiber lay-down process on a
  moving conveyor belt.
\newblock {\em SIAM J. Math. Anal.}, 49(4):3233--3251, 2017.

\bibitem{brosse:durmus:moulines:sabanis:2019}
N.~Brosse, A.~Durmus, E.~Moulines, and S.~Sabanis.
\newblock The tamed unadjusted langevin algorithm.
\newblock {\em Stochastic Processes and their Applications}, 2018.

\bibitem{davies:1995}
E.~B. Davies.
\newblock {\em Spectral theory and differential operators}, volume~42 of {\em
  Cambridge Studies in Advanced Mathematics}.
\newblock Cambridge University Press, Cambridge, 1995.

\bibitem{davis:1993}
M.~Davis.
\newblock {\em Markov Models \& Optimization}, volume~49.
\newblock CRC Press, 1993.

\bibitem{Davis1984}
M.~H.~A. Davis.
\newblock Piecewise-deterministic {M}arkov processes: a general class of
  nondiffusion stochastic models.
\newblock {\em J. Roy. Statist. Soc. Ser. B}, 46(3):353--388, 1984.
\newblock With discussion.

\bibitem{deligiannidis2017arXiv170504579D}
G.~{Deligiannidis}, A.~{Bouchard-C{\^o}t{\'e}}, and A.~{Doucet}.
\newblock {Exponential Ergodicity of the Bouncy Particle Sampler}.
\newblock {\em ArXiv e-prints}, May 2017.

\bibitem{deligiannidis2018randomized}
G.~Deligiannidis, D.~Paulin, and A.~Doucet.
\newblock Randomized hamiltonian monte carlo as scaling limit of the bouncy
  particle sampler and dimension-free convergence rates.
\newblock {\em arXiv preprint arXiv:1808.04299}, 2018.

\bibitem{Dolbeault09}
J.~Dolbeault, C.~Mouhot, and C.~Schmeiser.
\newblock Hypocoercivity for kinetic equations with linear relaxation terms.
\newblock {\em C. R. Math. Acad. Sci. Paris}, 347(9-10):511--516, 2009.

\bibitem{Dolbeault15}
J.~{Dolbeault}, C.~{Mouhot}, and C.~{Schmeiser}.
\newblock Hypocoercivity for linear kinetic equations conserving mass.
\newblock {\em Trans. AMS}, 367:3807--3828, 2015.

\bibitem{douc-moulines-priouret-soulier-2019}
R.~Douc, Moulines, P.~\'Eric, Priouret, and P.~Soulier.
\newblock {\em Markov chains}.
\newblock Springer International Publishing, 2019.

\bibitem{Duane87}
S.~Duane, A.~Kennedy, B.~J. Pendleton, and D.~Roweth.
\newblock {Hybrid Monte Carlo}.
\newblock {\em Physics Letters B}, 195(2):216 -- 222, 1987.

\bibitem{durmus-geometric-2018arXiv180705401D}
A.~{Durmus}, A.~{Guillin}, and P.~{Monmarch{\'e}}.
\newblock {Geometric ergodicity of the bouncy particle sampler}.
\newblock {\em ArXiv e-prints}, July 2018.

\bibitem{durmus-invariant-2018arXiv180705421D}
A.~{Durmus}, A.~{Guillin}, and P.~{Monmarch{\'e}}.
\newblock {Piecewise Deterministic Markov Processes and their invariant
  measure}.
\newblock {\em ArXiv e-prints}, July 2018.

\bibitem{eckman:hairer:2003}
J.-P. Eckmann and M.~Hairer.
\newblock Spectral properties of hypoelliptic operators.
\newblock {\em Comm. Math. Phys.}, 235(2):233--253, 2003.

\bibitem{ethier:kurtz:1986}
S.~N. Ethier and T.~G. Kurtz.
\newblock {\em {M}arkov processes}.
\newblock Wiley Series in Probability and Mathematical Statistics: Probability
  and Mathematical Statistics. John Wiley \& Sons Inc., New York, 1986.
\newblock Characterization and convergence.

\bibitem{evans2017hypocoercivity}
J.~Evans.
\newblock Hypocoercivity in phi-entropy for the linear relaxation boltzmann
  equation on the torus.
\newblock {\em arXiv preprint arXiv:1702.04168}, 2017.

\bibitem{Faggionato2009}
A.~Faggionato, D.~Gabrielli, and M.~Ribezzi~Crivellari.
\newblock Non-equilibrium thermodynamics of piecewise deterministic markov
  processes.
\newblock {\em Journal of Statistical Physics}, 137(2):259, Oct 2009.

\bibitem{GelmanCaStDuVeRu2014}
A.~Gelman, J.~B. Carlin, H.~S. Stern, D.~B. Dunson, A.~Vehtari, and D.~B.
  Rubin.
\newblock {\em Bayesian data analysis}.
\newblock Texts in Statistical Science Series. CRC Press, Boca Raton, FL, third
  edition, 2014.

\bibitem{grothaus2020hypocoercivity}
M.~Grothaus and M.~Mertin.
\newblock Hypocoercivity of langevin-type dynamics on abstract smooth
  manifolds, 2020.

\bibitem{GS2014}
M.~Grothaus and P.~Stilgenbauer.
\newblock Hypocoercivity for {K}olmogorov backward evolution equations and
  applications.
\newblock {\em J. Funct. Anal.}, 267(10):3515--3556, 2014.

\bibitem{GS2015}
M.~Grothaus and P.~Stilgenbauer.
\newblock A hypocoercivity related ergodicity method for singularly distorted
  non-symmetric diffusions.
\newblock {\em Integral Equations Operator Theory}, 83(3):331--379, 2015.

\bibitem{GS2016}
M.~Grothaus and P.~Stilgenbauer.
\newblock Hilbert space hypocoercivity for the {L}angevin dynamics revisited.
\newblock {\em Methods Funct. Anal. Topology}, 22(2):152--168, 2016.

\bibitem{grothaus2017weak}
M.~Grothaus and F.-Y. Wang.
\newblock Weak poincarée inequalities for convergence rate of degenerate
  diffusion processes.
\newblock {\em arXiv preprint arXiv:1703.04821}, 2017.

\bibitem{hankwan:leautaud:2015}
D.~Han-Kwan and M.~L\'{e}autaud.
\newblock Geometric analysis of the linear {B}oltzmann equation {I}. {T}rend to
  equilibrium.
\newblock {\em Ann. PDE}, 1(1):Art. 3, 84, 2015.

\bibitem{herau:2006}
F.~H\'{e}rau.
\newblock Hypocoercivity and exponential time decay for the linear
  inhomogeneous relaxation {B}oltzmann equation.
\newblock {\em Asymptot. Anal.}, 46(3-4):349--359, 2006.

\bibitem{herau:nier:2004}
F.~H\'{e}rau and F.~Nier.
\newblock Isotropic hypoellipticity and trend to equilibrium for the
  {F}okker-{P}lanck equation with a high-degree potential.
\newblock {\em Arch. Ration. Mech. Anal.}, 171(2):151--218, 2004.

\bibitem{Holley87}
R.~Holley and D.~Stroock.
\newblock Logarithmic {S}obolev inequalities and stochastic {I}sing models.
\newblock {\em J. Statist. Phys.}, 46(5-6):1159--1194, 1987.

\bibitem{hormander:1967}
L.~H\"{o}rmander.
\newblock Hypoelliptic second order differential equations.
\newblock {\em Acta Math.}, 119:147--171, 1967.

\bibitem{yoshida:1980}
K.Yoshida.
\newblock {\em Functional analysis}.
\newblock Grundlehren der mathematischen Wissenschaften in Einzeldarstellungen
  mit besonderer Ber{\"u}cksichtigung der Anwendungsgebiete, Bd. 123.
  Springer-Verlag, 6ed. edition, 1980.

\bibitem{liu2008monte}
J.~S. Liu.
\newblock {\em Monte Carlo strategies in scientific computing}.
\newblock Springer Science \& Business Media, 2008.

\bibitem{meyn2012markov}
S.~P. Meyn and R.~L. Tweedie.
\newblock {\em Markov chains and stochastic stability}.
\newblock Springer Science \& Business Media, 2012.

\bibitem{Michel14}
M.~Michel, S.~C. Kapfer, and W.~Krauth.
\newblock {Generalized event-chain Monte Carlo: Constructing rejection-free
  global-balance algorithms from infinitesimal steps}.
\newblock {\em J. Chem. Phys.}, 140(5):054116, 2014.

\bibitem{Michel17}
M.~Michel and S.~S{\'e}n{\'e}cal.
\newblock {Forward Event-Chain Monte Carlo: a general rejection-free and
  irreversible Markov chain simulation method}.
\newblock {\em arXiv preprint arXiv:1702.08397}, 2017.

\bibitem{monmarche2017note}
P.~Monmarch{\'e}.
\newblock A note on fisher information hypocoercive decay for the linear
  boltzmann equation.
\newblock {\em arXiv preprint arXiv:1703.10504}, 2017.

\bibitem{mouhot:nuemann:2006}
C.~Mouhot and L.~Neumann.
\newblock Quantitative perturbative study of convergence to equilibrium for
  collisional kinetic models in the torus.
\newblock {\em Nonlinearity}, 19(4):969--998, 2006.

\bibitem{odonnell:2014}
R.~O'Donnell.
\newblock {\em Analysis of {B}oolean functions}.
\newblock Cambridge University Press, New York, 2014.

\bibitem{pardoux2001}
E.~Pardoux and Y.~Veretennikov.
\newblock On the {P}oisson equation and diffusion approximation. i.
\newblock {\em Ann. Probab.}, 29(3):1061--1085, 07 2001.

\bibitem{pedersen1995analysis}
G.~K. Pedersen.
\newblock {\em Analysis now}, volume 118.
\newblock Springer Science \& Business Media, 1995.

\bibitem{persson1960bounds}
A.~Persson.
\newblock Bounds for the discrete part of the spectrum of a semi-bounded
  schr{\"o}dinger operator.
\newblock {\em Mathematica Scandinavica}, 8(1):143--153, 1960.

\bibitem{peters2012rejection}
E.~A. J.~F. Peters and G.~de~With.
\newblock Rejection-free monte carlo sampling for general potentials.
\newblock {\em Phys. Rev. E}, 85:026703, Feb 2012.

\bibitem{Redon16}
S.~Redon, G.~Stoltz, and Z.~Trstanova.
\newblock Error analysis of modified {Langevin} dynamics.
\newblock {\em J. Stat. Phys.}, 164(4):735--771, 2016.

\bibitem{reed1972methods}
M.~Reed and B.~Simon.
\newblock {\em Methods of Modern Mathematical Physics: Functional
  Analysis.-1972.-(RU-idnr: M103448034)}.
\newblock Academic Press, 1972.

\bibitem{robert2013monte}
C.~Robert and G.~Casella.
\newblock {\em Monte Carlo Statistical Methods}.
\newblock Springer Science \& Business Media, 2013.

\bibitem{Vanetti17}
P.~Vanetti, A.~Bouchard-C{\^o}t{\'e}, G.~Deligiannidis, and A.~Doucet.
\newblock {Piecewise Deterministic Markov Chain Monte Carlo}.
\newblock {\em arXiv preprint arXiv:1707.05296}, 2017.

\bibitem{villani2006hypocoercive}
C.~Villani.
\newblock Hypocoercive diffusion operators.
\newblock In {\em International Congress of Mathematicians}, volume~3, pages
  473--498, 2006.

\bibitem{Villani09}
C.~Villani.
\newblock Hypocoercivity.
\newblock {\em Mem. Amer. Math. Soc.}, 202(950), 2009.

\bibitem{Wu17}
C.~Wu and C.~P. Robert.
\newblock Generalized bouncy particle sampler.
\newblock {\em arXiv preprint arXiv:1706.04781}, 2017.

\end{thebibliography}

\end{document}